%% file: main_arxiv.tex
\documentclass[10pt, conference, onecolumn]{IEEEtran}[onecolumn]

\usepackage{import}
\usepackage{preamble}
\makeatletter
\let\c@proposition\c@theorem
\let\c@corollary\c@theorem
\let\c@lemma\c@theorem
\let\c@definition\c@theorem
\let\c@example\c@theorem
\let\c@remark\c@theorem
\makeatother
\numberwithin{proposition}{section}
\numberwithin{corollary}{section}
\numberwithin{lemma}{section}
\numberwithin{theorem}{section}
\numberwithin{definition}{section}
\numberwithin{example}{section}
\numberwithin{remark}{section}
\begin{document}

\title{Session Logical Relations for Noninterference }


\author{\IEEEauthorblockN{Farzaneh Derakhshan}
 \IEEEauthorblockA{Carnegie Mellon University}
 \and
 \IEEEauthorblockN{Stephanie Balzer}
 \IEEEauthorblockA{Carnegie Mellon University}
 \and
 \IEEEauthorblockN{Limin Jia}
 \IEEEauthorblockA{Carnegie Mellon University}}
\maketitle

\thispagestyle{plain}
\pagestyle{plain}


\begin{abstract}

\input{abstract}
\end{abstract}

\IEEEpeerreviewmaketitle

\input{intro}
\input{mot_example}
\input{key-ideas}
\input{type-system}
\input{key-ideas2}
\input{logicalrel}

\input{metatheory}

\input{related}
\clearpage
\appendix
\input{appendix}









\bibliographystyle{IEEEtran}
\bibliography{references}

\end{document}

%% file: abstract.tex
Information flow control type systems statically restrict the propagation of sensitive data to
ensure end-to-end confidentiality.  The property to be shown is noninterference,
asserting that an attacker cannot infer any secrets from made observations.  Session types
delimit the kinds of observations that can be made along a communication channel by imposing a
protocol of message exchange.  These protocols govern the exchange along a single channel and leave unconstrained the propagation along adjacent channels.  This paper contributes an
information flow control type system for linear session types.  The type system stands in close
correspondence with intuitionistic linear logic.  Intuitionistic linear logic typing
ensures that process configurations form a tree such that client processes
are parent nodes and provider processes child nodes.  To control the propagation of secret
messages, the type system is enriched with secrecy levels and arranges these levels to be
aligned with the configuration tree.  Two levels are associated with every process: the maximal
secrecy denoting the process' security clearance and the running secrecy denoting the highest
level of secret information obtained so far.  The computational semantics naturally stratifies
process configurations such that higher-secrecy processes are parents of lower-secrecy ones, an invariant enforced by typing.  Noninterference is stated in terms of
a logical relation that is indexed by the secrecy-level-enriched session types.  The logical
relation contributes a novel development of logical relations for session typed languages as it
considers open configurations, allowing for more nuanced equivalence statement. 

%% file: intro.tex
\section{Introduction}

\emph{Message-passing} is a successful concurrency paradigm, adopted by languages such as
Erlang, Go, and Rust.  In this setting, a program amounts to a number of processes connected
via channels, and computation happens by the concurrent exchange of messages along channels.
To prescribe the \emph{protocols} of message exchange and assert their adherence at run-time,
\emph{session types}~\cite{HondaCONCUR1993, HondaESOP1998} were introduced.  Since then various
session-typed programming languages were designed~\cite{ToninhoESOP2013,ToninhoPhD2015,BalzerICFP2017} as well
as session type libraries for mainstream
languages developed~\cite{DezaniECOOP2006,PucellaHASKELL2008,ImaiPLACES2010,JespersenWGP2015,LindleyHASKELL2016,ScalasECOOP2016,PadovaniARTICLE2017,ImaiARTICLE2019,KokkeICE2019,
  ChenARXIV2020}.  Session types moreover have a logical foundation by a Curry-Howard
correspondence between \emph{linear logic} and the session-typed
$\pi$-calculus~\cite{CairesCONCUR2010,WadlerICFP2012,KokkePOPL2019}.

In addition to session fidelity enforced by session types, preventing information leakage is another desirable goal of such systems, as OS processes, android apps, and web applications can all be modeled in them. One promising direction to achieve this goal is to develop an information flow type system to enforce information flow control (IFC) to enforce a {\em noninterference} property, guaranteeing that an adversary cannot infer any secrets from observing message exchanges~\cite{volpano96, ifc-survey}. 
While prior work has investigated both information flow type systems for  process
calculi~\cite{HondaESOP2000,HondaYoshidaPOPL2002,CrafaARTICLE2002,CrafaTGC2005,CrafaFMSE2006,Crafa2007,HENNESSYRIELY2002,HENNESSY20053,KOBAYASHI2005,ZDANCEWIC2003,POTTIER2002}
and run-time monitoring in the application domain of OS, android apps, and web
applications~\cite{cowl,ifc-browser-ndss15,android-esorics13,krohn:flume}, very few information
flow session type systems exist~\cite{CapecchiCONCUR2010,CapecchiARTICLE2014}. In particular,
no one has investigated information flow types in the context of linear binary session types based on a sequent calculus of intuitionistic linear logic, which is a natural fit for a more flexible {\em flow-sensitive} information flow type system.



This paper develops a flow-sensitive information flow session type system for the language
$\lang$ and proves noninterference for $\lang$ in addition to type safety.  $\lang$ is a
terminating language with higher-order channels, allowing channels to be sent along channels.
It builds on the Curry-Howard correspondence between \emph{intuitionistic} linear logic and the
session-typed $\pi$-calculus~\cite{CairesCONCUR2010,CairesARTICLE2016}.  The intuitionistic
foundation turns run-time configurations of processes into \emph{trees}, connecting a
\emph{providing} process with \emph{exactly one} \emph{client}.

The $\lang$ type system takes advantage of the tree structure imposed by intuitionism and
stratifies process trees according to the security order.  Two secrecy levels are associated
with each process: the \emph{maximal secrecy}, denoting the maximal level of information the
process may ever obtain, and the \emph{running secrecy}, denoting the highest level of
information a process has obtained so far and whose changes are tracked by the type system.  To align the process tree with the security
lattice, typing asserts the following invariant, for any node in the tree: \textit{(i)} the
maximal secrecy of a child node is at most as high as the maximal secrecy of the parent node
and \textit{(ii)} the running secrecy of the parent node is capped by its maximal secrecy.  By
complementing the maximal secrecy with a running secrecy, the $\lang$ type system becomes
\emph{flow-sensitive}, allowing more secure programs to successfully type check than would be
possible with maximal secrecy alone.

Noninterference of $\lang$ is stated in terms of a \emph{logical
  relation}~\cite{TaitARTICLE1967,StatmanARTICLE1985}.  The use of logical relations for
session types has focused predominantly on unary logical relations (predicates) for proving
termination~\cite{PerezESOP2012, PerezARTICLE2014,DeYoungFSCD2020} with the exception of a
binary logical relation for parametricity~\cite{CairesESOP2013}.  Noninterference, however,
demands a more nuanced binary relation, requiring communication to be perceived in either
direction of the channel.  This paper generalizes binary logical relations for session
typed languages to support \emph{open configurations}, considering both the antecedent and
succedent of the typing judgment.


In summary, the paper makes the following contributions:

\begin{itemize}

\item development of an flow-sensitive IFC type system for binary session types, yielding the
  language $\lang$;

\item proofs of type safety and noninterference of $\lang$;

\item generalization of (binary) logical relations to the session typed setting, supporting
  open configurations and higher-order channels.

\end{itemize}

\textbf{Paper structure:} \secref{sec:motexample} familiarizes the reader with information flow
control and intuitionistic session-typed programming.  \secref{sec:key-ideas} develops the main
ideas underlying the $\lang$ type system, which are concretized in \secref{sec:type-system}.
\secref{sec:key-ideas2} develops the main ideas underlying the session logical relation,
further detailed in \secref{sec:logicalrel}.  \secref{sec:metatheory} proves noninterference of
$\lang$ as well as type safety.  \secref{sec:related} summarizes related and future work.  Further technical developments and proofs can be
found in the appendix.

%% file: mot_example.tex
\section{Motivating Example}\label{sec:motexample}

This section provides an introduction to programming with intuitionistic linear logic session
types~\cite{ToninhoESOP2013,ToninhoPhD2015,CairesARTICLE2016,BalzerICFP2017} based on a banking
example and illustrates violations of end-to-end confidentiality.  We base the discussion on
the language $\lang$ that we formalize and for which we prove noninterference in the remainder
of this paper.  $\lang$ is a terminating language with higher-order channels, allowing channels
to be sent over channels.

In $\lang$, we can define the protocol according to which an authorization process interacts
with a customer seeking access to their bank account as follows:

\begin{center}
\begin{small}
\begin{minipage}[t]{\textwidth}
\begin{tabbing}
$\m{auth}$ $=$ \= $\extchoice \{$\= $\mathit{tok_1}{:} \oplus \{ \mathit{succ}{:} \,\m{account} \otimes 1, \mathit{fail}{:}\, 1\}, \dots, $ \\
\> \> $\mathit{tok_n}{:} \oplus \{ \mathit{succ}{:} \,\m{account} \otimes 1,  \mathit{fail}{:}\, 1\}\}$
\end{tabbing}
\end{minipage}
\end{small}
\end{center}

\noindent The connectives $\extchoice$, $\oplus$, $\otimes$, and $1$ can be found in
\tabref{tab:session_types}, providing an overview of intuitionistic linear session types and
their operational reading.
The first column indicates the session type before the message exchange, the second column the
session type after the exchange.  The corresponding process terms are listed in the third and
fourth column, respectively.  The fifth column provides the operational meaning of a connective
and the last column its \emph{polarity}.  Positive connectives have a sending semantics,
negative connectives a receiving semantics.

\begin{table*}
\centering
\begin{small}
\begin{tabular}{@{}lllllc@{}}
\toprule
\multicolumn{2}{@{}l}{\textbf{Session type (current / cont)}} &
\multicolumn{2}{l}{\textbf{Process term (current / cont)}} &
\textbf{Description} &
\textbf{Pol} \\[3pt]
$x : \oplus\{\ell:A\}_{\ell \in L}$ & $x : A_k$ & $x.k; P$ &
$P$ & provider sends label $k$ along $x$ and continues with $P$ & + \\
 & & $\mathbf{case}\, x (\ell \Rightarrow Q_{\ell})_{\ell \in L}$ & $Q_k$ &
client receives label $k$ along $x$ and continues with $Q_k$ & \\[2pt]
$x : \&\{\ell:A\}_{\ell \in L}$ & $x : A_k$ & $\mathbf{case}\, x (\ell \Rightarrow P_{\ell})_{\ell \in L}$ &
$P_k$ & provider receives label $k$ along $x$ and continues with $P_k$  & - \\
 & & $x.k\; Q$ & $Q$ & client sends label $k$ along $x$ and continues with $Q$  & \\[2pt]
$x : A \otimes B$ & $x : B$ & $\mathbf{send}\,y\,x; P$ &
$P$ & provider sends channel $y {:} A$ along $x$ and continues with $P$  & + \\
 & & $z \leftarrow \mathbf{recv}\,x;Q$ & $\subst{y}{z}{Q_z}$ &
client receives channel $y {:} A$ along $x$ and continues with $Q$  & \\[2pt]
$x : A \multimap B$ & $x : B$ & $z \leftarrow \mathbf{recv}\,x;P$ & $\subst{y}{z}{P_z}$ &
provider receives channel $y {:} A$ along $x$ and continues with $P$  & - \\
 & & $\mathbf{send}\,y\,x; Q$ & $Q$ &
client sends channel $y {:} A$ along $x$ and continues with $Q$  & \\[2pt]
$x : 1$ & - & $\mathbf{close}\, x$ &
- & provider sends ``$\m{end}$'' along $x$ and terminates  & + \\
 & & $\mathbf{wait}\,x;Q$ & $Q$ & client receives ``$\m{end}$'' along $x$ and continues with Q  & \\
\bottomrule
\end{tabular}
\caption{Overview of intuitionistic linear session types in $\lang$ together with their operational meaning.}
\label{tab:session_types}
\end{small}
\end{table*}
 
 Linearity ensures that a channel connects exactly two processes.  An intuitionistic viewpoint
moreover allows the distinction of one process as the \emph{provider} and the other as the
\emph{client}, where linearity ensures that every providing process has \emph{exactly one}
client process.  As a result, channels in intuitionistic linear session type languages can
be typed with the session type of the providing process.
In developments of linear session types based on classical logic~\cite{WadlerICFP2012}, the two
endpoints of a channel are instead typed separately, using linear negation to make sure that
the two endpoint types are dual to each other.
The fact that a provider process and client process must behave dually to each other surfaces
in an intuitionistic setting at the level of the process terms, which come in matching pairs.
\tabref{tab:session_types} lists the process term of a provider in the first line for each
connective and the client's term in the second line.

The above session type $\m{auth}$ thus requires the client to send their authorization token
($\mathit{tok_i}$), after which the authorization process will respond with $\mathit{succ}$ in
case of successful authorization and $\mathit{fail}$, otherwise.  In the former case, the
authorization process sends the channel to the customer's bank account and then terminates, in
the latter case it just terminates.  A corresponding authorization process is implemented for
each customer, accepting only the customer's authorization token.  We assume that session type
$\m{auth}$ includes a label $\mathit{tok_i}$ for every imaginable authorization token.

We complete the example with the addition of the following session types:

\begin{center}
\begin{small}
\begin{minipage}[t]{\textwidth}
\begin{tabbing}
$\m{customer}$ \= $=$ \= $\m{auth} \multimap 1$ \\[2pt]
$\m{account}$ \> $=$ \> $\oplus\{\mathit{high}{:} 1,\,\mathit{med}{:}1,\, \mathit{low}{:}1\}$ \\[2pt]
$\m{rate}$ \> $=$ \> $\&\{\mathit{lowRate}{:} 1,\,\mathit{highRate}{:}1\}$
\end{tabbing}
\end{minipage}
\end{small}
\end{center}

\noindent As the names suggest, $\m{customer}$ denotes the protocol of a customer process,
indicating that it is waiting to receive an authorization channel, after which it eventually
terminates.  A bank account process (session type $\m{account}$), on the other hand, will
indicate whether its balance is high ($\mi{high}$), medium ($\mi{med}$), or low ($\mi{low}$)
and then terminate.  The last session type $\m{rate}$ allows a bank to advertise the current
interest rate, for example by displaying it on a bulletin board.

For our example, we assume that the bank has two customers, Alice and Bob, which own accounts
with the bank.  In a secure system, Alice's account can only be queried by Alice or the bank,
but neither by Bob or any walk-in customer.  The same must hold for Bob's account.  We can
express these dependencies by defining corresponding secrecy levels and a lattice on them:
{\small\[
\mb{guest} \sqsubseteq \mb{alice} \sqsubseteq \mb{bank} \qquad \mb{guest} \sqsubseteq \mb{bob} \sqsubseteq \mb{bank}
\]}
We next show the corresponding process implementations concerning Alice.  We first define
process $\m{Alice}$ for the alice customer process:

\begin{center}
\begin{small}
\begin{minipage}[t]{\textwidth}
\begin{tabbing}
$\cdot \vdash \m{Alice} :: y{:}\,\m{customer}\maxsec{[\mb{alice}]}$ \\
$y \leftarrow \m{Alice}\leftarrow \cdot = ($
 \comment{// $\cdot \vdash y{:}\m{customer}$} \\
\quad \= $w \leftarrow \mb{recv}\, y; w.\mathit{tok}_j;$ 
 \comment{// $w{:}\oplus \{ \mathit{succ}{:} \,\m{account} \otimes 1, \mathit{fail}{:}\, 1\} \vdash y{:}1$} \\
\> $\mb{case} \, w \, ($\= $\mathit{succ} \Rightarrow v \leftarrow \mb{recv}\, w;$
 \comment{// $w{:}1, v{:}\m{account} \vdash y{:}1$} \\
\> \> $\mb{case} \, v \, ($\= $\mathit{high} \Rightarrow \mb{wait}\,v;\mb{wait}\,w; \mb{close}\,y$ \\
\> \> \> $\mid \mathit{med} \Rightarrow \mb{wait}\,v;\mb{wait}\,w; \mb{close}\,y$ \\
\> \> \> $\mid \mathit{low} \Rightarrow \mb{wait}\,v;\mb{wait}\,w; \mb{close}\,y)$ \\
\> \> $\mid \mathit{fail}\Rightarrow \mb{wait}\,w; \mb{close}\,y ))\runsec{@\mb{alice}}$
\end{tabbing}
\end{minipage}
\end{small}
\end{center}

\noindent The first line of the above process definition denotes the process' signature.  It is
in line with the process term typing judgment introduced in \secref{sec:type-system} and
indicates that process $\m{Alice}$ provides a session of type $\m{customer}$ along channel $y$
without being a client of any other sessions (denoted by $\cdot$ on the left of the turnstile).
The next line introduces the bindings of channels variables to be used in the body of the
process, appearing to the right of the $=$ sign.  We generally use the symbol $\leftarrow$
denote variable bindings.  For the time being, we ignore the secrecy annotations
$\maxsec{[\mb{alice}]}$ and $\runsec{@\mb{alice}}$.

In its body, the $\m{Alice}$ process first receives a channel to Alice's authorization process.
Along this channel it then sends Alice's authorization token.  If that token is correct, the
authorization process will respond by sending a channel to Alice's account process.  Otherwise,
the $\m{Alice}$ process waits for the authorization process to terminate and then terminates
itself.  In case of successful authentication, the $\m{Alice}$ process queries its account
process for its balance, willing to receive any of the labels $\mathit{high}$, $\mathit{med}$,
or $\mathit{low}$, and then waits for the authorization and account processes to terminate,
before terminating itself.

A distinguishing feature of session type programming is that channels and the processes
offering along those channels change their types along with the messages exchange.  It is
instructive to walk through the body of process $\m{Alice}$ to follow these state changes,
consulting \tabref{tab:session_types} as needed.  We include annotations as comments,
indicating the types of all channels existing at the various points in the code\footnote{We have
omitted secrecy annotations for compactness.}.

Next, we show the implementation of Alice's authorization process $\m{aAuth}$.  This process
offers a session of type $\m{auth}$ along its offering channel $x$ and uses a process along
channel $u$, which offers a choice between access to Alice's account process (label $\mi{s}$)
or a terminating process (label $\mi{f}$).  The $\m{aAuth}$ process waits to receive an
authorization token along its offering channel.  If the sent token is Alice's authorization
token ($\mathit{tok}_j$), the authorization process sends the label $\mi{succ}$ along its
offering channel as well as the label $\mi{s}$ along channel $u$, after which it sends the
channel $u$ providing access to Alice's account process along $x$ and then terminates.
Otherwise, the authorization process sends the labels $\mi{fail}$ and $\mi{f}$ along channel
$x$ and $u$, respectively, waits for $u$ to terminate and then terminates itself.

\begin{center}
\begin{small}
\begin{minipage}[t]{\textwidth}
\begin{tabbing}
$u{:}\,\&\{\mathit{s}{:}\m{account}, \mathit{f}{:}1\}\maxsec{[\mb{alice}]} \vdash \m{aAuth} :: x{:}\m{auth}\maxsec{[\mb{alice}]}$ \\
$x \leftarrow \m{aAuth}\leftarrow u = ($ \\
\quad \= $\mb{case} \, x \, ($\= $\mathit{tok}_j \Rightarrow x.\mathit{succ};u.\mathit{s}; \mb{send}\, u\, x; \mb{close}\,x$ \\
\> \> $\mid \mathit{tok}_{i \neq j} \Rightarrow x.\mathit{fail}; u.\mathit{f}; \mb{wait}\, u; \mb{close}\, x ))\runsec{@\mb{alice}}$
\end{tabbing}
\end{minipage}
\end{small}
\end{center}

The implementation of Alice's account process $\m{aAcc}$ is finally shown below.  We leave it
to the reader to walk through the code, consulting \tabref{tab:session_types} as needed.

\begin{center}
\begin{small}
\begin{minipage}[t]{\textwidth}
\begin{tabbing}
$\cdot \vdash \m{aAcc} :: u{:}\&\{\mathit{s}{:}\m{account}, \mathit{f}{:}1\}\maxsec{[\mb{alice}]}$ \\
$u \leftarrow \m{aAcc}\leftarrow \cdot = ($ \\
\quad \= $\mb{case} \, u \, ($\= $\mathit{s} \Rightarrow u.\mathit{high}; \mb{close}\, u$ \\
\> \> $\mid \mathit{f} \Rightarrow \mb{close}\,u ))\runsec{@\mb{alice}}$
\end{tabbing}
\end{minipage}
\end{small}
\end{center}

It is instructive to look at the implementation of the bank process, which instantiates our
running example.  We assume corresponding process definitions for Bob and the rate
to be displayed on the bulletin board.

\begin{center}
\begin{small}
\begin{minipage}[t]{\textwidth}
\begin{tabbing}
$x{:}\,\m{auth}\maxsec{[\mb{alice}]} , y{:}\,\m{customer}\maxsec{[\mb{alice}]}, x'{:}\m{auth}\maxsec{[\mb{bob}]},$ \\
$y'{:}\,\m{customer}\maxsec{[\mb{bob}]}, u{:}\,\m{rate}\maxsec{[\mb{guest}]} \vdash \m{Bank} :: z{:}1\maxsec{[\mb{bank}]}$ \\
$z \leftarrow \m{Bank}\leftarrow x,x',y,y',u = ($ \\
\quad \= $\mb{send}\, x\,y;\mb{send}\, x'\,y'; u.\mathit{lowRate};$ \\
\>$\mb{wait}\,y ;\mb{wait}\,y'; \mb{wait}\, u; \mb{close}\,z)\runsec{@\mb{guest}}$
\end{tabbing}
\end{minipage}
\end{small}
\end{center}

\noindent \figref{fig:example} shows the run-time configuration of processes that exist before and after
executing the first statement in the above code. Intuitionistic linear typing imposes a
\emph{tree} structure on process configurations such that client processes are parent nodes and
provider processes child nodes.  \figref{fig:example} also demonstrates that message exchanges
may not only change the type of a channel and its offering process but also the structure of
the tree.  Changes in the tree structure, in particular, are due to the connectives $\chanin$
and $\chanout$, which make a sibling subtree the child of the recipient and a child subtree a
sibling of the sender, respectively.

\begin{figure}
\begin{center}
\includegraphics[scale=0.39]{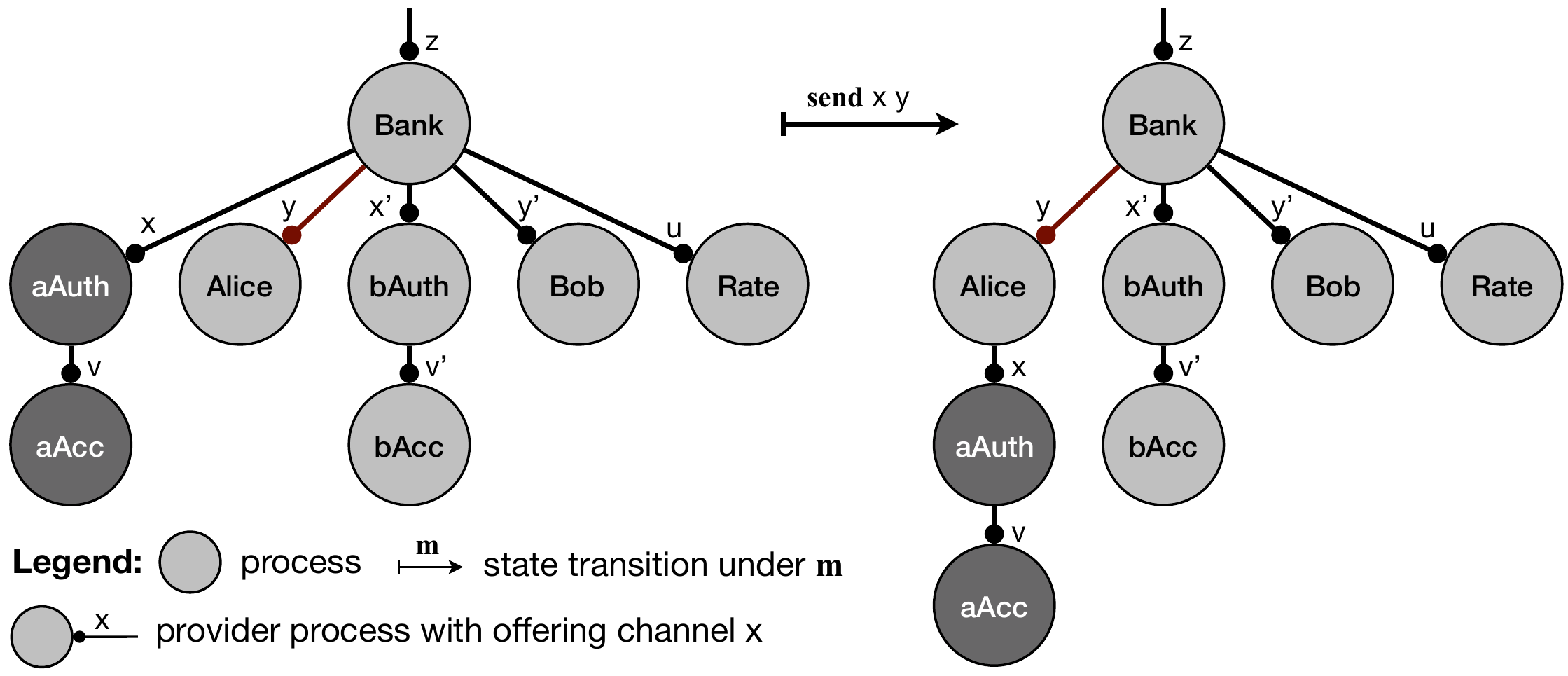}
\caption{State transition in process configuration due to $\multimap$.}
\label{fig:example}
\end{center}
\end{figure}

It is time to ask ourselves whether the $\m{Bank}$ process is actually secure.  For this
purpose we now consider the red secrecy annotations $\maxsec{[\mb{d}]}$.  These annotations
indicate the \emph{maximal secrecy} of a process, \ie the maximal level of secret information
the process may ever obtain.  As to be expected, the processes $\m{Alice}$, $\m{aAuth}$, and
$\m{aAcc}$ have maximal secrecy $\maxsec{[\mb{alice}]}$ because they know Alice's authorization
token and account balance.  Similarly, the processes associated with Bob have maximal secrecy
$\maxsec{[\mb{bob}]}$.  The $\m{Bank}$ process itself has the highest maximal secrecy of
$\maxsec{[\mb{bank}]}$.  The process associated with the rate bulletin board, on the other
hand, has the lowest maximal secrecy $\maxsec{[\mb{guest}]}$ because information about interest
rates are available to any walk-in customer.  Given these annotations and the security lattice
defined earlier, we can conclude that process $\m{Bank}$ is secure: it sends Alice's
authorization process to Alice and Bob's authorization process to Bob, but not other way
around.

Next, let's ask the same question for the below $\m{LeakyBank}$ process implementation.  As
its name suggests, this implementation is not secure.  Information is leaked by sending the
channel to Alice's authorization process to a customer with a maximal secrecy of
$\maxsec{[\mb{guest}]}$, potentially allowing such a customer to get access to Alice's bank
account.

\begin{center}
\begin{small}
\begin{minipage}[t]{\textwidth}
\begin{tabbing}
$x{:}\,\m{auth}\maxsec{[\mb{alice}]} , y{:}\,\m{customer}\maxsec{[\mb{guest}]} \vdash \m{LeakyBank} :: z{:}1\maxsec{[\mb{bank}]}$ \\
$z \leftarrow \m{LeakyBank}\leftarrow x,y = ($ \\
\quad \= $\mb{send}\, x\,y;$
\comment{// insecure send} \\
\> $\mb{wait}\,y; \mb{close}\,z)\runsec{@{\mb{guest}}}$
\end{tabbing}
\end{minipage}
\end{small}
\end{center}

While process $\m{LeakyBank}$ contains what is referred to as a \emph{direct} flow there also
exist \emph{indirect} flows, which are more subtle.  For example, consider the below process
definition $\m{SneakyaAuth}$ that not only authenticates Alice but also indirectly leaks
information about whether Alice's authorization was successful to the adversary $x_1$.

\begin{center}
\begin{small}
\begin{minipage}[t]{\textwidth}
\begin{tabbing}
$x_1{:} \&\{s{:}1, f{:} 1\}\maxsec{[\mb{guest}]},$ \\
$u{:}\,\&\{\mathit{s}{:}\m{account}, \mathit{f}{:}1\}\maxsec{[\mb{alice}]} \vdash \m{SneakyaAuth} :: x{:}\m{auth}\maxsec{[\mb{alice}]}$ \\
$x \leftarrow \m{SneakyaAuth}\leftarrow u, x_1 = ($ \\
\quad \= $\mb{case} \, x \, ($\= $\mathit{tok}_j \Rightarrow$ \=$x.\mathit{succ};u.\mathit{s}; x_1.s;$
\comment{// insecure send} \\
\> \> \>$\mb{send}\, u\, x ; \mb{wait}\, x_1; \mb{close}\,x$ \\
\> \> $\mid \mathit{tok}_{i \neq j} \Rightarrow$ \=$x.\mathit{fail}; u.\mathit{f}; x_1.f;$
\comment{// insecure send} \\
\> \> \> $\mb{wait}\, u; \mb{wait}\, x_1; \mb{close}\, x ))\runsec{@\mb{alice}}$
\end{tabbing}
\end{minipage}
\end{small}
\end{center}

\noindent Process $\m{SneakyaAuth}$ is not secure because the sends to the adversary $x_1$ with
maximal secrecy $\maxsec{[\mb{guest}]}$ happen when branching on channel $x$ whose maximal
secrecy is $\maxsec{[\mb{alice}]}$.

To rule out indirect information flows in $\lang$, we complement the maximal secrecy of a process with its
\emph{running secrecy}, occurring as green process term level annotations $\runsec{@\mb{c}}$.
The running secrecy denotes the highest level of secret information a process has obtained so
far.  When defining a process, a programmer must indicate the process' maximal secrecy as well
as the \emph{initial} running secrecy the process starts out with when spawned.  As we will see
in Sections~\ref{sec:type-system} and \ref{sec:key-ideas}, the $\lang$ type system increases
the running secrecy accordingly whenever information of higher secrecy is received and
disallows sends from contexts of a higher running secrecy than the one of the receiver.

%% file: key-ideas.tex
\section{Key Ideas - Part I}\label{sec:key-ideas}

This section develops the main ideas underlying the $\lang$ type system.

The banking example discussed in the previous section reveals that a process configuration
naturally aligns with the security lattice of the application: processes with higher maximal
secrecy are ancestors (direct or transitive parents) of processes with lower or same maximal secrecy.
For the $\m{Bank}$ configuration shown in \figref{fig:example}, for example, the $\m{Bank}$
process has the top maximal secrecy $\maxsec{[\mb{bank}]}$ and is the root process of the
configuration, whereas all its descendants (direct or transitive children) have a lower maximal
secrecy.

We can impose this property as a \emph{presupposition} on the typing judgment for process terms:
{\small\[
\Psi; \Delta \vdash P \runsec{@c}:: (x{:}A\maxsec{[d]})
\]}
\noindent with presuppositions:

\begin{enumerate}[label=(\roman*)]

\item {\small$\forall y{:}B\maxsec{[d']} \in \Delta\,  (\Psi \Vdash \maxsec{d'} \sqsubseteq \maxsec{d})$}

\item {\small$\Psi \Vdash \runsec{c} \sqsubseteq \maxsec{d}$}

\end{enumerate}

The typing judgment states that process $P$ with maximal secrecy $\maxsec{[d]}$ and running
secrecy $ \runsec{@c}$ provides a session of type $A$ along channel variable $x$, given the
typing of sessions offered along channel variables in $\Delta$ and given the secrecy levels in the
security lattice $\Psi$.  $\Delta$ is a \emph{linear}
context that consists of a finite set of assumptions of the form $y_i{:}B_i\maxsec{[d_i']}$,
indicating for each channel variable $y_i$ its maximal secrecy $\maxsec{[d_i']}$ and the
offered session type $B_i$.  Channel variables $y_i$ must be unique in $\Delta$ and different
from $x$.  This well-formedness condition together with the fact the sequent has exactly one
succedent, turns process configurations into trees.  The process $P$ under consideration is the
parent node of all the processes providing along channels in $\Delta$.

We point out our use of ``channel variable'' for $x$ and $y_i$.  Channels only exist at
run-time, being allocated whenever a process is spawned and substituted for the channel
variables occurring in process terms.  As a result, channel variables can be $\alpha$-varied,
as usual.  For brevity, we will use the term channel rather than channel variable, whenever the
context determines whether a variable or run-time channel is meant.

The presuppositions guarantee that \textit{(i)} the maximal secrecy of a child node is at most
as high as the maximal secrecy of the providing (parent) node and that \textit{(ii)} the
running secrecy of the providing (parent) node is capped by its maximal secrecy.  By
transitivity, assertion \textit{(i)} holds equally for any descendant of the providing node.
Assertion \textit{(ii)} ensures that a node can never obtain more secrets than it is licensed
to.  We refer to both assertions as the \emph{tree invariant}.  Stating the tree invariant as a
presupposition requires the process term typing rules to preserve, but not to establish the
invariant.  This is sufficient because the tree invariant holds for any well-typed process
configuration, as expressed by the configuration typing rules discussed in
\secref{sec:config_typing}.

The tree invariant is sufficient to rule out any direct flows.  For example, the attempt to
send Alice's authorization process to a walk-in customer in process $\m{LeakyBank}$ (see
\secref{sec:motexample}), violates the tree invariant and thus does not type-check.  The tree
invariant, however, is not sufficient to rule out indirect flows.  To tackle indirect flows the
type system must make sure that the running secrecy of a process always soundly reflects the
level of secret information a process has obtained so far.  To this end, it increases the
running secrecy upon each receive and correspondingly guards sends, according to the following
schema:

\begin{enumerate}[label=(\roman*)]

\item \textbf{After} receipt of a message, the running secrecy of the receiving process must be
  increased to \textbf{at least} the maximal secrecy of the sending process, and

\item \textbf{before} sending a message, the running secrecy of the sending process must be
  \textbf{at most} the maximal secrecy of the receiving process.

\end{enumerate}

\noindent This schema intimately relies on the tree invariant and uses the maximal secrecy as a
sound approximation for the running secrecy of a process.  We refer to it as the \emph{secrecy pas
  de deux}.  The next section puts the discussed ideas into action.

%% file: type-system.tex
\section{IFC Session Type System}\label{sec:type-system}

This section formalizes $\lang$, giving the process term typing, configuration typing, and 
asynchronous semantics.  The system implements the ideas discussed in the previous section to
rule out both direct and indirect information flows.  We defer proofs of type safety and
noninterference to \secref{sec:metatheory}.

\subsection{Process Typing}\label{sec:proc_typing}

Our process typing rules are based on the \emph{sequent calculus}, leading to a left and a
right rule for each connective, describing the interaction from the point of view of the
provider and client, respectively.  We first discuss the rules for the individual connectives
in \tabref{tab:session_types} and then conclude with the judgmental rules cut and identity.

\subsubsection{Internal and External Choice}

Internal ($\oplus$) and external ($\&$) choice are the branching constructs, giving the choice to the provider or the client, respectively.
\begin{small}
\[
\inferrule*[right=$\oplus R$]
{\Psi; \Delta \vdash P@d_1::  y{:}A_{k}[c] \\
k \in L}
{\Psi;\Delta \vdash (y.k; P)@d_1 ::  y{:}\oplus\{\ell{:}A_{\ell}\}_{\ell \in L}[c]}
\]
\[
\inferrule*[right=$\oplus L$]
{\Psi \Vdash d_2=c \sqcup d_1 \\\\
\Psi;  \Delta, x{:}A_{k}[c] \vdash  Q_k@d_2::  y{:}C[c'] \\
\forall k \in L}
{\Psi; \Delta, x{:}\oplus\{\ell:A_{\ell}\}_{\ell \in L}[c] \vdash (\mathbf{case}\, x(\ell \Rightarrow Q_\ell)_{\ell \in L})@d_1::  y{:}C[c']}
\]
\[
\inferrule*[right=$\& R$]
{\Psi; \Delta \vdash Q_k@c::  y{:}A_{k}[c] \\
\forall k \in L}
{\Psi;\Delta \vdash (\mathbf{case}\, y(\ell \Rightarrow Q_\ell)_{\ell \in I})@d_1::  y{:}\&\{\ell:A_{\ell}\}_{\ell \in L}[c]}
\]
\[
\inferrule*[right=$\& L$]
{\Psi \Vdash d_1 \sqsubseteq c \\
\Psi;  \Delta, x{:}A_{k}[c] \vdash  P@d_1::  y{:}C[c'] \\
k \in L}
{\Psi; \Delta, x{:}\&\{\ell:A_{\ell}\}_{\ell \in I}[c] \vdash (x.k; P)@d_1::  y{:}C[c']}
\]
\end{small}
Let's convince ourselves that the rules preserve the tree invariant.  To preserve
the invariant, we may assume that the invariant holds for the conclusion and must establish it
for the premise.  Since the rules do neither add to or remove any channels from $\Delta$, they
preserve the invariant by assumption.  Let's examine whether the rules implement the secrecy pas de deux.  In
case of a receive, the running secrecy of the continuation must be increased to at least the
maximal secrecy of the sending channel.  In $\oplus L$, the premise $d_1 \sqcup c$ makes this
adjustment.  In $\& R$, no explicit adjustment is needed because the new running secrecy
$d_1\sqcup c$ amounts to $c$, by the tree invariant.  In case of a send, on the other hand, the
send is only admissible if the running secrecy of the sender is at most the maximal secrecy of
the receiving channel.  In $\oplus R$, this guard ($d_1 \sqsubseteq c$) is already established
by the tree invariant. $\&L$ explicitly establishes the guard with the premise
$\Psi \Vdash d_1 \sqsubseteq c$.

\subsubsection{Higher-Order Channels}

Tensor ($\chanout$) and lolli ($\chanin$) denote channel output (send) and input (receive), respectively.
\begin{small}
\[
\inferrule*[right=$\otimes R$]
{\Psi;\Delta \vdash  P@d_1:: y{:}B[c]}
{\Psi;\Delta, z{:}A[c]\vdash (\mathbf{send}\,z\,y; P)@d_1::  y{:}A \otimes B[c]}
\]
\[
\inferrule*[right=$\otimes L$]
{d_2= c\sqcup d_1 \\
\Psi':= (\Psi, \psi = c) \\
\Psi';\Delta, z:A[\psi], x{:}B[c] \vdash  P@d_2:: y{:} C[c']}
{\Psi;\Delta, x{:}A \otimes B[c] \vdash (z\leftarrow \mathbf{recv}\, x; P)@d_1::  y{:}C[c']}
\]
\[
\inferrule*[right=$\chanin R$]
{\Psi':=(\Psi, \psi = c)\\
\Psi';\Delta,  z{:}A[\psi] \vdash  P@c:: y{:} B[c]}
{\Psi; \Delta \vdash (z \leftarrow \mathbf{recv}\,y; P)@d_1::  y{:}A \multimap B[c]}
\]
\[
\inferrule*[right=$\chanin L$]
{\Psi \Vdash d_1 \sqsubseteq d \\
\Psi; \Delta, x{:}B[d] \vdash  P@d_1:: y{:}C[c']}
{\Psi;\Delta, z{:}A[d], x{:}A \multimap B[d] \vdash (\mathbf{send}\, z\,x; P)@d_1::  y{:}C[c']}
\]
\end{small}
To understand that the rules preserve the tree invariant, it is helpful to remind ourselves
that the connectives $\chanout$ and $\chanin$ change the tree structure, making a child a
sibling of the sender and a sibling the child of the recipient, respectively.  $\otimes R$
preserves the tree invariant without any extra conditions.  By assumption we know that the
maximal secrecy of the sent channel is equal to the maximal secrecy $c$ of the
provider.  Also by assumption, we know that the maximal secrecy of the provider is less than or
equal to the one of its parent, \emph{ensuring} that the tree invariant is preserved for
$\chanout L$ as well.  $\chanout R$ also implements the secrecy pas de deux, since
$d_1 \sqsubseteq c$ by assumption.  While $\chanout R$ license us to \emph{assume} in
$\chanout L$ that the maximal secrecy $\psi$ of the received channel $z$ is equal
to the maximal secrecy $c$ of the sending channel $x$, the actual maximal secrecy level of $z$
is statically unknown.  As a result, $\psi$ stands for a secrecy \emph{variable}, and we extend
the security lattice with $ \psi = c$.  The premise
$d_2= c\sqcup d_1$ in $\chanout L$ lastly implements the secrecy pas de deux, raising the
running secrecy of the continuation $P$ to $c$, unless $c \sqsubseteq d_1$.  The reasoning for
$\chanin R$ and $\chanin L$ are analogous, but with the roles reversed.  

\subsubsection{Termination}

The multiplicative unit ($1$) denotes process termination.
\begin{small}
\[
\inferrule*[right=$1 R$]
{\strut}
{\Psi;\cdot \vdash (\mathbf{close}\,y)@d_1 ::  y{:}1[c]}
\]
\[
\inferrule*[right=$1 L$]
{\Psi \Vdash d_2= c \sqcup d_1 \\
\Psi; \Delta \vdash Q@d_2 ::  y{:}T[d]}
{\Psi;\Delta, x{:}1[c] \vdash (\mathbf{wait}\,x;Q)@d_1 ::  y{:}C[d]}
\]
\end{small}
$1 R$ trivially preserves the tree invariant because there is no continuation and implements
the secrecy pas de deux since $d_1 \sqsubseteq c$ by assumption.  Similarly, $1 L$ preserves
the tree invariant by simply removing a channel from the continuation and implements the
secrecy pas de deux with the left premise.

\subsubsection{Identity and Cut}

Identity and cut are the two rules that do not result in any communication.  Identity amounts
to termination after identifying the involved channels and cut to process spawning.  For
simplicity, we do not support process definitions.  The examples from \secref{sec:motexample}
can be rewritten by inlining the body of the process definition when called.
\begin{small}
\[
\inferrule*[right=$\m{Fwd}$]
{\strut}
{\Psi; x{:}A[c] \vdash (y \leftarrow x)@{d_1} ::  y{:}A[c]}
\]
\[
\inferrule*[right=$\m{Cut}$]
{\Psi \Vdash d_1 \sqsubseteq d_2 \sqsubseteq d' \\
\forall \, z{:}A[c']\in \Delta_1.\, \Psi \Vdash c' \sqsubseteq d' \\
\Psi; \Delta_1 \vdash P@{d_2} ::x: B[d'] \\\\
\Psi \Vdash  d'\sqsubseteq d \\
\Psi; x:B[d'], \Delta_2  \vdash Q@{d_1} ::  y{:}C[d]}
{\Psi;\Delta_1, \Delta_2 \vdash  ((x^{d'} \leftarrow P)@d_2; Q)@{d_1} ::  y{:}C[d]}
\]
\end{small}
We briefly comment on $\m{Cut}$.  The premise $\Psi \Vdash d'\sqsubseteq d$ establishes the
tree invariant for the continuation $Q$ and the premise
$\forall \, z{:}A[c']\in \Delta_1.\, \Psi \Vdash c' \sqsubseteq d'$ for the spawned process
$P$.  The premise $\Psi \Vdash d_1 \sqsubseteq d_2 \sqsubseteq d'$ is vital to prevent any
indirect flows from $Q$ via $P$.  It ensures that the newly
spawned process has at least the knowledge of secret information that its spawner has.  Thanks
to this premise the below insecure example, which indirectly leaks information about the success
of Alice's authorization to the adversary $x_1$, is rejected.

\begin{center}
\begin{small}
\begin{minipage}[t]{\textwidth}
\begin{tabbing}
$x_1{:} \&\{s{:}1, f{:} 1\}\maxsec{[\mb{guest}]},$ \\
$u{:}\,\&\{\mathit{s}{:}\m{account}, \mathit{f}{:}1\}\maxsec{[\mb{alice}]} \vdash \m{SneakyaAuth} :: x{:}\m{auth}\maxsec{[\mb{alice}]}$ \\
$x \leftarrow \m{SneakyaAuth}\leftarrow u, x_1 = ($ \\
\quad \= $\mb{case} \, x \, ($\= $\mathit{tok}_j \Rightarrow$ \=$x.\mathit{succ};u.\mathit{s}; z_1 \leftarrow S \leftarrow x_1;$
\comment{// insecure spawn} \\
\> \> \>$\mb{send}\, u\, x ; \mb{wait}\, z_1; \mb{close}\,x$ \\
\> \> $\mid \mathit{tok}_{i \neq j} \Rightarrow$ \=$x.\mathit{fail}; u.\mathit{f}; z_1 \leftarrow F \leftarrow x_1;$
\comment{// insecure spawn} \\
\> \> \> $\mb{wait}\, u; \mb{wait}\, z_1; \mb{close}\, x ))\runsec{@\mb{alice}}$ \\[3pt]
$x_1{:}\&\{s{:}1, f{:} 1\}\maxsec{[\mb{guest}]} \vdash \m{S} :: z_1{:}1\maxsec{[\mb{alice}]}$ \\
$z_1 \leftarrow \m{S}\leftarrow x_1 = (x_1.s; \mb{wait}\, x_1; \mb{close}\, z_1)\runsec{@\mb{guest}}$ \\[3pt]
$x_1{:}\&\{s{:}1, f{:} 1\}\maxsec{[\mb{guest}]} \vdash \m{F} :: z_1{:}1\maxsec{[\mb{alice}]}$ \\
$z_1 \leftarrow \m{F}\leftarrow x_1 = (x_1.f; \mb{wait}\, x_1; \mb{close}\, z_1)\runsec{@\mb{guest}}$
\end{tabbing}
\end{minipage}
\end{small}
\end{center}

\subsection{Asynchronous Dynamics}

We define an \emph{asynchronous} dynamics for $\lang$ because it is not only more practical but
also allows for a more accurate statement of noninterference.  The dynamics is in line
with~\cite{BalzerICFP2017,DasCSF2021}, with the difference that it considers open
configurations.  The result is shown in \figref{fig:dynamics}.  We first convey the main ideas
and then comment on selected rules.

In an asynchronous semantics only receivers can be blocked, while senders just output the
message and proceed with their continuation.  We model such outputted messages as special
$\mb{msg}(P)$ processes that just contain the particular message.  In order to ensure that an
outputted message is properly sequenced with the sender's continuation, we use forwarding.
\figref{fig:dynamics_idea} schematically illustrates this idea, showing the case of a
\emph{positive} (sending) connective in the first line and the case of a \emph{negative}
(receiving) connective in the second line, with $\m{S}$ and $\m{R}$ standing for the sending
and receiving process, respectively.  The message process $\mb{msg}(P)$ is depicted in red.
This process has a subtree, in case of $\chanout$ and $\chanin$.  We can think of the message
as being spawned by the sender.  This results in the allocation of a new generation
$y_{\alpha+1}$ of the carrier channel $y_\alpha$.  The forward then links the two generations
$y_{\alpha+1}$ and $y_\alpha$ appropriately.  In case of a positive connective, $y_{\alpha+1}$
is forwarded to $y_\alpha$, in case of a negative connective, $y_\alpha$ is forwarded to
$y_{\alpha+1}$.  Once the message has been received, it terminates and $y_{\alpha+1}$ is
substituted for $y_\alpha$ in the receiver's continuation $\m{R'}$.  Messages can be ``queued
up'' as long as the polarity of the carrier channel stays the same.  Session typing ensures
that any messages ``in flight'' must first be received before the polarity of the carrier
channel changes.

\begin{figure}
\begin{center}
\includegraphics[scale=0.38]{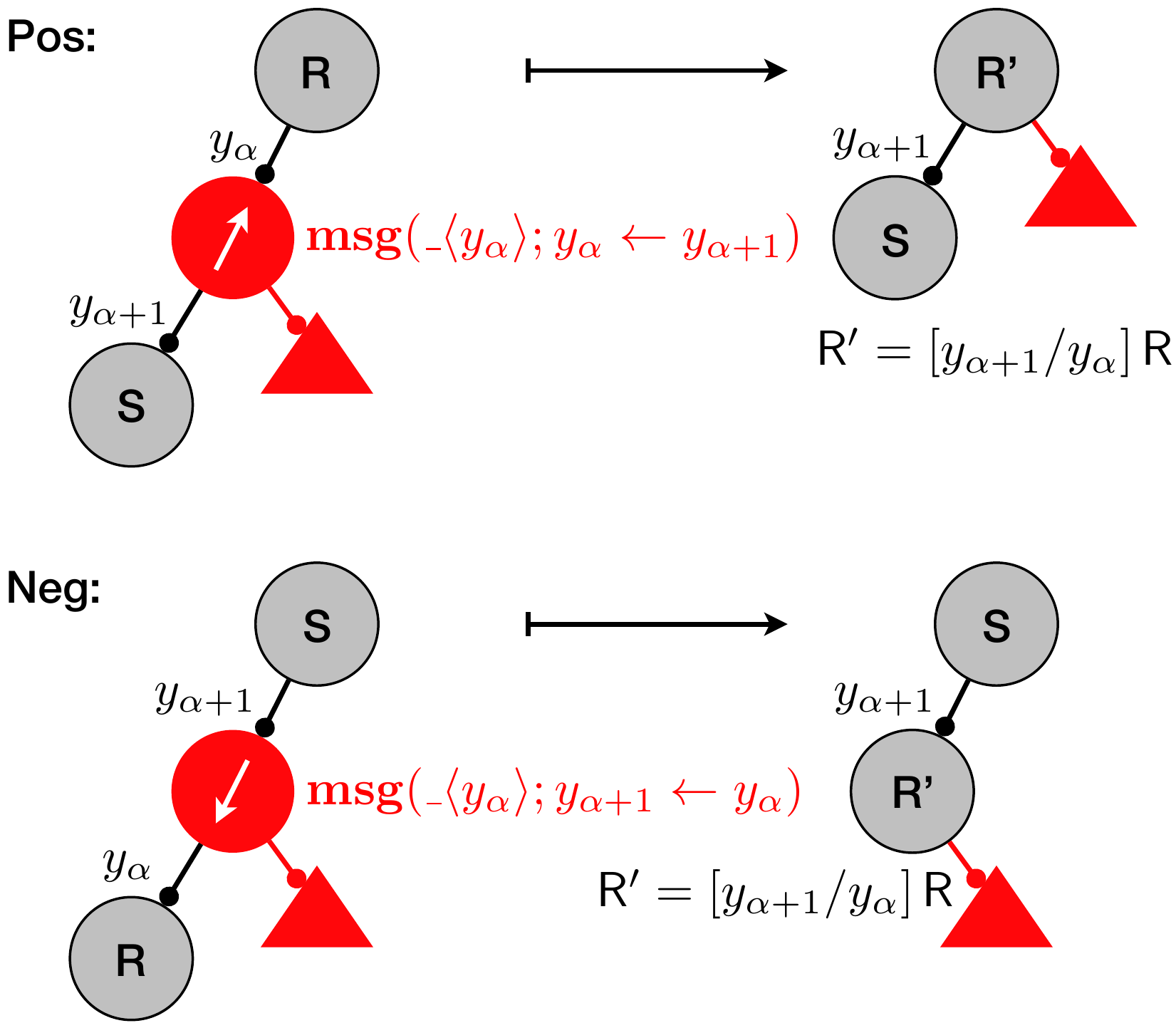}
\caption{Schematic illustration of asynchronous dynamics.}
\label{fig:dynamics_idea}
\end{center}
\end{figure}

\begin{figure*}
    \centering
    {\small\[\begin{array}{lclr}
    \mc{C}_1 \mb{proc}(y_\alpha[c], (y_\alpha \leftarrow x_\beta)@d_1) \mc{C}_2 \ostep \mc{C}_1 [x_\beta/y_\alpha]\mc{C}_2 \qquad (y_\alpha \not \in \Delta')& \; \m{fwd}\\
    \mc{C}_1 \mb{proc}(y_\alpha[c],  (x^d \leftarrow P)@d_2 ; Q@d_1) \mc{C}_2 \ostep \mc{C}_1 \mb{proc}(x_0[d], (\subst{x_0}{x}{P})@d_2)\mb{proc}(y_\alpha[c],  (\subst{x_0}{x}{Q})@d_1) \mc{C}_2 \qquad (x_0\, \textit{fresh}) & \m{Cut}\\
     \mb{proc}(y_\alpha[c],(\mathbf{close}\, y_\alpha)@d_1)\,\mc{C}_2 \ostep \mb{msg}(\mathbf{close}\, y_\alpha) \mc{C}_2 & 1 \\ 
        \mc{C}_1 \mb{msg}(\mathbf{close}\, y_\alpha) \mc{C}'\mb{proc}(x_\beta[{c'}], (\mathbf{wait}\,y_\alpha;Q)@d_1) \mc{C}_2 \ostep \mc{C}_1 \mc{C}' \mb{proc}(x_\beta[{c'}], Q@(d_1 \sqcup c)) \mc{C}_2 & 1 \\ 
        \mc{C}_1 \mb{proc}(y_\alpha[c], y_\alpha.k; P@d_1)\mc{C}_2 \ostep \mc{C}_1 \mb{proc}(y_{\alpha+1}[c], ([y_{\alpha+1}/y_{\alpha}]P)@d_1) \mb{msg}( y_\alpha.k;y_\alpha \leftarrow y_{\alpha+1}) \mc{C}_2 & \oplus\\
          \mc{C}_1 \mb{msg}(y_\alpha[{c}].k; y_\alpha \leftarrow v_{\delta}))\mc{C}' \mb{proc}(u_\gamma[{c'}],\mb{case}\, y_\alpha (( \ell \Rightarrow P_\ell)_{\ell \in L})@d_1) \mc{C}_2 \ostep   \mc{C}_1\mc{C}'\mb{proc}(u_{\gamma}[{c'}], ([v_{\delta}/y_{\alpha}] P_k)@(d_1\sqcup c)) \mc{C}_2 & \oplus\\
    \mc{C}_1 \mb{proc}(y_\alpha[c], (x_\beta.k; P)@d_1)\mc{C}_2 \ostep  \mc{C}_1 \mb{msg}( x_\beta.k; x_{\beta+1} \leftarrow x_\beta) \mb{proc}(y_{\alpha}[c], ([x_{\beta+1}/x_{\beta}]P)@d_1) \mc{C}_2 & \& \\ 
        \mc{C}_1  \mb{proc}(y_\alpha[c],(\mb{case}\, y_\alpha ( \ell \Rightarrow P_\ell)_{\ell \in L})@d_1) \mc{C}'\mb{msg}(y_\alpha.k; v_\delta \leftarrow y_\alpha) \mc{C}_2 \ostep  \mc{C}_1\mb{proc}(v_{\delta}[c],([v_{\delta}/y_{\alpha}] P_k)@c)\mc{C}' \mc{C}_2 & \& \\
          \mc{C}_1 \mb{proc}(y_\alpha[c],(\mathbf{send}\,x_\beta\,y_\alpha; P)@d_1)\mc{C}_2 \ostep \mc{C}_1 \mb{proc}(y_{\alpha+1}[c], ([y_{\alpha+1}/y_{\alpha}]P)@d_1) \mb{msg}( \mathbf{send}\,x_\beta\,y_\alpha; y_\alpha \leftarrow y_{\alpha+1}) \mc{C}_2 & \otimes\\
    \mc{C}_1 \mb{msg}(\mathbf{send}\,x_\beta\,y_\alpha;y_{\alpha}\leftarrow v_\delta)\mc{C}' \mb{proc}(u_\gamma[{c'}],(w_\eta\leftarrow \mathbf{recv}\,y_\alpha; P)@d_1) \mc{C}_2 \ostep   \mc{C}_1\mc{C}'\mb{proc}(u_{\gamma}[{c'}], ([x_\beta/w_\eta][v_{\delta}/y_{\alpha}] P)@(d_1\sqcup c)) \mc{C}_2 & \otimes\\
    \mc{C}_1 \mb{proc}(y_\alpha[c],(\mathbf{send}\,x_\beta\,u_\gamma; P)@d_1) \mc{C}_2 \ostep  \mc{C}_1 \mb{msg}( \mathbf{send}\,x_\beta\,u_\gamma; u_{\gamma+1}\leftarrow u_\gamma) \mb{proc}(y_{\alpha}[c], ([u_{\gamma+1}/u_{\gamma}]P)@d_1) \mc{C}_2 & \multimap\\
    \mc{C}_1  \mb{proc}(y_\alpha[c],(w_\eta\leftarrow \mathbf{recv}\,y_\alpha; P)@d_1)\mc{C}'\mb{msg}(\mathbf{send}\,x_\beta\,y_\alpha;v_{\delta}\,\leftarrow y_\alpha) \mc{C}_2 \ostep   \mc{C}_1\mb{proc}(v_\delta[c], ([x_\beta/w_\eta][v_{\delta}/y_{\alpha}] P)@ c) \mc{C}'\mc{C}_2 & \multimap\\
\end{array}\]
}
    \caption{Asynchronous dynamics of $\lang$.}
    \label{fig:dynamics}
\end{figure*}

\figref{fig:dynamics} defines the asynchronous dynamics in terms of rewriting rules
$\mc{C} \ostepp \mc{C'}$ that rewrite open configuration $\mc{C}$ with type
$\Psi;\Delta\Vdash \mc{C}:: \Delta'$ to open configuration $\mc{C'}$ with type
$\Psi;\Delta\Vdash \mc{C'}:: \Delta'$.  We detail the configuration typing in the next section.
$\m{Cut}$ allocates a fresh channel $x_0$ at \emph{generation} 0.  This channel is substituted
for the channel variable $x$ occurring in the process terms $P$ and $Q$ in the post-state.  The
generation $\alpha$ of a channel $y_\alpha$ is incremented to $\alpha+1$ whenever a new message
is spawned, except for $1$ because there is no continuation.  Lastly we point out that
$\m{fwd}$ is not defined for any channels in $\Delta'$ because those configurations are
considered \emph{poised}, as we discuss in \secref{sec:poised_config}.

\subsection{Configuration Typing}\label{sec:config_typing}

We use the judgment $\Psi;\Delta\Vdash \mc{C}:: \Delta'$ to type an open configuration
$\mc{C}$.  An open configuration consists of an open \emph{forest} of processes
$\mathbf{proc}(x[d], P@d_1)$ and messages $\mathbf{msg}(P)$.  While our logical relation is
phrased in terms of an open \emph{tree} --- representing the partial program under
consideration --- typing of an open forest is necessitated by the inductive nature of the below
rules.  The judgment indicates that $\mc{C}$ provides sessions in $\Delta'$, using sessions in
$\Delta$, and given the security lattice $\Psi$.  Both $\Delta'$ and $\Delta$ are linear
contexts, consisting of a finite set of assumptions of the form $y_i{:}B_i\maxsec{[d_i']}$,
where $y_i$ denotes an actual \emph{channel} that has been allocated upon spawning a process.
For simplicity, we do not display a channel's generation.
\begin{small}
\[
\inferrule*[right=$\mathbf{emp}_1$]
{\strut}
{\Psi; x{:}A[d] \Vdash \cdot :: (x{:}A[d])}
\qquad
\inferrule*[right=$\mathbf{emp}_2$]
{\strut}
{\Psi; \cdot \Vdash \cdot :: (\cdot)}
\]
\[
\inferrule*[right=$\mathbf{proc}$]
{\Psi \Vdash d_1 \sqsubseteq d \\
\forall y{:}B[d'] \in \Delta'_0, \Delta \, (\Psi \Vdash d' \sqsubseteq d) \\\\
\Psi; \Delta_0 \Vdash \mathcal{C}:: \Delta \\
\Psi; \Delta'_0, \Delta \vdash P@d_1:: (x{:}A[d])}
{\Psi; \Delta_0, \Delta'_0 \Vdash \mathcal{C}, \mathbf{proc}(x[d], P@d_1):: (x{:}A[d])}
\]
\[
\inferrule*[right=$\mathbf{msg}$]
{\forall y{:}B[d'] \in \Delta'_0, \Delta \, (\Psi \Vdash d' \sqsubseteq d) \\
\Psi; \Delta_0 \Vdash \mathcal{C}:: \Delta \\
\Psi; \Delta'_0, \Delta \vdash P@d:: (x{:}A[d])}
{\Psi; \Delta_0, \Delta'_0 \Vdash \mathcal{C}, \mathbf{msg}(P):: (x{:}A[d])}
\]
\[
\inferrule*[right=$\mathbf{comp}$]
{\Psi; \Delta_0 \Vdash \mathcal{C}:: \Delta \\
\Psi; \Delta'_0 \Vdash \mc{C}_1:: x{:}A[d]}
{\Psi; \Delta_0, \Delta'_0 \Vdash \mathcal{C}, \mc{C}_1 :: \Delta, x{:} A[d]}
\]
\end{small}
Rule $\mathbf{comp}$ types an open forest, singling out the open tree $\mc{C}_1$ rooted at $x$.
Rules $\mathbf{proc}$ and $\mathbf{msg}$ type open trees, singling out their root process or
message, respectively.  Both rules include sufficient premises to establish the tree invariant.
Unlike processes, messages have no running secrecy associated because their running secrecy is
determined by the maximal secrecy of the sender.  \secref{sec:metatheory} provides further
details.  Rules $\mathbf{emp}_1$ and $\mathbf{emp}_2$, finally, type an empty open forest.

\subsection{Poised Configuration}\label{sec:poised_config}

What are values in a functional setting are \emph{poised} configurations here.  Prior
work~\cite{PfenningFOSSACS2015} has defined that notion only for closed configurations, and we
generalize it to open configurations.  An open configuration $\Psi;\Delta\Vdash \mc{C}::
\Delta'$ is poised, iff it is empty or none of its processes and messages can communicate with each other and
there exists at least one process or message that attempts to communicate along a channel in
$\Delta$ or $\Delta'$.

\begin{definition}[Poised Configuration]
A configuration $\Delta_1, \Delta_2 \Vdash \mc{C}_1, \mc{C}_2:: \Lambda, w{:}A'[c]$ is poised iff either $\mc{C}_1, \mc{C}_2$ is empty or $\Delta_1 \Vdash \mc{C}_1:: \Lambda$ is poised and $ \Delta_2 \Vdash \mc{C}_2::  w{:}A'[c]$ is poised. The configuration $\Delta_2 \Vdash \mc{C}_2 :: w{:}A'[c]$ is poised iff it cannot take any steps and at least one of the following conditions hold:
\begin{itemize}
\item $\mc{C}_2$ is an empty configuration.
\item $\mathcal{C}_2= \mc{C}'_2\mb{msg}(P)\mc{C}''_2$ such that $\mb{msg}(P)$ is a negative message along $y \in \Delta_2$, i.e. $y{:}\&\{\ell{:}A_\ell\}_{\ell \in L}[c_1]\Vdash \mb{msg}(P):: x{:}A_k[c_1]$ or $y{:}A\multimap B[c_1], z{:}A[c_1] \Vdash \mb{msg}(P):: x{:}B[c_1]$, and both subconfigurations $\mc{C}_2'$ and $\mc{C}_2''$ are poised. 
\item  $\mathcal{C}_2= \mb{proc}(x[c'],P@d_1)\,\mc{C}_2'$ such that $\mb{proc}(x[c'], P@d_1)$ attempts to receive along a channel $y {\in} \Delta_2$.
\item $\mathcal{C}_2= \mc{C}_2'\mb{msg}(P)$ such that $\mb{msg}(P)$ is a positive message sent along $w{:}A'[c]$, i.e. $ x{:}A_k[c] \Vdash \mb{msg}(P)::w{:}\oplus\{\ell{:}A_\ell\}_{\ell \in L}[c]$ or $x{:}B[c], z_{\gamma}{:}A[c] \Vdash \mb{msg}(P):: w{:}A\otimes B[c]$,or $\cdot \Vdash \mb{msg}(P):: w{:}1[c]$,
and subconfiguration $\mc{C}'_2$  is poised. 
\item   $\mathcal{C}_2= \mb{proc}(w[c], P@d_1)\,\mc{C}_2'$ such that $\mb{proc}(w[c], P@d_1)$ attempts to receive along $w{:}A'[c]$.
\item   $\mathcal{C}_2= \mc{C}_2'\mb{proc}(w^c \leftarrow x^c @d_1)\,\mc{C}_2''$.
    
\end{itemize}
\end{definition}

%% file: key-ideas2.tex
\section{Key Ideas - Part II}\label{sec:key-ideas2}

This section develops the main ideas underlying the session logical relation used to prove
noninterference of $\lang$.  The next section puts these ideas into action.

\emph{Noninterference} essentially amounts to a \emph{program equivalence up to} the secrecy
level $\xi$ of the observer, requiring that two runs of a program may only differ in outputs
whose secrecy level is above or incomparable to $\xi$.  The fundamental property of the logical
relation for noninterference then is stated for two runs of any partial program, showing that
the runs are related, if given related inputs.

In a session-typed setting, partial programs amount to \emph{open trees} and outputs to
\emph{messages sent} from that open tree.  Inputs, on the other hand, consist of the
\emph{messages received} from any \emph{closing configurations}.

Given these basic correspondences, we can develop our session logical relation for
noninterference schematically based on \figref{fig:logical_relation}.
\figref{fig:logical_relation} shows two runs $\color{brown}{\mc{D}}_1$ and $\color{brown}{\mc{D}}_2$ of a
partial program with closing substitutions $\color{red}{\mc{C}}_1$, $\color{blue}{\mc{F}}_1$ and
$\color{red}{\mc{C}}_2$, $\color{blue}{\mc{F}}_2$, respectively, and post-states
$\lre{\mathcal{\mc{C}}_1'}{\mathcal{D}_1'} {\mathcal{F}_1'}$ and
$\lre{\mathcal{\mc{C}}_2'}{\mathcal{D}_2'} {\mathcal{F}_2'}$, resulting from a message exchange.
The session logical relation now mandates that $\color{brown}{\mc{D}}_1$ and $\color{brown}{\mc{D}}_2$
will send the same messages to $\color{red}{\mc{C}}_1$, $\color{blue}{\mc{F}}_1$ and $\color{red}{\mc{C}}_2$,
$\color{blue}{\mc{F}}_2$, respectively, provided that $\color{red}{\mc{C}}_1$, $\color{blue}{\mc{F}}_1$ and
$\color{red}{\mc{C}}_2$, $\color{blue}{\mc{F}}_2$ will send the same messages to $\color{brown}{\mc{D}}_1$ and
$\color{brown}{\mc{D}}_2$, respectively.  This property is expressed as
\[
({\lr{\mc{C}_1} {\mc{D}_1} {\mc{F}_1};{\lr{\mc{C}_2} {\mc{D}_2} {\mc{F}_2}}}) \in \mc{V}^\xi_{ \Psi}\llbracket \Delta \Vdash K\rrbracket
\]
\noindent where $\Delta$ amounts to the typing of channels connecting $\color{red}{\mc{C}}_1$ and
$\color{red}{\mc{C}}_2$ with $\color{brown}{\mc{D}}_1$ and $\color{brown}{\mc{D}}_2$, respectively, and $K$ to
the typing of the channel connecting $\color{brown}{\mc{D}}_1$ and $\color{brown}{\mc{D}}_2$ with
$\color{blue}{\mc{F}}_1$ and $\color{blue}{\mc{F}}_2$, respectively.  We refer to $\Delta$ and $K$ as the
\emph{interface} of $\color{brown}{\mc{D}}_1$ and $\color{brown}{\mc{D}}_2$.

\begin{figure}
\begin{center}
\includegraphics[scale=0.38]{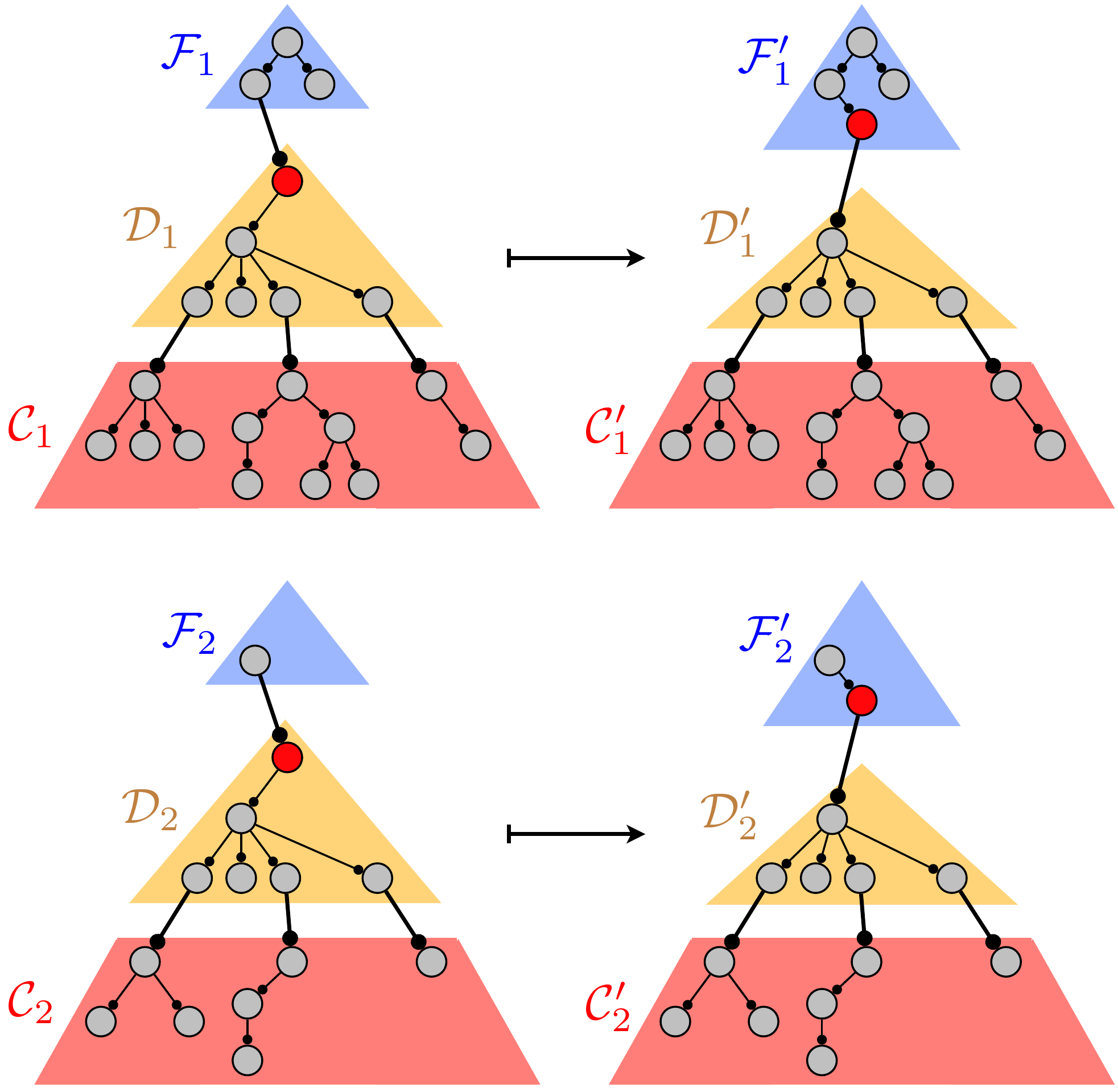}
\caption{Session logical relation for noninterference: key ideas.}
\label{fig:logical_relation}
\end{center}
\end{figure}

Clearly, the above property can only hold for exchanged messages of at most the observer's
secrecy level.  We call such messages and their carrying channels \emph{observable}.  We thus
phrase the logical relation only over observable channels, requiring us to determine
$\color{brown}{\mc{D}}_1$ and $\color{brown}{\mc{D}}_2$ for two runs $\mc{D}_3$ and $\mc{D}_4$ of a
partial program, with $\Psi; \Delta_3 \Vdash \mc{D}_3:: K_3$ and
$\Psi; \Delta_4 \Vdash \mc{D}_4:: K_4$, such that the observable channels defined by the
projection $\_ \Downarrow \xi$ are the same, \ie
$\Delta_3 \Downarrow \xi = \Delta_4 \Downarrow \xi = \Delta$ and
$K_3\Downarrow \xi = K_4 \Downarrow \xi = K$.  The left-over, non-observable channels in
$\Delta_3$, $K_3$ and $\Delta_4$, $K_4$ are closed off and internalized into $\mc{D}_3$ and
$\mc{D}_4$, yielding $\color{brown}{\mc{D}}_1$ and $\color{brown}{\mc{D}}_2$, respectively.

The message exchange depicted in \figref{fig:logical_relation} is a send, denoted by the red
node in $\color{brown}{\mc{D}}_1$ and $\color{brown}{\mc{D}}_2$.  The node is a message $\mb{msg}(P)$,
and the figure captures the positive case depicted in \figref{fig:dynamics_idea}.  In the post-states
$\lre{\mathcal{C}_1'}{\mathcal{D}_1'} {\mathcal{F}_1'}$ and
$\lre{\mathcal{C}_2'}{\mathcal{D}_2'} {\mathcal{F}_2'}$, this message is simply pushed into the
substitutions $\color{blue}{\mc{F}}_1'$ and $\color{blue}{\mc{F}}_2'$.  The value interpretation of
$({\lr{\mc{C}_1} {\mc{D}_1} {\mc{F}_1};{\lr{\mc{C}_2} {\mc{D}_2} {\mc{F}_2}}})$ is now phrased
in terms of the transition, requiring that
\[
\begin{array}{r}
({\lr{\mc{C}_1} {\mc{D}_1} {\mc{F}_1};{\lr{\mc{C}_2} {\mc{D}_2} {\mc{F}_2}}}) \in \mc{V}^\xi_{ \Psi}\llbracket \Delta \Vdash K\rrbracket\text{, if} \\
({\lre{\mc{C}_1'} {\mc{D}_1'} {\mc{F}_1'},{\lre{\mc{C}_2'} {\mc{D}_2'} {\mc{F}_2'}}}) \in \mc{E}^\xi_{ \Psi}\llbracket \Delta' \Vdash K'\rrbracket
\end{array}
\]
The post-states $\lre{\mathcal{C}_1'}{\mathcal{D}_1'} {\mathcal{F}_1'}$ and
$\lre{\mathcal{C}_2'}{\mathcal{D}_2'} {\mathcal{F}_2'}$ now take any number of internal
transitions until $\color{red}{\mc{C}}_1''$, $\color{brown}{\mc{D}}_1''$, $\color{blue}{\mc{F}}_1''$,
$\color{red}{\mc{C}}_2''$, $\color{brown}{\mc{D}}_2''$, and $\color{blue}{\mc{F}}_2''$ are each individually poised,
demanding a message exchange along an observable channel.  We thus require this poised
configuration to be in the value interpretation
\[
({\lr{\mc{C}_1''} {\mc{D}_1''} {\mc{F}_1''};{\lr{\mc{C}_2''} {\mc{D}_2''} {\mc{F}_2''}}}) \in \mc{V}^\xi_{ \Psi}\llbracket \Delta' \Vdash K' \rrbracket
\]

The choice to simply push a message $\mb{msg}(P)$ across the interface to the recipient, rather
then consuming it with a corresponding receiving action, allows for more runs to be soundly
equated.  In particular, two runs are allowed to differ in the order in which the messages
$\mb{msg}(P)$ are consumed by the recipients, whenever typing ensures that the
recipients can no longer send back any messages to the senders.

Like any logical relation, our session logical relation accounts for the \emph{polarity} of the
connectives in $\D$ and $K$.  Moreover, it considers whether the message is being sent
along a channel in $\D$ or in $K$.  In case the message is being sent along a channel in $\D$, we refer
to it as communicating on the \emph{left}, otherwise, as communicating on the \emph{right}.
These two dimensions span the space of value interpretations of two program runs, requiring
\emph{positive} connectives to \emph{assert} the sending of the same message in both runs when
communicating on the \emph{right} and to \emph{assume} their existence when communicating on
the \emph{left}.  Conversely, \emph{negative} connectives can \emph{assume} that the same
messages are being sent in both runs when communicating on the \emph{right} and must
\emph{assert} sending of the same message in both runs when communicating on the \emph{left}.

%% file: logicalrel.tex
\section{Noninterference Logical Relation}\label{sec:logicalrel}

\begin{figure*}
    \centering
   {\small\[\begin{array}{llcl}
   (1)&  ({\lr{\cdot} {\mc{D}_1} {\mc{F}_1};{\lr{\cdot} {\mc{D}_1} {\mc{F}_1}}}) \in \mc{V}^\xi_{ \Psi}\llbracket \cdot \Vdash y_\alpha{:}1[c]\rrbracket & \;& (\lr{\cdot} {\mc{D}_1} {\mc{F}_1};\lr{\cdot} {\mc{D}_2} {\mc{F}_2}) \in \m{Tree}_\Psi(\cdot \Vdash y_\alpha:1[c])\, \m{and}\\ &&&
    \mc{D}_1=\mathbf{msg}(\mathbf{close}\,y_\alpha^c ) \, \m{and}\, \mc{D}_2=\mathbf{msg}( \mathbf{close}\,y_\alpha^c)\\[4pt]
   (2) & (\lr{\mc{C}_1} {\mc{D}_1} {\mc{F}_1};\lr{\mc{C}_2} {\mc{D}_2} {\mc{F}_2})\in  & \; &(\lr{\mc{C}_1} {\mc{D}_1} {\mc{F}_1};\lr{\mc{C}_2} {\mc{D}_2} {\mc{F}_2}) \in \m{Tree}_\Psi(\Delta \Vdash y_\alpha:\oplus\{\ell{:}A_\ell\}_{\ell \in I}[c])\, \m{and} \\
  & \qquad \qquad \mathcal{V}^\xi_{ \Psi}\llbracket (\Delta \Vdash y_\alpha{:}\oplus\{\ell{:}A_\ell\}_{\ell \in I}[c])\rrbracket &&  \mathcal{D}_1=\mathcal{D}'_1\mathbf{msg}( y_\alpha^c.k;y_\alpha^c \leftarrow u^c_\delta ) \; \m{and}\; \mathcal{D}_2=\mathcal{D}'_2 \mathbf{msg}(y_\alpha^c.k; y_\alpha^c \leftarrow u^c_\delta )\, \m{and} \\
   &&&  (\lre{\mc{C}_1} {\mc{D}'_1} {\mathbf{msg}(y_\alpha^c.k;y_\alpha^c \leftarrow u^c_\delta  )\mc{F}_1}, \lre{\mc{C}_2} {\mc{D}'_2} {\mathbf{msg}(y_\alpha^c.k;y_\alpha^c \leftarrow u^c_\delta  )\mc{F}_2}) \in  \mathcal{E}^\xi_{ \Psi}\llbracket \Delta \Vdash u_{\delta}{:}A_k[c]\rrbracket \\[4pt]
   
   (3) & (\lr{\mc{C}_1} {\mc{D}_1} {\mc{F}_1};\lr{\mc{C}_2} {\mc{D}_2} {\mc{F}_2}) \in  & \; &  (\lr{\mc{C}_1} {\mc{D}_1} {\mc{F}_1};\lr{\mc{C}_2} {\mc{D}_2} {\mc{F}_2}) \in \m{Tree}_\Psi(\Delta \Vdash y_\alpha{:}\&\{\ell{:}A_\ell\}_{\ell \in I}[c]),\,\m{and}\\
&\qquad \qquad \mathcal{V}^\xi_{ \Psi}\llbracket \Delta \Vdash y_\alpha{:}\&\{\ell{:}A_\ell\}_{\ell \in I}[c]\rrbracket&&\m{if} \,( \mathcal{F}_1=\mathbf{msg}( y_\alpha^c.k; u_\delta^c \leftarrow y_\alpha^c ) \mathcal{F}'_1 \; \m{and}\; \mathcal{F}_2=\mathbf{msg}( y_\alpha^c.k;u_\delta^c \leftarrow y_\alpha^c)\mathcal{F}'_2) \, \m{then} \\
   &&& \;(\lre{\mc{C}_1} {\mc{D}_1 \mathbf{msg}( y_\alpha^c.k;u_\delta^c \leftarrow y_\alpha^c)} {\mc{F}'_1},\lre{\mc{C}_2} {\mc{D}_2\mathbf{msg}( y_\alpha^c.k;u_\delta^c \leftarrow y_\alpha^c)} {\mc{F}'_2})\in  \mathcal{E}^\xi_{\Psi}\llbracket \Delta \Vdash u_{\delta}{:}A_k[c]\rrbracket  \\[4pt]
   
 (4)&  (\lr{\mc{C}'_1\mc{C}''_1} {\mc{D}_1} {\mc{F}_1};\lr{\mc{C}'_2\mc{C}''_2} {\mc{D}_2} {\mc{F}_2})  \in  & \; &  (\lr{\mc{C}'_1\mc{C}''_1} {\mc{D}_1} {\mc{F}_1};\lr{\mc{C}'_2\mc{C}''_2} {\mc{D}_2} {\mc{F}_2}) \in \m{Tree}_\Psi(\Delta \Vdash y_\alpha{:}A\otimes B[c])\,\m{and}\\ &\qquad \qquad  \mathcal{V}^\xi_{ \Psi}\llbracket \Delta',\Delta'' \Vdash y_\alpha{:}A\otimes B[c]\rrbracket &&\mathcal{D}_1=\mathcal{D}'_1\mathcal{T}_1\mathbf{msg}( \mb{send}x_\beta^c\,y_\alpha^c; y_\alpha^c \leftarrow u^c_\delta )\; \m{and}\\
 & &&   \mathcal{D}_2=\mathcal{D}'_2\mc{T}_2 \mathbf{msg}( \mb{send},x_\beta^c\,y_\alpha^c; y_\alpha^c \leftarrow u^c_\delta )\, \m{and}\\
   &&& (\lre{\mc{C}''_1} {\mc{T}_1} {\mc{C}'_1 \mc{D}'_1 \mathbf{msg}( \mb{send}x_\beta^c\,y_\alpha^c; y_\alpha^c \leftarrow u^c_\delta )\mc{F}_1},\\ &&& \qquad \qquad  \lre{\mc{C}''_2} {\mc{T}_2} {\mc{C}'_2\mc{D}'_2 \mathbf{msg}( \mb{send}x_\beta^c\,y_\alpha^c; y_\alpha^c \leftarrow u^c_\delta )\mc{F}_2})  \in  \mathcal{E}^\xi_{ \Psi}\llbracket \Delta'' \Vdash x_{\beta}{:}A[c]\rrbracket\,\m{and}\\
   &&&  (\lre{\mc{C}'_1} {\mc{D}'_1} {\mc{C}''_1\mc{T}_1 \mathbf{msg}( \mb{send}x_\beta^c\, y_\alpha^c; y_\alpha^c \leftarrow u^c_\delta )\mc{F}_1},\\ &&& \qquad \qquad  \lre{\mc{C}'_2} {\mc{D}'_2} {\mc{C}''_2\mc{T}_2 \mathbf{msg}( \mb{send}x_\beta^c\,y_\alpha^c; y_\alpha^c \leftarrow u^c_\delta )\mc{F}_2})  \in  \mathcal{E}^\xi_{ \Psi}\llbracket \Delta' \Vdash u_{\delta}{:}B[c]\rrbracket\\[4pt]
  
  (5) & (\lr{\mc{C}_1} {\mc{D}_1} {\mc{F}_1};\lr{\mc{C}_2} {\mc{D}_2} {\mc{F}_2})  \in  & \; &  (\lr{\mc{C}_1} {\mc{D}_1} {\mc{F}_1};\lr{\mc{C}_2} {\mc{D}_2} {\mc{F}_2}) \in \m{Tree}_\Psi(\Delta \Vdash y_\alpha{:}A\multimap B[c])\, \m{and}\\ &\qquad \qquad  \mathcal{V}^\xi_{ \Psi}\llbracket \Delta \Vdash y_\alpha{:}A\multimap B[c]\rrbracket&& \m{if}\, (\mathcal{F}_1=\mathcal{T}_1\mathbf{msg}( \mb{send}x_\beta^c\,y_\alpha^c;u_\delta^c \leftarrow y_\alpha^c)\mathcal{F}'_1\; \m{and} \\

  & && \qquad \qquad  \mathcal{F}_2=\mathcal{T}_2\mathbf{msg}( \mb{send}x_\beta^c\,y_\alpha^c;u_\delta^c \leftarrow y_\alpha^c)\mathcal{F}'_2)\,\m{then} \\
   &&&  \;\;(\lre{\mc{C}_1\mc{T}_1} {\mc{D}_1\mathbf{msg}( \mb{send}x_\beta^c\,y_\alpha^c;u_\delta^c \leftarrow y_\alpha^c)} {\mc{F}'_1},\\&&&\qquad \qquad \lre{\mc{C}_2\mc{T}_2} {\mc{D}_2\mathbf{msg}(\mb{send}x_\beta^c\,y_\alpha^c;u_\delta^c \leftarrow y_\alpha^c)} {\mc{F}'_2})  \in  \mathcal{E}^\xi_{ \Psi}\llbracket \Delta, x_\beta{:}A[c]\Vdash u_{\delta}{:}B[c]\rrbracket\\[4pt]
   
 (6) &  (\lr{\mc{C}_1} {\mc{F}_1}{\mc{D}_1}; \lr{\mc{C}_2} {\mc{F}_2}{\mc{D}_2})\in & \;& 
    (\lr{\mc{C}_1} {\mc{F}_1}{\mc{D}_1}; \lr{\mc{C}_2} {\mc{F}_2}{\mc{D}_2}) \in \m{Tree}_\Psi(\Delta, y_\alpha{:}1[c] \Vdash K)\, \m{and}\\
    & \qquad \qquad  \mc{V}^\xi_{ \Psi}\llbracket \Delta, y_\alpha{:}1[c] \Vdash K\rrbracket&&\m{if}\,(\mc{C}_1=\mathcal{C}'_1\mathbf{msg}( \mathbf{close}\,y_\alpha^c) \, \m{and}\, \mc{C}_2=\mathcal{C}'_2\mathbf{msg}( \mathbf{close}\,y_\alpha^c))\, \m{then} \\
   &&& \;\;(\lre{\mc{C}'_1} {\mathbf{msg}( \mathbf{close}\,y_\alpha^c)\mc{D}_1} {\mc{F}_1},\lre{\mc{C}'_2} {\mathbf{msg}( \mathbf{close}\,y_\alpha^c )\mc{D}_2} {\mc{F}_2}) \mathcal{E}^\xi_{\Psi}\llbracket \Delta \Vdash K\rrbracket  \\[4pt]
   (7) & (\lr{\mc{C}_1} {\mc{D}_1} {\mc{F}_1};\lr{\mc{C}_2} {\mc{D}_2} {\mc{F}_2}) \in  & \; & (\lr{\mc{C}_1} {\mc{D}_1} {\mc{F}_1};\lr{\mc{C}_2} {\mc{D}_2} {\mc{F}_2}) \in \m{Tree}_\Psi(\Delta, y_\alpha{:}\oplus\{\ell{:}A_\ell\}_{\ell \in I}[c] \Vdash K)\, \m{and}\\
&\qquad \qquad   \mathcal{V}^\xi_{ \Psi}\llbracket \Delta, y_\alpha:\oplus\{\ell{:}A_\ell\}_{\ell \in I}[c] \Vdash K\rrbracket  &&\m{if}\, (\mathcal{C}_1=\mathcal{C}'_1\mathbf{msg}( y_\alpha^c.k; y_\alpha^c \leftarrow u^c_\delta) \; \m{and}\; \mathcal{C}_2=\mathcal{C}'_2 \mathbf{msg}(y_\alpha^c.k;y_\alpha^c \leftarrow u^c_\delta ))\, \m{then}\\
   &&& \;\;(\lre{\mc{C}'_1} {\mathbf{msg}(y_\alpha^c.k; y_\alpha^c \leftarrow u^c_\delta)\mc{D}_1} {\mc{F}_1},\lre{\mc{C}'_2} {\mathbf{msg}(y_\alpha^c.k; y_\alpha^c \leftarrow u^c_\delta)\mc{D}_2} {\mc{F}_2})\\ &&& \qquad \qquad \in  \mathcal{E}^\xi_{ \Psi}\llbracket \Delta, u_{\delta}{:}A_k[c] \Vdash K \rrbracket  \\[4pt]

(8) &  (\lr{\mc{C}_1} {\mc{D}_1} {\mc{F}_1};\lr{\mc{C}_2} {\mc{D}_2} {\mc{F}_2}) \in   & \; &(\lr{\mc{C}_1} {\mc{D}_1} {\mc{F}_1};\lr{\mc{C}_2} {\mc{D}_2} {\mc{F}_2}) \in \m{Tree}_\Psi(\Delta, y_\alpha{:}\&\{\ell{:}A_\ell\}_{\ell \in I}[c] \Vdash K)\, \m{and}\\
&\qquad \qquad  \mathcal{V}^\xi_{ \Psi}\llbracket \Delta, y_\alpha{:}\&\{\ell{:}A_\ell\}_{\ell \in I}[c] \Vdash K \rrbracket&& \mathcal{D}_1=\mathbf{msg}( y_\alpha^c.k;u_\delta^c \leftarrow y_\alpha^c) \mathcal{D}'_1 \; \m{and}\ \mathcal{D}_2=\mathbf{msg}( y_\alpha^c.k;u_\delta^c \leftarrow y_\alpha^c)\mathcal{D}'_2\,\m{and} \\
   &&&  (\lre{\mc{C}_1 \mathbf{msg}( y_\alpha^c.k;u_\delta^c \leftarrow y_\alpha^c)} {\mc{D}'_1} {\mc{F}_1},\\
   &&& \qquad \qquad\lre{\mc{C}_2\mathbf{msg}( y_\alpha^c.k;u_\delta^c \leftarrow y_\alpha^c)} {\mc{D}'_2} {\mc{F}_2}) \in  \mathcal{E}^\xi_{ \Psi}\llbracket \Delta, u_{\delta}{:}A_k[c] \Vdash K \rrbracket  \\[4pt]
   
 (9)&  (\lr{\mc{C}_1} {\mc{D}_1} {\mc{F}_1};\lr{\mc{C}_2} {\mc{D}_2} {\mc{F}_2})  \in   & \; & (\lr{\mc{C}_1} {\mc{D}_1} {\mc{F}_1};\lr{\mc{C}_2} {\mc{D}_2} {\mc{F}_2}) \in \m{Tree}_\Psi(\Delta, y_\alpha{:}A\otimes B[c] \Vdash K)\,\m{and} \\
   &\qquad  \qquad\,  \mathcal{V}^\xi_{ \Psi}\llbracket \Delta, y_\alpha{:}A \otimes B[c] \Vdash K \rrbracket&&  \m{if}\,(\mathcal{C}_1=\mathcal{C}'_1\mc{T}_1\mathbf{msg}( \mb{send}x_\beta^c\,y_\alpha^c ; y_\alpha^c \leftarrow u^c_\delta)\, \m{and}\,\\    &&& \qquad \mathcal{C}_2=\mathcal{C}'_2 \mc{T}_2\mathbf{msg}( \mb{send}x_\beta^c\,y_\alpha^c ; y_\alpha^c \leftarrow u^c_\delta ) \m{then}  \\ &&& \qquad  \; \qquad (\lre{\mc{C}'_1\mc{T}_1} {\mathbf{msg}( \mb{send}x_\beta^c\,y_\alpha^c;y_\alpha^c \leftarrow u^c_\delta )\mc{D}_1} {\mc{F}_1},\\&&&\qquad\;\; \qquad \lre{\mc{C}'_2\mc{T}_2} {\mathbf{msg}( \mb{send}x_\beta^c\,y_\alpha^c ; y_\alpha^c \leftarrow u^c_\delta )\mc{D}_2} {\mc{F}_2})  \in  \mathcal{E}^\xi_{ \Psi}\llbracket \Delta,x_\beta{:}A[c], u_{\delta}{:}B[c] \Vdash K \rrbracket\\[4pt]
  
  (10) & (\lr{\mc{C}'_1\mc{C}''_1} {\mc{D}_1} {\mc{F}_1};\lr{\mc{C}'_2\mc{C}''_2} {\mc{D}_2} {\mc{F}_2})  \in  & \; &  (\lr{\mc{C}'_1\mc{C}''_1} {\mc{D}_1} {\mc{F}_1};\lr{\mc{C}'_2\mc{C}''_2} {\mc{D}_2} {\mc{F}_2}) \in \m{Tree}_\Psi(\Delta', \Delta'', y_\alpha{:}A\multimap B[c] \Vdash K)\,\m{and}\\ &\qquad \qquad  \mathcal{V}^\xi_{\Psi}\llbracket \Delta', \Delta'', y_\alpha{:}A\multimap B[c] \Vdash K \rrbracket&& \mathcal{D}_1=\mathcal{T}_1\mathbf{msg}( \mb{send}x_\beta^c\,y_\alpha^c; ;u_\delta^c \leftarrow y_\alpha^c)\;\mathcal{D}''_1\, \m{and}\\
  & &&\mathcal{D}_2=\mathcal{T}_2\mathbf{msg}( \mb{send}x_\beta^c\,y_\alpha^c; ;u_\delta^c \leftarrow y_\alpha^c)\,\mc{D}''_2\, \m{and} \\
    &&&  (\lre{\mc{C}''_1}{\mathcal{T}_1}{\mc{C}'_1\mathbf{msg}( \mb{send}x_\beta^c\,y_\alpha^c; u_\delta^c \leftarrow y_\alpha^c)\mc{D}''_1\mc{F}_1},\\ &&& \qquad \qquad  \lre{\mc{C}''_2}{\mathcal{T}_2}{\mc{C}'_2\mathbf{msg}( \mb{send}x_\beta^c\,y_\alpha^c; u_\delta^c \leftarrow y_\alpha^c) \mc{D}''_2 \mc{F}_2}) \in  \mathcal{E}^\xi_{ \Psi}\llbracket \Delta' \Vdash x_\beta{:}A[c] \rrbracket\,\m{and}\\
   &&&  (\lre{\mc{C}'_1\mc{C}''_1\mathcal{T}_1\mathbf{msg}( \mb{send}x_\beta^c\,y_\alpha^c; u_\delta^c \leftarrow y_\alpha^c)} {\mc{D}''_1} {\mc{F}_1},\\ &&& \qquad \qquad  \lre{\mc{C}'_2\mc{C}''_2\mathcal{T}_2\mathbf{msg}( \mb{send}x_\beta^c\,y_\alpha^c; u_\delta^c \leftarrow y_\alpha^c)} {\mc{D}''_2} {\mc{F}_2}) \in  \mathcal{E}^\xi_{ \Psi}\llbracket \Delta'', u_{\delta}{:}B[c] \Vdash K \rrbracket\\[4pt]

   (11)&  ({\lr{\mc{C}_1} {\mc{D}_1} {\mc{F}_1};{\lr{\mc{C}_2} {\mc{D}_2} {\mc{F}_2}}}) \in  & \;&  (\lr{\mc{C}_1} {\mc{D}_1} {\mc{F}_1};\lr{\mc{C}_2} {\mc{D}_2} {\mc{F}_2}) \in \m{Tree}_\Psi(\Delta', y_\alpha{:}A[c] \Vdash K)\,\m{and}\\
   &\qquad \qquad \mc{V}^\xi_{ \Psi}\llbracket  \Delta \Vdash y_\alpha{:}A[c] \rrbracket&& \mc{D}_1=\mc{D}'_1\mathbf{proc}(y_\alpha^c, y_\alpha^c\leftarrow x_\beta^c @d_1) \, \m{and}\,  \mc{D}_2=\mc{D}'_2\mathbf{proc}(y_\alpha^c, y_\alpha^c\leftarrow x_\beta^c @d_2)\,\m{and}\\&&&([x_{\beta}^c/y_\alpha^c]{\lre{\mc{C}_1} {\mc{D}'_1} {\mc{F}_1},[x_{\beta}^c/y_\alpha^c]{\lre{\mc{C}_2} {\mc{D}'_2} {\mc{F}_2}}}) \in \mathcal{E}^\xi_{ \Psi}\llbracket \Delta\Vdash x_{\beta}{:}A[c] \rrbracket\\[4pt]
    (12)&  ({\lr{\mc{C}_1} {\mc{D}_1} {\mc{F}_1};{\lr{\mc{C}_2} {\mc{D}_2} {\mc{F}_2}}}) \in  & \;& (\lr{\mc{C}_1} {\mc{D}_1} {\mc{F}_1};\lr{\mc{C}_2} {\mc{D}_2} {\mc{F}_2}) \in \m{Tree}_\Psi(\Delta', y_\alpha{:}A[c] \Vdash K)\,\m{and}\\
    &\qquad \qquad \mc{V}^\xi_{ \Psi}\llbracket \Delta, y_\alpha{:}A[c] \Vdash K \rrbracket&&    \m{if}\, (\mc{C}_1=\mc{C}'_1\mathbf{proc}(y_\alpha^c, y_\alpha^c\leftarrow x_\beta^c @d_1) \, \m{and}\,  \mc{C}_2=\mc{C}'_2\mathbf{proc}(y_\alpha^c, y_\alpha^c\leftarrow x_\beta^c @d_2))\, \m{then}\\ &&&\;\;([x_{\beta}^c/y_\alpha^c]{\lre{\mc{C}'_1} {\mc{D}_1} {\mc{F}_1},[x_{\beta}^c/y_\alpha^c]{\lre{\mc{C}'_2} {\mc{D}_2} {\mc{F}_2}}}) \in \mathcal{E}^\xi_{ \Psi}\llbracket \Delta, x_{\beta}{:}A[c] \Vdash K \rrbracket\\[4pt]
     (13) &  (\mc{B}_1, \mc{B}_2 ) \in \mc{V}^\xi_{ \Psi}\llbracket \cdot \Vdash \cdot \rrbracket &  \; &
         (\mc{B}_1, \mc{B}_2) \in \m{Tree}_{ \Psi }\llbracket \cdot \Vdash \cdot \rrbracket \\[4pt]  
  (14) & ({\mc{B}_1}, {\mc{B}_2}) \in \mc{E}^\xi_{ \Psi}\llbracket \Delta \Vdash K \rrbracket & \; & \mc{B}_1=\mc{C}_1\mc{D}_1\mc{F}_1 \,\m{and}\,\mc{B}_2=\mc{C}_2\mc{D}_2\mc{F}_2\,\m{and}\,(\lr{\mc{C}_1}{\mc{D}_1}{\mc{F}_1}; \lr{\mc{C}_2}{\mc{D}_2}{\mc{F}_2})\in \m{Tree}_{\Psi}(\Delta \Vdash K),\\
  &&& \m{and}\,\lr{\mc{C}_1}{\mc{D}_1}{\mc{F}_1}\mapsto^{\m{poised}}_{\Delta \Vdash K} \lr{\mc{C}'_1}{\mc{D}'_1}{\mc{F}'_1} \,\m{and}\,  \lr{\mc{C}_2}{\mc{D}_2}{\mc{F}_2}\mapsto^{\m{poised}}_{\Delta \Vdash K} \lr{\mc{C}'_2}{\mc{D}'_2}{\mc{F}'_2}\\
          &&& \m{and}\,(\lr{\mc{C}'_1}{\mc{D}'_1}{\mc{F}'_1}; \lr{\mc{C}'_2}{\mc{D}'_2}{\mc{F}'_2})\in \mc{V}^\xi_{\Psi}\llbracket  \Delta \Vdash K \rrbracket.\\[4pt]

(15) & (\lr{\mc{C}_1}{\mc{D}_1}{\mc{F}_1}; \lr{\mc{C}_2}{\mc{D}_2}{\mc{F}_2}) \in \m{Tree}_\Psi(\Delta \Vdash K) & \; &  \Psi; \cdot \Vdash \mc{C}_1 :: \Delta,\,\m{and} \, \Psi; \cdot  \Vdash \mc{C}_2 ::\Delta,\, \Psi; \Delta  \Vdash \mc{D}_1 :: K \,\m{and}\, \Psi; \Delta  \Vdash \mc{D}_2 ::K\\
&&&\m{and}\, \Psi; K \Vdash \mc{F}_1 :: \cdot \,\m{and}\, \Psi; K  \Vdash \mc{F}_2 ::\cdot

    \end{array}\]}
    \caption{Session Logical Relation for Noninterference: property (left), condition (right).}
    \label{fig:def_lr}
\end{figure*}

In this section we formalize the session logical relation for noninterference as explained in \secref{sec:key-ideas2}. We are interested in a property that asserts that two partial programs send the same messages along their observable channels if being closed with any well-typed configurations. The closing configurations are assumed to send the same messages along the observable channels. For this property to hold, the partial programs must agree on their set of observable channels and the closing configuration have to be well-typed a priori. 

For a partial program $\Psi;\Delta_1 \Vdash \mc{D}_1:: x_\alpha {:}A_1[{c_1}]$ we need two closing configurations. One to provide $\Delta_1$ without using any resources, \ie $\Psi; \cdot \Vdash \mc{C}_1:: \Delta_1$. The other to use $x_\alpha {:}A_1[{c_1}]$ as a resource and offer a terminating type, \ie $\Psi; x_\alpha {:}A_1[{c_1}]\Vdash \mc{F}_1:: y_\alpha{:}1[c']$. The name and secrecy of the channel  provided by $\mc{F}_1$ is not significant in our setting; $y_\alpha$ can only send a closing message when all observable channels are already closed. Thus we disregard it and alternatively write
$\Psi; x_\alpha {:}A_1[{c_1}]\Vdash \mc{F}_1:: \cdot$.
We keep in mind that a providing type $\cdot$ behaves as a terminating channel. In this paper, we often use $K:=x_\alpha[c]{:}A\mid \cdot$ for the providing channel to account for this notation. 

Our property of interest is formalized in \defref{def:noninterference}.

\begin{definition}[Equivalence up to Observable Messages]\label{def:noninterference}
$ {\color{red} (\Delta_1 \Vdash \mc{D}_1:: x_\alpha {:}A_1[{c_1}]) \equiv^\Psi_{\xi} (\Delta_2 \Vdash \mc{D}_2:: y_\beta {:}A_2[{c_2}])}$
is defined as 
$\Psi; \Delta_1 \Vdash \mc{D}_1:: x_\alpha {:}A_1[{c_1}]\; \mathit{and}\; \Psi; \Delta_2 \Vdash \mc{D}_2:: y_\beta {:}A_2[{c_2}]$ and
       $\Delta_1 {\Downarrow} \xi= \Delta_2 {\Downarrow} \xi=\Delta \; \mathsf{and}\;  x_\alpha {:}A_1[{c_1}]{\Downarrow} \xi= y_\beta {:}A_2[{c_2}] {\Downarrow} \xi= K\, \mathit{and}$
          for all $\mc{C}_1, \mc{C}_2, \mc{F}_1, \mc{F}_2,$
          with $\ctype{\Psi}{\cdot}{\mc{C}_1}{\Delta_1}$ and  $\ctype{\Psi}{\cdot}{\mc{C}_2}{\Delta_2}$ and
            $\ctype{\Psi}{x_\alpha {:}A_1[{c_1}]}{\mc{F}_1}{\cdot}$ and
            $\ctype{\Psi}{y_\beta {:}A_2[{c_2}]}{\mc{F}_2}{\cdot}$,  
we have \[{(\mc{C}_1\mc{D}_1\mc{F}_1, \mc{C}_2\mc{D}_2\mc{F}_2) \in \mc{E}_{\Psi}^\xi\llbracket \Delta \Vdash K \rrbracket.}\]
\end{definition}

The relation ${(\mc{B}_1, \mc{B}_2) \in \mc{E}_{\Psi}^\xi\llbracket \Delta \Vdash K \rrbracket}$ is defined in \figref{fig:def_lr}~\footnote{For ease of reference in our proofs, we annotate channel names appearing in process terms with their generations (subscript) and maximal secrecy (superscript).}, line 14. It is an apparatus to track the computation of two closed configurations $\mc{B}_1$ and $\mc{B}_2$ looking for messages being sent and received along their mutual set of observable channels $\Delta$ and $K$. The content of messages sent or received along other channels are not significant and disregarded. In particular, if the offering channels of the two partial programs are not observable, we dismiss them from consideration and put $K=\cdot$ as a placeholder in the relation. 

To track the observable messages using $\mc{E}^\xi_{\Psi}$, we need to know that $\mc{B}_i$ can be broken down into  $\lre{\mc{C}_i}{\mc{D}_i}{\mc{F}_i}$ such that $\ctype{\Psi}{\cdot}{\mc{C}_i}{\Delta}$, and $\ctype{\Psi}{\Delta}{\mc{D}_i}{K}$, and $\ctype{\Psi}{K}{\mc{F}_i}{\cdot}$.  We prove that this property holds for any $\mc{B}_i$ that is built by closing a partial program with observable channels $\Delta$ on the left and $K$ on the right.  The interested reader can refer to Lemma~4 and Figure~1 in the appendix for further details.  The key idea is to internalize any trees rooted at non-observable channels in the bottom closing configuration and the closing non-observable tree constituting the top closing configuration.

After decomposing the configurations $\mc{B}_1$ and $\mc{B}_2$ into $\lre{\mc{C}_1}{\mc{D}_1}{\mc{F}_1}$ and $\lre{\mc{C}_2}{\mc{D}_2}{\mc{F}_2}$, respectively, we compute each subconfiguration separately.  We write $\lr{\mc{C}}{\mc{D}}{\mc{F}}\mapsto_{\Delta \Vdash K}\lr{\mc{C}'}{\mc{D}'}{\mc{F}'}$ if (i) ${\mc{C}}\mapsto_{\cdot \Vdash \Delta}{\mc{C}'}$, (ii) ${\mc{D}}\mapsto_{\Delta \Vdash K}{\mc{D}'}$, and (iii) ${\mc{F}}\mapsto_{K \Vdash \cdot}{\mc{F}'}$.
We are interested in the state in which none of the subconfigurations can proceed  without communicating along an observable channel. This state is closely related to the property of being poised introduced in \secref{sec:type-system}. We call $\lr{\mc{C}_i}{\mc{D}_i}{\mc{F}_i}$ poised if its subconfigurations $\mc{C}_i$, and $\mc{D}_i$, and $\mc{F}_i$ are poised. We write $\lr{\mc{C}_i}{\mc{D}_i}{\mc{F}_i} \mapsto^{\m{poised}}_{\Delta \Vdash K}\lr{\mc{C}'_i}{\mc{D}'_i}{\mc{F}'_i}$ stating that $\lr{\mc{C}_i}{\mc{D}_i}{\mc{F}_i} \mapsto^*_{\Delta \Vdash K}\lr{\mc{C}'_i}{\mc{D}'_i}{\mc{F}'_i}$ and $\lr{\mc{C}'_i}{\mc{D}'_i}{\mc{F}'_i}$ is poised.
($\mapsto^*_{\Delta \Vdash K}$ refers to zero or more steps taken with $\mapsto_{\Delta \Vdash K}$.)

To relate two poised configurations we use the value relation $\mc{V}^\xi_{ \Psi}\llbracket \Delta \Vdash K\rrbracket$. This relation \emph{establishes} equality of the content of every message fired from $\mc{D}_1$ and $\mc{D}_2$ before adding them to the closing configurations $\mc{C}_i$ and $\mc{F}_i$ (See \figref{fig:logical_relation}).  In the case of sending higher order channels (lines 4 and 10 in \figref{fig:def_lr}), we further \emph{assure} that the trees sent along the messages are also related and will behave similarly when received by the closing configuration. 


For the messages being fired from the poised closing configurations $\mc{C}_i$ and $\mc{F}_i$, we \emph{assume} that they have the same content ready to be moved to $\mc{D}_1$ and $\mc{D}_2$. In particular for higher order channels (lines 5 and 9 in \figref{fig:def_lr}) we \emph{assume} the channels sent by the closing configurations will send the same observable messages too. We add the received channels to the set of observable channels to make sure that the partial programs do not send them different messages. 

A forwarding process does not send or receive an explicit message. However, when process $\mb{proc}(y[c], y_\alpha \leftarrow x_\beta)$ forwards  channel ($x_\beta$) to an observable channel ($y_\alpha$) the substitution of $x_\beta$ for $y_\alpha$ amounts to a broadcast of the name of $x_\beta$ along the observable channel $y_\alpha$. The channel $x_\beta$ has a secrecy level lower than or equal to the observer and now can be observed too. In our relation (line 11 of \figref{fig:def_lr})  we \emph{assert} that such forwarding rules in $\mc{D}_1$ and $\mc{D}_2$ always  broadcast the same names. In the dual case (line 12 of \figref{fig:def_lr}) we \emph{assume} that the closing configurations $\mc{C}_1$ and $\mc{C}_2$ broadcast the same names. In both cases we continue by monitoring the rest of the computation along $x_\beta$. The same holds for a forwarding on the tail of a message.

The well-foundedness of our logical relation is based on a lexicographic order on the structure of observable types and a multiset order $<$ on the size of configurations.  
The order $<$ is a multiset order on finite multiset $\mc{M}$ of the process typing judgments $\ptype{\Psi}{\Delta''}{P}{y}{d}{\gamma}{A}$ used in the typing derivation of $\mc{C}$.  Process typing judgments are ordered based on the size of the process term. As a result, well-foundedness of $<$ follows from the well-foundedness of process terms~\cite{Pierre1986}.




%% file: metatheory.tex
\section{Metatheory}\label{sec:metatheory}

In this section we establish the main properties of the system. We show that $\lang$ is a terminating language with the standard preservation and progress properties. More importantly, we prove that it enjoys the noninterference property.

 \begin{theorem}[Preservation]\label{thm:preservation}
If $\Psi; \Delta \Vdash \mc{C} :: \Delta'$ and $\mc{C}\mapsto_{\Delta \Vdash \Delta'} \mc{C}'$, then  $\Psi; \Delta \Vdash \mc{C}' :: \Delta'$. Moreover $\mc{C}'<\mc{C}$ by the multiset ordering.
\end{theorem}
\begin{proof}
 The proof is by case analysis of $\mc{C}\mapsto_{\Delta \Vdash \Delta'} \mc{C}'$ and inversion on the typing judgment $\Psi; \Delta \Vdash \mc{C} :: \Delta'$. See the appendix for more details.
\end{proof}
 
 \begin{theorem}[Progress]\label{thm:progress}
If $\Psi; \Delta \Vdash \mc{C}::\Delta'$, then 
either $\mc{C} \mapsto_{\Delta\Vdash \Delta'}\mc{C}'$ or $\mc{C}$ is poised.
\end{theorem}
\begin{proof}
The proof is by induction on the configuration typing of $\mc{C}$. See the appendix for the complete proof.
\end{proof}

Termination of $\lang$ follows from \thmref{thm:progress} and \thmref{thm:preservation} and well-foundedness of the $<$ order.
 
The fundamental property of our logical relation is noninterference stated as below. 
\begin{theorem}[Noninterference]
For all security levels $\xi$ and configurations $\Psi; \Delta \Vdash \mathcal{D}:: x_\alpha {:}T[c]$, we have \[ (\Delta \Vdash \mathcal{D}:: x_\alpha {:}T[c])  \equiv^\Psi_{\xi} (\Delta \Vdash \mathcal{D}:: x_\alpha {:}T[c]).\]  
\end{theorem}

Our noninterference theorem asserts that a well-typed open configuration $\mc{D}$ is equivalent to itself. It states that if we run a program twice but with different closing configurations, the contents of messages sent by the program along the observable channels will be the same in both runs. The assertion is based on the assumption that the closing configurations send the same messages along the observable channels in both runs.

The two runs start out as $\mc{D}$, guaranteeing that their tree structure is identical and their processes are running the same code.  The two runs continue to be identical until a process in each run receives a message from a closing configuration along a non-observable channel.  The received messages may differ in contents because the carrier channel's maximal secrecy is higher than or incomparable to $\xi$.  Based on the contents of the received message, the two runs may choose different continuations, after which they begin to diverge in their tree structure and the code the individual processes are running.   On the other hand, the running secrecy of the receiving processes will be adjusted upon receiving to become higher or incomparable to $\xi$.  This adjustment makes sure that the receiving processes can no longer send any messages along channels of lower or equal secrecy than the observer. In particular, they cannot send a message along an observable channel. We call such processes that can no longer affect any observable messages \emph{irrelevant}.


Throughout the computation, the code and structure of some processes may diverge as they receive non-observable messages. However, the \emph{relevant} processes, \ie the processes that can affect the contents of observable messages, stay identical. Later in this section, we state the fundamental theorem  (\thmref{thm:main}) that proves  two configurations to be equivalent up to observable messages if their relevant processes are identical. The noninterference property is then an immediate corollary of the fundamental theorem.


Before stating the fundamental theorem, we need to define the notion of a relevant process. We discussed earlier that a process with running secrecy higher or incomparable to the observer's secrecy level is irrelevant.  An irrelevant process can no longer spawn any observable messages. Moreover, if it sends a message along a non-observable channel, then the receiver becomes irrelevant too. There is another form of irrelevant process with running secrecy less than or equal to $\xi$ but with paths to observable channels passing through channels with maximal secrecy level higher than or incomparable to $\xi$.  These channels block the flow of information because any process receiving along such a channel becomes irrelevant.


To  establish  a  sound  definition  of  relevant  processes  and messages  for  the  asynchronous  semantics,  we  need  a  lookahead  for  the  running  secrecy.   Consider a process $ \mb{case}\, y^{c}(\cdots)@d_1$ in the partial program and its counterpart $\mb{case}\, y^c(\cdots)@d_1$ in the other run. They both have  running  secrecy $d_1\sqsubseteq \xi$, and are ready to receive a label along a non-observable channel $y[c]$. By the previous discussion, right after receiving a label the two processes become irrelevant. The non-observable messages may not be ready at the same time. For example,  the process in the first run may receive the message right away and become irrelevant, while the other process may need to wait for a while. This results in a discrepancy between relevant processes in the two runs. However, these processes cannot affect any observable channels even before they receive a channel. Based on their code they can only receive in the current step and right after the receive they become irrelevant. To account for delays in the receives, we label these two processes as irrelevant even before they receive, using a lookahead called \emph{quasi running secrecy}. The quasi running secrecy of a receiving process is defined as its running secrecy at the next step, \ie right after the receive.


We determine the running secrecy of a message to be the maximal secrecy of the channel that the message is sent along. However, messages are only temporary holders of a label or tree that they transport. Unless a message is observable, \ie sent to a closing configuration, its contents can only affect an observable channel \emph{after} it is received by a process. The quasi running secrecy of a message accounts for this and reflects the future potential of a message once it is received and is determined by examining the running secrecy of the recipient.  In case of a negative message (see \figref{fig:dynamics_idea}), the receiver is a child of the message. By the tree invariant, the running secrecy of the child is less than or equal to the maximal secrecy of the carrier channel. After receiving the message the running secrecy of the receiver will be equal to the maximal secrecy of the carrier channel. As a result, the quasi running secrecy of a negative message amounts to the maximal secrecy of the channel along which the message is sent.  In case of a positive message (see \figref{fig:dynamics_idea}), the receiver is the parent of the message. The running secrecy $d_1$ of the parent may be higher or incomparable to the maximal secrecy $c$ of the carrier channel. After the message is received, the running secrecy of the parent is adjusted to at least $c\sqcup d_1$. As a result, we determine the quasi running secrecy of a positive message to be the running secrecy of its parent after the message has been received $(c\sqcup d_1)$.

 

The notions of quasi running secrecy and relevancy are formally defined in \defref{def:quasi} and \defref{def:relevant}

\begin{definition}[Quasi Running Secrecy]\label{def:quasi}
In the configuration tree, the quasi running secrecy of a message or process is determined based on its running secrecy, its process term, and the running secrecy of its parent. 
\begin{itemize}
    \item If the node is a process with a process term other than $\mb{recv}$ or $\mb{case}$, then its quasi running secrecy is equal to its running secrecy.
    \item If the process term is of the form $\mb{case}\,y^c_{\alpha}(\ell \Rightarrow P_\ell)_{\ell \in L} @d_1$ or $x^\psi \leftarrow \mb{recv}\,y^c_{\alpha};P_{x} @d_1$, then its quasi running secrecy is $d_1 \sqcup c$.
    \item If the node is a message of a negative type along channel $y^c_\alpha$, its quasi running secrecy is $c$.
    \item  If the node is a message of a positive type along channel $y^c_\alpha$ and it has a parent with quasi running secrecy $d_1$, its quasi running secrecy is $d_1 \sqcup c$.
\end{itemize}
The quasi running secrecy can be determined by traversing the tree top to bottom.

\end{definition}

\begin{definition}[Relevant Channels and Processes]\label{def:relevant}
Consider configuration $\Delta \Vdash \mc{D}:: K$ and observer level $\xi$. A channel is relevant in $\mc{D}$ if 1) it is has a maximal secrecy level lower than or equal to $\xi$, and 2) it is either an observable channel or it shares a process or message with quasi running secrecy less than $\xi$ with a relevant channel. (A channel shares a process with another channel if they are siblings or one is the parent of another.)

The set of all relevant channels can be found by traversing the tree bottom-up. If $K$ is observable, then by the tree invariant, every channel in $\mc{D}$ will be relevant.

A relevant process or message has quasi running secrecy less than or equal to $\xi$ and at least one relevant channel. $\cproj{\mc{C}}{\xi}$ are the relevant processes and messages in $\mc{C}$. We write $\cproj{\mc{C}_1}{\xi}=_{\xi}\cproj{\mc{C}_2}{\xi}$ if they are identical up to renaming of channels with higher or incomparable secrecy than the observer.
\end{definition}

The fundamental theorem is stated as below.

\begin{theorem}[Fundamental Theorem]\label{thm:main}
For all security levels $\xi$, and configurations ${\Psi; \Delta_1 \Vdash \mathcal{D}_1:: u_\alpha {:}A_1[c_1]}$ and ${\Psi; \Delta_2 \Vdash \mathcal{D}_2:: v_\beta {:}A_2[c_2]}$  with $\cproj{\mathcal{D}_1}{\xi} = \cproj{\mathcal{D}_2}{\xi}$, $\Delta_1 \Downarrow \xi = \Delta_2 \Downarrow \xi$, and $u_\alpha {:}A_1[c_1]\Downarrow \xi = v_\beta {:}A_2[c_2] \Downarrow \xi$  we have \[ (\Delta_1 \Vdash \mathcal{D}_1:: u_\alpha {:}A_1[c_1])  \equiv^\Psi_{\xi} (\Delta_2 \Vdash \mathcal{D}_2:: v_\beta {:}A_2[c_2]).\]  
\end{theorem}
\begin{proof}
The proof is by induction on the type structure and the multiset ordering. For the details of the proof see the appendix.
\end{proof}

To prove that our fundamental theorem entails the desired property, we define an alternative stepping definition $\hookrightarrow_{\Delta \Vdash K}$ in Figure 1 in the appendix for a closed configuration $\lr{\mc{C}}{\mc{D}}{\mc{F}}\in \m{Tree}(\Delta \Vdash K)$. Where  $\lr{\mc{C}}{\mc{D}}{\mc{F}}\in \m{Tree}(\Delta \Vdash K)$ is defined as
$\cdot \Vdash \mc{C} :: \Delta $, and $\Delta \Vdash \mc{D} :: K $, and $K \Vdash \mc{F} :: y{:}1[c]$. The idea is to run this closed configuration to completion, while accumulating the messages exchanged between the partial program $\mc{D}$ and closing configurations $\mc{C}$ and $\mc{F}$ in a queue. 


It is straightforward to show that $\mc{C}\mc{D}\mc{F}\mapsto_{\cdot \Vdash \cdot}^* \msg{\mb{close}\_}$ if and only if for some $\m{queue}$, we have $\lr{\mc{C}}{\mc{D}}{\mc{F}}\hookrightarrow_{\Delta \Vdash K}\m{queue}$, where $\m{queue}$ is the list of observable messages being exchanged between $\mc{D}$ and the closing configurations $\mc{C}$ and $\mc{F}$ along $\Delta$ and $K$.  An overline indicates that a message is sent from $\mc{C}$ or $\mc{F}$ to $\mc{D}$, otherwise the message is sent from $\mc{D}$ to $\mc{C}$ or $\mc{F}$.

\begin{definition}
Define $\m{queue}_1 \meq \m{queue}_2$ as either
\begin{itemize}
\item $\m{queue}_1= q_1\, \m{queue}'_1$, and $\m{queue}_2= q_2\, \m{queue}'_2$, and $q_1=q_2$, and $\m{queue}'_1 \meq \m{queue}'_2$, or

 \item $\m{queue}_1= \overline{q_1}\, \m{queue}'_1$, and $\m{queue}_2= \overline{q_2}\, \m{queue}'_2$, and if $q_1=q_2$ then $\m{queue}'_1 \meq \m{queue}'_2$.
\end{itemize}

\end{definition}

\begin{theorem}\label{thm:soundness}
For  {\small$(\lr{\mc{C}_1}{\mc{D}_1}{\mc{F}_1}; \lr{\mc{C}_2}{\mc{D}_2}{\mc{F}_2})\in \m{Tree}_{\Psi}(\Delta \Vdash K)$,}
if {\small$(\lre{\mathcal{C}_1}{\mathcal{D}_1} {\mathcal{F}_1};\lre{\mathcal{C}_2}{\mathcal{D}_2} {\mathcal{F}_2}) \in \mathcal{E}^\xi_\Psi\llbracket \Delta \Vdash K \rrbracket,$} then {\small \[(\lr{\mathcal{C}_1}{\mathcal{D}_1} {\mathcal{F}_1}) \hookrightarrow_{\Delta \Vdash K} \m{queue}_1\,\m{and}\,\,(\lr{\mathcal{C}_2}{\mathcal{D}_2} {\mathcal{F}_2}) \hookrightarrow_{\Delta \Vdash K} \m{queue}_2\]}
such that {\small$\m{queue}_1 \meq \m{queue}_2$}.
\end{theorem}
\begin{proof}
The proof is straightforward by matching the cases in the definition of $\mc{E}$ with the cases in the definition of $\hookrightarrow_{\Delta\Vdash K}$.
\end{proof}

%% file: related.tex
\section{Related and Future Work}\label{sec:related}

Related work can be categorized along the following axes:

\paragraph{Information Flow Control Type Systems for Functional/Imperative Languages}
Following Volpano et al.'s seminal work~\cite{volpano96}, much work has been done in the field of information flow type systems for sequential programs (c.f.~\cite{ifc-survey}). Our noninterference definition is inspired by Bowman et al.'s work on Noninterference and parametricity~\cite{BowmanIAhmedCFP2015}. 

\paragraph{Information Flow Control for Process Calculi}
Projects on enforcing IFC on process calculi using security types share a similar goal as ours: to prevent information leakage in process communications~\cite{HondaESOP2000,HondaYoshidaPOPL2002,CrafaARTICLE2002,CrafaTGC2005,CrafaFMSE2006,Crafa2007,HENNESSYRIELY2002,HENNESSY20053,KOBAYASHI2005,ZDANCEWIC2003,POTTIER2002}. Similar to our work, a security label is typically added to types. Much of these work associate the security labels to channels. Yoshida et al. associate the labels to actions~\cite{HondaYoshidaPOPL2002}; Hennessy and Riely associate read and write policies to channels~\cite{HENNESSYRIELY2002,HENNESSY20053}; and Crafa et al. associate a security label to the process and capabilities to expressions~\cite{CrafaARTICLE2002}. 
We associate two security labels with the process. One main difference that sets our system
apart from prior work is that ours is a flow-sensitive system: the running secrecy changes as
the process receives more information. The running secrecy and flow sensitivity is a natural consequence of building our work on the sequent calculus. 
Some of the existing work also consider declassification~\cite{CrafaFMSE2006,Crafa2007}, which we leave as future work.

Timing channels and race conditions can contribute to information leakage. Unlike prior work~\cite{KOBAYASHI2005,ZDANCEWIC2003}, 
our linear types ensure progress, termination, and freedom of race conditions; and therefore do not need additional checks to rule out such leaks.  
Prior work also proposed different noninterference definitions, relying on barbed-congruence, P-congruence,  may-testing and must-testing, per-models, and trace equivalence. Our noninterference definition is based on a novel binary session logical relation. It is closest to barbed-congruence definitions and entails trace equivalences. Since our processes' behavior is finite, we do not need co-inductive definitions.


\paragraph{Information Flow Control for Multiparty Session Types}
Only recently, have researchers investigated incorporating information flow security into session
types~\cite{CapecchiCONCUR2010,CapecchiARTICLE2014,CastellaniARTICLE2016,Ciancaglini2016}. In
addition to developing information flow session type systems that allow
declassification~\cite{CapecchiCONCUR2010,CapecchiARTICLE2014}, researchers also designed
flexible run-time monitoring techniques for preventing information
leakage~\cite{CastellaniARTICLE2016,Ciancaglini2016}, all in the context of multiparty session
types. Ours is the first information flow binary session type system.  Again, our flow-sensitive type system and logical relation-based definition for noninterference sets us apart from existing work.

\paragraph{Hybrid Logic Modal Worlds in Session Types}

Our typing judgment includes world modalities from hybrid logic as syntactic objects in
propositions, where worlds amount to secrecy levels.  A hybrid logic approach has been used in prior
work on binary session types to ensure deadlock-freedom of shared binary session
types~\cite{BalzerESOP2019} and accessibility in linear binary session
types~\cite{CairesCONCUR2019}.  Our work differs not only in the established property of
interest (noninterference) but also in the use of a novel binary relation for session types.

\paragraph{Logical Relations for Session Types}

The use of logical relations for session types has focused predominantly on unary logical
relations (predicates) for proving termination~\cite{PerezESOP2012,
  PerezARTICLE2014,DeYoungFSCD2020} with the exception of a binary logical relation for
parametricity~\cite{CairesESOP2013}.  Noninterference, however, demands a more nuanced binary
relation, requiring communication to be perceived in either direction of the channel.  Our work
generalizes binary logical relations for session typed languages to support \emph{open
  configurations}, considering both the antecedent and succedent of the typing judgment.  While
we have defined the logical relation for noninterference, we believe that the technical
developments in this paper can serve as a stepping stone for future explorations.

\paragraph{Kripke Logical Relations}

Conceptually, our work seems related to Kripke logical
relations~\cite{PittsStarkHOOTS1998} and in particular the works that use possible
worlds~\cite{AhmedPOPL2009} and state machines~\cite{DreyerPOPL2010} to impose invariants on
program heaps.  In our setting, the program heap is a configuration of processes.  Session
types constrain how the configuration can evolve, and configuration typing asserts that
configurations align with the security lattice.  It seems that our secrecy-level-enriched
session types internalize Kripke logical worlds into the type system.  We would like to explore
this connection in future work.

%% file: appendix.tex



\section{Appendix}
\begin{definition}[Projections]
Projection for linear context is defined as follows:
\[\begin{array}{lclc}
\Delta, x_\alpha{:}T[c] \Downarrow \xi& \defeq& \Delta \Downarrow \xi,  x_\alpha{:}T[c] & \m{if}\; c\sqsubseteq \xi  \\
     \Delta, x_\alpha{:}T[c] \Downarrow \xi& \defeq& \Delta \Downarrow \xi & \m{if}\; c \not \sqsubseteq \xi  \\
    \cdot \Downarrow \xi& \defeq& \cdot &  \\
     x_\alpha{:}T[c] \Downarrow \xi& \defeq&   x_\alpha{:}T[c] & \m{if}\; c\sqsubseteq \xi  \\
    x_\alpha{:}T[c] \Downarrow \xi& \defeq& \cdot & \m{if}\; c\not\sqsubseteq \xi \\
\end{array}\]
\end{definition}

\begin{lemma}\label{lem:inv}
 If $\ctype{\Psi}{\Delta}{\mc{C}\,\mc{C}'}{\Delta'}$, then for some $\Delta_1$ we have $\ctype{\Psi}{\Lambda_1}{\mc{C}}{\Lambda'_1, \Delta_1}$ and
 $\ctype{\Psi}{\Lambda_2,\Delta_1}{\mc{C}'}{\Lambda'_2}$, where $\Delta=\Lambda_1, \Lambda_2$ and $\Delta'=\Lambda'_1, \Lambda'_2$.
\end{lemma}
\begin{proof}
The proof is by a straightforward induction on the configuration typing rules.
\end{proof}

\begin{lemma}[Permutation of Configurations]
Writing $Q\langle x\rangle$ for a process term $Q$ with an occurrence of channel $x$, and $M\langle x \rangle$ for a message $M$ that sends a message along $x$ the following permutations are admissible:
\begin{itemize}
    \item  If $\Psi; \Delta \Vdash \mc{C}_1 \mb{proc}(x,Q\langle x\rangle)\mc{C}_2\mb{proc}(y,P\langle x \rangle)\mc{C}_3:: \Delta'$ then $\Psi; \Delta \Vdash \mc{C}_1\mc{C}_2  \mb{proc}(x,Q\langle x\rangle) \mb{proc}(y,P\langle x \rangle)\mc{C}_3:: \Delta'$.
    
    \item  If $\Psi; \Delta \Vdash \mc{C}_1 \mb{msg}(M\langle x \rangle)\mc{C}_2\mb{proc}(y,P\langle x \rangle)\mc{C}_3:: \Delta'$ then $\Psi; \Delta \Vdash \mc{C}_1\mc{C}_2  \mb{msg}(M\langle x \rangle) \mb{proc}(y,P\langle x \rangle)\mc{C}_3:: \Delta'$.
    
    \item If $\Psi; \Delta \Vdash \mc{C}_1 \mb{proc}(x,Q\langle x\rangle)\mc{C}_2\mb{msg}(M\langle x\rangle)\mc{C}_3:: \Delta'$ then $\Psi; \Delta \Vdash \mc{C}_1\mc{C}_2  \mb{proc}(x,Q\langle x\rangle) \mb{msg}(M\langle x\rangle)\mc{C}_3:: \Delta'$.
\end{itemize}
 
\end{lemma}

\begin{proof}
    See~[5] for the proof.
\end{proof}

\begin{definition}[Poised Configuration]
A configuration $\Delta_1, \Delta_2 \Vdash \mc{C}_1, \mc{C}_2:: \Lambda, w{:}A'[c]$ is poised iff either $\mc{C}_1, \mc{C}_2$ is empty or $\Delta_1 \Vdash \mc{C}_1:: \Lambda$ is poised and $ \Delta_2 \Vdash \mc{C}_2::  w{:}A'[c]$ is poised. The configuration $\Delta_2 \Vdash \mc{C}_2 :: w{:}A'[c]$ is poised iff it cannot take any steps and at least one of the following conditions hold:
\begin{itemize}
\item $\mc{C}_2$ is an empty configuration.
\item $\mathcal{C}_2= \mc{C}'_2\mb{msg}(P)\mc{C}''_2$ such that $\mb{msg}(P)$ is a negative message along $y \in \Delta_2$, i.e.
\[\begin{array}{l}
   y{:}\&\{\ell{:}A_\ell\}_{\ell \in L}[c_1]\Vdash \mb{msg}(P):: x{:}A_k[c_1], \; \mbox{or}  \\
     y{:}A\multimap B[c_1], z{:}A[c_1] \Vdash \mb{msg}(P):: x{:}B[c_1],
\end{array}\]
and both subconfigurations $\mc{C}_2'$ and $\mc{C}_2''$ are poised. 
\item  $\mathcal{C}_2= \mb{proc}(x[c'],P@d_1)\,\mc{C}_2'$ such that $\mb{proc}(x[c'], P@d_1)$ attempts to receive along a channel $y {\in} \Delta_2$, i.e.
\[\begin{array}{l}
     \Delta'_2, y{:}\oplus\{\ell{:}A_\ell\}_{\ell\in L}[c_1]\Vdash \mb{proc}(x[c'], \mb{case}y^{c_1}(\ell \Rightarrow P'_\ell)@d_1):: x{:}T[c'],\; \mbox{or}\\ \Delta'_2, y{:}A\otimes B[c_1] \Vdash \mb{proc}(x[c'], v \leftarrow \mb{recv}y^c; P' @d_1):: x{:}T[c'], \; \mbox{or}  \\
     \Delta'_2, y{:}1[c_1] \Vdash \mb{proc}(x[c'],\mb{wait} y^c;P' @d_1):: x{:}T[c'].
\end{array}\]

\item $\mathcal{C}_2= \mc{C}_2'\mb{msg}(P)$ such that $\mb{msg}(P)$ is a positive message along $w{:}A'[c]$, i.e.
\[\begin{array}{l}
 x{:}A_k[c] \Vdash \mb{msg}(P)::w{:}\oplus\{\ell{:}A_\ell\}_{\ell \in L}[c], \; \mbox{or}\\
 x{:}B[c], z{:}A[c] \Vdash \mb{msg}(P):: w{:}A\otimes B[c],\; \mbox{or}\\ 
 \cdot \Vdash \mb{msg}(P):: w{:}1[c],
\end{array}\]
and subconfiguration $\mc{C}'_2$  is poised. 
\item   $\mathcal{C}_2= \mb{proc}(w[c], P@d_1)\,\mc{C}_2'$ such that $\mb{proc}(w[c], P@d_1)$ attempts to receive along $w{:}A'[c]$, i.e. 
\[\begin{array}{l}
\Lambda \Vdash \mb{proc}(w[c], \mb{case}w^c(\ell \Rightarrow P'_\ell)@d_1):: w{:}\&\{\ell{:}A_\ell\}_{\ell \in L}[c], \; \mbox{or}\\ \Lambda \Vdash \mb{proc}(w[c], v \leftarrow \mb{recv}w^c; P' @d_1):: w{:}A\multimap B[c].
\end{array}\]
\item $\mathcal{C}_2= \mc{C}_2'\mb{proc}(w^c \leftarrow x^c @d_1)\,\mc{C}_2''$.
    
\end{itemize}
\end{definition}

\begin{theorem}[Preservation]\label{thm:preservation}
If $\Psi; \Delta \Vdash \mc{C} :: \Delta'$ and $\mc{C}\mapsto_{\Delta \Vdash \Delta'} \mc{C}'$, then  $\Psi; \Delta \Vdash \mc{C}' :: \Delta'$. Moreover $\mc{C}'<\mc{C}$ by the multiset ordering.
\end{theorem}
\begin{proof}
The proof is by considering different cases of $\mc{C}\mapsto_{\Delta \Vdash \Delta'} \mc{C}'$. And then by inversion on the typing derivations. We only consider a couple of interesting cases here. The proof of other cases is similar.
\begin{description}
\item {\bf Case 1. ($\m{Cut}$)} \[\mc{C}_1 \mb{proc}(y_\alpha[c],  (x^d \leftarrow P_{x^d})@d_2 ; Q_{x^d}\,@d_1) \mc{C}_2  \mapsto_{\Delta \Vdash \Delta'}  \mc{C}_1 \mb{proc}(x_0[d],  P_{x_0^d}@d_2)\mb{proc}(y_\alpha[c],  Q_{x_0^d}\,@d_1) \mc{C}_2.\]

{\bf By assumption of the theorem:} $\Psi; \Delta \Vdash \mc{C}_1 \mb{proc}(y_\alpha[c],  (x^d \leftarrow P_{x^d})@d_2 ; Q_{x^d}\,@d_1) \mc{C}_2:: \Delta'$. 

{\bf By \lemref{lem:inv}:} $\Psi; \Lambda_1 \Vdash \mc{C}_1:: \Delta_1, \Delta_2, \Lambda'_1$ and $\Psi; \Delta_1, \Lambda_2 \Vdash \mb{proc}(y_\alpha[c],  (x^d \leftarrow P_{x^d})@d_2 ; Q_{x^d}\,@d_1) :: y_\alpha{:}A[c]$. If $y_\alpha{:}A[c] \not \in \Delta'$ then $\Psi; \Lambda_3, \Delta_2, y_\alpha{:}A[c] \Vdash \mc{C}_2 :: \Lambda'_2$ and otherwise $\Psi; \Lambda_3, \Delta_2 \Vdash \mc{C}_2 :: \Lambda'_2$.
Where $\Delta=\Lambda_1, \Lambda_2, \Lambda_3$ and $\Delta'=\Lambda'_1, \Lambda'_2$ .

{\bf By inversion on $\mb{proc}$ rule:}
$\Psi; \Delta_1, \Lambda_2 \vdash  ((x^d \leftarrow P_{x^d})@d_2) ; Q_{x^d}\,@d_1 :: y_\alpha{:}A[c]$. Moreover, $(\star)\; \; \Psi \Vdash d_1 \sqsubseteq c$ and $\forall u_\gamma{:}T[d'] \in \Delta_1, \Lambda_2. \, \Psi \Vdash d' \sqsubseteq c$. 

{\bf By inversion on $\m{Cut}$ rule:} $\Psi; \Delta'_1, \Lambda'_2 \vdash  P_{x_0^d}@d_2:: x_0{:}B[d]$ and $\Psi; \Delta''_1, \Lambda''_2, x_0{:}B[d] \Vdash   Q_{x_0^d}\,@d_1 :: y_\alpha{:}A[c]$ where $\Delta_1=\Delta'_1,\Delta''_1$ and $\Lambda_2=\Lambda'_2,\Lambda''_2$. Moreover, $(\star')\;\;\Psi \Vdash d_1 \sqsubseteq d_2\sqsubseteq d \sqsubseteq c$ and $\forall u_\gamma{:}T[d'] \in \Delta'_1, \Lambda'_2. \, \Psi \Vdash d' \sqsubseteq d$.

{\bf By $\mb{proc}$ rule, $(\star)$ and $(\star')$:} $\Psi; \Delta'_1, \Lambda'_2 \Vdash \mb{proc}(x_0[d],  P_{x_0^d}@d_2):: x_0{:}B[d]$ and $\Psi; \Delta''_1, \Lambda''_2,x_0{:}B[d] \Vdash \mb{proc}(y_\alpha[c],  Q_{x_0^d}\,@d_1) :: y_\alpha{:}A[c]$ where $\Delta_1=\Delta'_1,\Delta''_1$ and $\Lambda_2=\Lambda'_2,\Lambda''_2$.

{\bf By configuration typing rules:}
$\Psi; \Delta \Vdash \mc{C}_1 \mb{proc}(x_0[d],  P_{x_0^d}@d_2) \mb{proc}(y_\alpha[c],  Q_{x_0^d}\,@d_1) \mc{C}_2:: \Delta'$. \\

Moreover, after taking this step,  $\Psi; \Delta_1, \Lambda_2 \vdash  ((x^d \leftarrow P_{x^d})@d_2) ; Q_{x^d}\,@d_1 :: y_\alpha{:}A[c]$ is replaced in the configuration by two smaller typing judgments $\Psi; \Delta'_1, \Lambda'_2 \vdash  P_{x_0^d}@d_2:: x_0{:}B[d]$ and $\Psi; \Delta''_1, \Lambda''_2, x_0{:}B[d] \vdash   Q_{x_0^d}\,@d_1 :: y_\alpha{:}A[c]$. By the definition of multiset ordering, we have \[\mc{C}_1 \mb{proc}(x_0[d],  P_{x_0^d}@d_2) \mb{proc}(y_\alpha[c],  Q_{x_0^d}\,@d_1) \mc{C}_2\,<\,\mc{C}_1 \mb{proc}(y_\alpha[c],  (x^d \leftarrow P_{x^d})@d_2 ; Q_{x^d}\,@d_1) \mc{C}_2.\]

\item {\bf Case 2. ($\otimes$)} {\small \[\mc{C}_1 \mb{msg}(\mathbf{send}\,x_\beta^{c}\,y_\alpha^{c};y_{\alpha}^{c}\leftarrow v_\delta^{c}) \mc{C}' \mb{proc}(u_\gamma[c'],w_\eta\leftarrow \mathbf{recv}\,y_\alpha^{c}; P\,@d_1) \mc{C}_2 \mapsto_{\Delta \Vdash \Delta'}    \mc{C}_1\mc{C}'\mb{proc}(u_{\gamma}[c'], [x_\beta^c/w_\eta][v^c_{\delta}/y^c_{\alpha}] P\,@d_1\sqcup c) \mc{C}_2.\]}

{\bf By assumption of the theorem:} $\ctype{\Psi}{\Delta}{\mc{C}_1 \mb{msg}(\mathbf{send}\,x_\beta^{c}\,y_\alpha^{c};y_{\alpha}^{c}\leftarrow v_\delta^{c}) \mc{C}' \mb{proc}(u_\gamma^{c'},w_\eta\leftarrow \mathbf{recv}\,y_\alpha^{c}; P\,@d_1) \mc{C}_2}{\Delta'}$.

{\bf By \lemref{lem:inv}:} $\ctype{\Psi}{\Lambda_1}{\mc{C}_1}{\Delta_1, \Delta_2, \Delta_3,\Delta_4,\Lambda'_1}$ and $\ctype{\Psi}{\Lambda_2,\Delta'_1, x_\beta{:}A[c]}{\mb{msg}(\mathbf{send}\,x_\beta^{c}\,y_\alpha^{c};y_{\alpha}^{c}\leftarrow v_\delta^{c})}{y_\alpha{:}A\otimes B[c]}$ and $\ctype{\Psi}{\Lambda_3, \Delta_2}{\mc{C}'}{\Lambda'_2, \Delta_5, \Delta_6}$ and
$\ctype{\Psi}{\Lambda_4, \Delta_3, y_\alpha{:}A\otimes B[c], \Delta_5}{\mb{proc}(u_\gamma^{c'},w_\eta\leftarrow \mathbf{recv}\,y_\alpha^{c}; P\,@d_1)}{(u_\gamma{:}C[c'])}$ and
$\ctype{\Psi}{\Lambda_5, \Delta_4, \Delta'_6}{\mc{C}_2}{\Lambda'_3}$.

Where $\Delta=\Lambda_1,\Lambda_2,\Lambda_3,\Lambda_4,\Lambda_5$ and $\Delta'=\Lambda'_1, \Lambda'_2, \Lambda'_3$, and  $\Delta'_1= \Delta_1$ if $x_\beta{:}A[c] \in \Delta$, and $\Delta'_1, x_\beta{:}A[c] \in \Delta= \Delta_1$ otherwise. Also, $\Delta'_6 = \Delta_6$ if $u_\gamma{:}C[c'] \in \Delta'$ and otherwise $\Delta'_6 = \Delta_6, u_\gamma{:}C[c']$.

{\bf By inversion on $\mb{msg}$ rule}
$\ptype{\Psi}{\Lambda_2,\Delta'_1, x_\beta{:}A[c]}{(\mathbf{send}\,x_\beta^{c}\,y_\alpha^{c};y_{\alpha}^{c}\leftarrow v_\delta^{c})@c}{y}{c}{\alpha}{A\otimes B}$. Moreover, $(\star)\; \; \forall u_\gamma{:}T[d'] \in \Lambda_2, \Delta'_1, x_\beta{:}A[c]. \, \Psi \Vdash d' \sqsubseteq c$. 
 
{\bf By inversion on $\mb{proc}$ rule}
$\ptype{\Psi}{\Lambda_4,\Delta_3, y_\alpha{:} A \otimes B[c], \Delta_5}{(w_\eta \leftarrow \mathbf{recv}\,y_\alpha^{c}; P)\,@d_1}{u}{c'}{\gamma}{C}$. Moreover, $(\star')\; \; \forall u_\gamma{:}T[d'] \in \Lambda_4, \Delta_3, y_\alpha{:}A\otimes B[c], \Delta_5. \, \Psi \Vdash d' \sqsubseteq c'$, and $\Psi \Vdash d_1 \sqsubseteq c'$.

{\bf By inversion on $\otimes\, L$  rule}
$\ptype{\Psi,  \psi = c }{\Lambda_4,\Delta_3, y_{\alpha}{:}B[c], w_{\eta}{:}A[\psi], \Delta_5}{P\,@d_1 \sqcup c}{u}{c'}{\gamma}{C}$.

{\bf By substitution of  $x^c_\beta$ for $w_\eta$ and $v_\delta$ for $y_{\alpha}$:}
\[\ptype{\Psi, c = c }{\Lambda_4,\Delta_3, v_{\delta}{:}B[c], x_{\beta}{:}A[c], \Delta_5}{[v_\delta/y_{\alpha}][x_\beta^c/w_\eta]P\,@d_1 \sqcup c}{u}{c'}{\gamma}{C}.\] Moreover, we have $\Psi = \Psi, c=c$.

{\bf By $\m{proc}$ rule, $(\star)$ and $(\star')$:} 
$\ctype{\Psi}{\Lambda_4, \Delta_3, v_{\delta}{:}B[c], x_\beta{:}A[c], \Delta_5}{\mb{proc}(u_{\gamma}[c'], [x_\beta^c/w_\eta][v_{\delta}/y_{\alpha}] P\,@d_1\sqcup c) \mc{C}_2}{u_\gamma{:}C[c']}.$

{\bf By configuration typing rules}
$\ctype{\Psi}{\Delta}{\mc{C}_1\mc{C}'\mb{proc}(u_{\gamma}[c'], [x_\beta^c/w_\eta][v_{\delta}/y_{\alpha}] P\,@d_1\sqcup c) \mc{C}_2}{\Delta'}$

Moreover, after taking this step, $\ptype{\Psi}{\Lambda_4,\Delta_3, y_\alpha{:} A \otimes B[c], \Delta_5}{(w_\eta \leftarrow \mathbf{recv}\,y_\alpha^{c}; P)\,@d_1}{u}{c'}{\gamma}{C}$ is replaced in the configuration by the typing judgment $\ptype{\Psi }{\Lambda_4,\Delta_3, v_{\delta}{:}B[c], x_{\beta}{:}A[c], \Delta_5}{[v_\delta/y_{\alpha}][x_\beta^c/w_\eta]P\,@d_1 \sqcup c}{u}{c'}{\gamma}{C}$  which has a smaller process term. The typing judgment for the message disappears after taking the step. This observation is enough to show that $\mc{C}_1\mc{C}'\mb{proc}(u_{\gamma}[c'], [x^c_\beta/w_\eta][v_{\delta}/y_{\alpha}] P\,@d_1\sqcup c) \mc{C}_2<\mc{C}_1 \mb{msg}(\mathbf{send}\,x_\beta^{c}\,y_\alpha^{c};y_{\alpha}^{c}\leftarrow v_\delta^{c}) \mc{C}' \mb{proc}(u_\gamma[c'],w_\eta \leftarrow \mathbf{recv}\,y_\alpha^{c}; P\,@d_1) \mc{C}_2$.
\end{description}
\end{proof}

 \begin{theorem}[Progress]\label{thm:progress}
If $\Psi; \Delta \Vdash \mc{C}::\Delta'$, then 
either $\mc{C} \mapsto_{\Delta\Vdash \Delta'}\mc{C}'$  or $\mc{C}$ is poised.
\end{theorem}
\begin{proof}
The proof is by induction on the configuration typing of $\mc{C}$. If $\mc{C} \mapsto \mc{C}'$ then the proof is complete. Otherwise we consider  the last rule in the typing of $\mc{C}$.
\begin{description}
\item {\bf Case 1.}\[\infer[\mathbf{emp}_1]{\Psi; x{:}A[d] \Vdash\cdot :: ( x{:}A[d]) }{}\qquad \infer[\mathbf{emp}_2]{\Psi; \cdot \Vdash\cdot :: (\cdot) }{}\]
In this case $\mc{C}$ is empty and thus poised. 
is empty and thus poised.\item {\bf Case 2.}\[
    \infer[\mathbf{proc}]{\Psi; \Delta_0, \Delta'_0 \Vdash \mathcal{C}', \mathbf{proc}(x[d], P@d_1):: (x{:}A[d]) }{\Psi; \Delta_0 \Vdash \mathcal{C}':: \Delta & \deduce{\Psi; \Delta'_0, \Delta \vdash P@d_1:: (x{:}A[d])}{\forall (y{:}B[d']) \in \Delta'_0, \Delta & \deduce{(\Psi \Vdash d' \sqsubseteq d)}{\Psi \Vdash d_1 \sqsubseteq d}} }\]
If  $\mathbf{proc}(x[d], P@d_1):: (x{:}A[d])$ attempts to receive on $x$ or is forward then the proof is complete. Otherwise, if it wants to send or spawn along one of its channels $\mc{C}$ can take a step. It remains to consider the case in which $\mc{C}$ wants to receive along one of its positive resources  $y{:}B[c] \in \Delta'_0, \Delta$. If $y{:}B[c]\in {\Delta}'_0$, then the proof is complete by the definition of poised configurations. If $y{:}B[c]\in \Delta$, we apply the induction hypothesis on $\mc{C}'$. If $\mc{C}'$ can take a step, so does $\mc{C}$ and the proof is complete. 
If $\mc{C}'$ is poised, the subconfiguration $\Delta_0''\Vdash \mc{C}''::y{:}B[c]$ for some $\Delta''_0 \subseteq \Delta_0$ is also poised. If $\mc{C}''$ is empty, then $y{:}B[c]\in \Delta''_0 \subseteq \Delta_0$ and by definition $\mathcal{C}$ is poised. Otherwise, by assumption $\mc{C}''$ cannot attempt to receive along $y{:}B[c]$ since $B$ is positive.  If it offers a positive message along $y{:}B[c]$, then the proof is complete since $\mc{C}$ can take a step. If $\mc{C}'$ has a forwarding on the root, then $\mc{C}'$ can take a step again. In the other cases, poisedness of $\mc{C}$ follows by definition.

\item {\bf Case 3.}\[  \infer[\mathbf{msg}]{\Psi; \Delta_0, \Delta'_0 \Vdash \mathcal{C}', \mathbf{msg}(P):: (x{:}A[d]) }{\Psi; \Delta_0 \Vdash \mathcal{C}':: \Delta & \deduce{\Psi; \Delta'_0, \Delta \vdash P@d:: (x{:}A[d])}{\forall (y{:}B[d']) \in \Delta'_0, \Delta & {(\Psi \Vdash d' \sqsubseteq d)}} }\]
If $\mathbf{msg}(P)$ is a positive message then the proof is complete by induction on $\mc{C}'$: if it can take a step, so does $\mc{C}$, and if it is poised, so is $\mc{C}$. If $\mathbf{msg}(P)$ is a negative message along $y_\alpha{:}A[c] \in \Delta$ then we proceed the proof by induction on $\mc{C}'$ again: if it can take a step, so does $\mc{C}$. If $\mc{C}'$ is poised, the subconfiguration $\Delta_0''\Vdash \mc{C}''::y_\alpha{:}A[c]$ for some $\Delta_0''\subseteq \Delta_0$ is also poised. If $\mc{C}''$ is empty then by typing rules $y_\alpha{:}A[c]\in \Delta_0''\subseteq \Delta_0$ and by definition $\mathcal{C}$ is poised. Otherwise, by assumption $\mc{C}''$ cannot offer a positive message along $y_\alpha{:}A[c]$. If it attempts to receive along $y_\alpha{:}A[c]$ then the proof is complete since $\mc{C}$ can take a step. If it has a forwarding on its root, then $\mc{C}$ can take a step again. In the other cases, poisedness of $\mc{C}$ follows by definition. If $\mathbf{msg}(P)$ is a negative message along $y_\alpha{:}A[c] \in \Delta'_0$ then with a similar argument by induction on $\mc{C}'$ we can prove the progress statement. 
\item {\bf Case 4.}\[   \infer[\mathbf{comp}]{\Psi; \Delta_0, \Delta'_0 \Vdash \mathcal{C}', \mc{C}'' :: \Delta, x{:} A[d]}{\Psi; \Delta_0 \Vdash \mathcal{C}':: \Delta & {\Psi; \Delta'_0 \Vdash \mc{C}'':: x{:}A[d]} }\]
where $\mc{C}=\mc{C}', \mc{C}''$. By induction hypothesis, either (i) $\mc{C}'$ can take a step or  (ii) $\mc{C}'$ is empty or (iii) $\mc{C}'$ is poised. 
In (i) the proof is complete, since $\mc{C}'\mc{C}''$ also can take a step. In (ii) the proof is complete since $\mc{C}=\mc{C}''$ and we can apply the induction hypothesis on $\mc{C}''$. In (iii) we apply the induction hypothesis on $\mc{C}''$ and consider the cases: (i') $\mc{C}''$ can take a step which completes the proof, or (ii') $\mc{C}''$ is empty which again completes the proof, or (iii') $\mc{C}''$ is poised, which is enough to prove that $\mc{C}$ is poised and completes the proof.
\end{description}
\end{proof}

\begin{definition}
$\mc{P}_{\mid \xi}$ stands for a rewriting of the process term $\mc{P}$ by renaming any number of its  channels with higher than or incomparable to the observer level.
\end{definition}

\begin{lemma}\label{lem:indinvariant}
Consider $(\lr{\mc{C}_1}{\mc{D}_1}{\mc{F}_1}; \lr{\mc{C}_2}{\mc{D}_2}{\mc{F}_2}) \in \m{Tree}_\Psi(\Delta \Vdash K)$. 
If $\lr{\mc{C}_1} {\mc{D}_1}{ \mc{F}_1} \mapsto_{\Delta \Vdash K}\lr{\mc{C}'_1} {\mc{D}'_1}{ \mc{F}'_1}$  and $\cproj{\mc{D}_1}{\xi}=\cproj{\mc{D}_2}{\xi}$, then for some $\lr{\mc{C}'_2}{\mc{D}'_2} {\mc{F}'_2}$, we have $\lr{\mc{C}_2} {\mc{D}_2}{ \mc{F}_2} \mapsto^{0,1}_{\Delta \Vdash K}\lr{\mc{C}'_2} {\mc{D}'_2}{ \mc{F}'_2}$ such  that $\cproj{\mc{D}'_1}{\xi}=\cproj{\mc{D}'_2}{\xi}$.
\end{lemma}
\begin{proof}
The proof is by cases on the possible $\mapsto_{\Delta \Vdash K}$ steps. If the step takes place in $\mathcal{C}_1$ or $\mathcal{F}_1$, the proof trivially holds. We only consider the cases in which $\mathcal{D}$ steps. In each case we prove that either the step does not change relevancy of any process in $\mc{D}_1$ or we can step $\mc{D}_2$ such that the same change of relevancy occurs in it too.\\
 {\color{ForestGreen} \bf Case 1. $\mc{D}_1=\mc{D}'_1\mb{proc}(y_\alpha[c],  y^c_\alpha.k ;P @d_1) \mc{D}''_1$ and
\[ \mc{D}'_1\mb{proc}(y_\alpha[c],  y^c_\alpha.k ;P @d_1) \mc{D}''_1\mapsto_{\Delta \Vdash K} \mc{D}'_1\mb{proc}(y_{\alpha+1}[c], [y^c_{\alpha+1}/y^c_{\alpha}] P@d_1) \mb{msg}(y^c_\alpha.k; y^c_{\alpha+1}\leftarrow y^c_{\alpha}) \mc{D}''_1\]}
We consider subcases based on relevancy of process offering along $y_\alpha[c]$: 
\begin{description}
\item {\color{Aquamarine}\bf Subcase 1. $\mb{proc}(y_\alpha[c],  y^c_\alpha.k ;P @d_1)$ is not relevant.} By inversion on the typing rules $d_1 \sqsubseteq c$. By definition either $d_1 \not\sqsubseteq \xi $ or none of the channels connected to $P$ including its offering channel $y_\alpha^c$ are relevant. In both cases neither $\mb{proc}(y_{\alpha+1}[c], [y^c_{\alpha+1}/y^c_{\alpha}] P@d_1)$, nor  $\mb{msg}(y^c_\alpha.k; y^c_{\alpha+1} \leftarrow y_c^\alpha)$ are relevant in the post step. Note that from $d_1 \sqsubseteq c$ and $d_1 \not\sqsubseteq \xi$, we get $c \not\sqsubseteq \xi $. Channel $y^c_\alpha$ is not relevant in the pre-step, and both $y^c_\alpha$ and $y^c_{\alpha+1}$ are not relevant in pre-step and post-step configurations. Every not relevant resource of $\mb{proc}(y^c_\alpha,  y^c_\alpha.k ;P @d_1)$ will remain irrelevant in the post-step too.

In this subcase, our goal is to show 
\[\cproj{\mc{D}'_1\mb{proc}(y_{\alpha+1}[c], [y^c_{\alpha+1}/y^c_{\alpha}] P @d_1) \mb{msg}(y^c_\alpha.k;y^c_{\alpha+1}\leftarrow y^c_{\alpha}) \mc{D}''_1}{\xi} \meq \cproj{\mc{D}'_1 \mc{D}''_1}{\xi}\meq \cproj{\mc{D}_1}{\xi}\meq \cproj{\mc{D}_2}{\xi}.\]

To prove this we need two observations:
\begin{itemize}
    \item Neither $\mb{proc}(y_{\alpha+1}[c], [y^c_{\alpha+1}/y^c_{\alpha}] P@d_1)$ nor  $\mb{msg}(y^c_\alpha.k;y^c_{\alpha+1}\leftarrow y^c_{\alpha})$ are relevant and they will be dismissed by the projection. (As explained above.)
    \item Replacing $\mb{proc}(y_\alpha[c],  y^c_\alpha.k ;P @d_1)$ with these two nodes, does not affect relevancy of the rest of processes in $\mc{D}'_1\mc{D}''_1$. Relevancy of processes in $\mc{D}''_1$ remains intact since $y^c_\alpha$ and $y^c_{\alpha+1}$ are irrelevant. 
    
    The relevancy of processes in $\mc{D}'_1$ remains intact too as we replace their irrelevant root with another irrelevant process. However, we need to be careful about the changes in the quasi running secrecy of a process and their effect on its (grand)children. The quasi-running secrecy of the process offering along $y^c_{\alpha+1}$ may be higher or incomparable to $d_1$ based on the code of $P$ (if it starts with a $\mb{recv}$ or $\mb{case}$). This is of significance only if $d_1 \sqsubseteq \xi$, and in the pre-step the process has a chain of positive messages along a channel with secrecy level lower than or equal to the observer level (e.g. $x:\_[d]$ where $d\sqsubseteq \xi$) as (grand)children. But by the assumption of the subcase, $x:\_[d]$ cannot be a relevant channel in the pre-step. Thus the chain of messages is not relevant in neither the pre-step and nor the post-step.
\end{itemize}

\item {\color{Aquamarine}\bf Subcase 2. $\mb{proc}(y_\alpha[c],  y^c_\alpha.k ;P @d_1)$ is relevant.}
By assumption ($\cproj{\mc{D}_1}{\xi}\meq\cproj{\mc{D}_2}{\xi}$):
\[\mc{D}_2=\mc{D}'_2\mb{proc}(z_\delta[c],  z^c_\delta.k ;P_{\mid \xi} @d_1) \mc{D}''_2,\] such that $z_\delta=y_\alpha$ if $c\sqsubseteq \xi$ and $P_{\mid \xi}$ is equal to $P$ modulo renaming of some channels with higher or incomparable secrecy to the observer. We have \[\lr{\mc{C}_2}{\mc{D}_2}{\mc{F}_2} \mapsto_{\Delta \Vdash K}\lr{\mc{C}_2}{\mc{D}'_2\mb{proc}(z_{\delta+1}[c], [z^c_{\delta+1}/z^c_\delta]P_{\mid \xi} @d_1) \mb{msg}(z^c_\delta.k; z^c_{\delta+1}\leftarrow z^c_{\delta})\mc{D}''_2}{\mc{F}_2}\]

It is enough to show 
{\small\[\cproj{\mc{D}'_1\mb{proc}(y_{\alpha+1}[c], [y^c_{\alpha+1}/y^c_{\alpha}] P @d_1) \mb{msg}(y^c_\alpha.k;  y^c_{\alpha+1}\leftarrow y^c_{\alpha}) \mc{D}''_1}{\xi} \meq \cproj{\mc{D}'_2\mb{proc}(z_{\delta+1}[c], [z^c_{\delta+1}/z^c_\delta]P_{\mid \xi} @d_1) \mb{msg}(z^c_\delta.k;z^c_{\delta+1}\leftarrow z^c_{\delta})\mc{D}''_2}{\xi}.\]}

If $c\not\sqsubseteq \xi$, then $\mb{msg}(y^c_\alpha.k; y^c_{\alpha+1} \leftarrow y^c_{\alpha})$ and $\mb{msg}(z^c_\delta.k; z^c_{\delta+1} \leftarrow z^c_{\delta})$ are not relevant in both runs, and will be dismissed by the projections. Moreover, in this case neither $y^c_\alpha$, nor $z^c_\delta$ are relevant in the pre-step and post-step configurations. Thus the relevancy of processes in $\mc{D}''_1$ and $\mc{D}''_2$ will remain intact.

If $c \sqsubseteq \xi$, then $z^c_\delta=y^c_\alpha$ are relevant in the pre-step in both runs. In the post-step, $z^c_{\delta+1}=y^c_{\alpha+1}$ are relevant in both runs. Relevancy of messages $\mb{msg}(y^c_\alpha.k; y^c_{\alpha+1} \leftarrow y^c_{\alpha})$ and $z^c_\delta=y^c_\alpha$ in the post-steps are determined by the quasi running secrecies ($d$ and $d'$) of their parents ($X$ and $X'$) in $\mc{D}''_1$ and $\mc{D}''_2$. If $d\sqsubseteq \xi$, then the parent ($X$) is relevant in the first run and by assumption is equal to a relevant $X'$ in the second run. Thus messages $\mb{msg}(y^c_\alpha.k; y^c_{\alpha+1} \leftarrow y^c_{\alpha})$ are relevant in both runs and $z^c_\delta=y^c_\alpha$ are relevant in the post-step too. The same holds when $d'\sqsubseteq \xi$.

Otherwise, in both runs the quasi running secrecy of the parent is higher than or incomparable to the observer (the parents are both irrelevant). Thus messages $\mb{msg}(y^c_\alpha.k; y^c_{\alpha+1} \leftarrow y^c_{\alpha})$ are not relevant in the post step of both runs, and will be dismissed by the projections. The channels $z^c_\delta=y^c_\alpha$ will be irrelevant in the post-step too. However, this does not affect the processes in $\mc{D}''_1$ and $\mc{D}''_2$ as the parents of messages ($X$ and $X'$) are already irrelevant in the pre-step.

It remains to show that projections of $\mc{D}'_1$ and $\mc{D}'_2$ are equal in the post-step too. The resources with secrecies higher than or incomparable to the observer offered along $\mc{D}'_1$ and $\mc{D}'_2$ in the pre-step will remain higher than or incomparable to the observer and thus irrelevant in the post-step too. For a relevant resources ($w$) offered along $\mc{D}_1$ and $\mc{D}_2$, we need to consider the change in quasi running secrecy. The quasi running secrecy of the processes offering along $y^c_{\alpha+1}$ and $z^c_{\delta+1}$ may increase based on their code (if the code of $P$ and $P_{\mid \xi}$ starts with a $\mb{recv}$ or $\mb{case}$) and become irrelevant. This means that a relevant sub-tree $\mc{T}_i$ of $\mc{D}'_i$ offering along $w$ in the pre-step will become irrelevant in the post-step.  But by the assumption of the theorem, we know that $\mc{T}_1=\mc{T}_2$. In the post step, we project out the equal subtrees $\mc{T}_1$ and $\mc{T}_2$ from the relevant parts of $\mc{D}'_1$ and $\mc{D}'_2$; the projections will remain equal in the post-step.
\end{description}

 {\color{ForestGreen}\bf Case 2. $\mc{D}_1=\mc{D}'_1\mb{proc}(x_\beta[d],  y^c_\alpha.k ;P @d_1) \mc{D}''_1$ and
 \[ \mc{D}'_1\mb{proc}(x_\beta[d],  y^c_\alpha.k ;P @d_1) \mc{D}''_1\mapsto_{\Delta \Vdash K} \mc{D}'_1\mb{msg}(y^c_\alpha.k;  y^c_{\alpha}\leftarrow y^c_{\alpha+1})\mb{proc}(x_\beta[d], [y^c_{\alpha+1}/y^c_{\alpha}]P @d_1)  \mc{D}''_1\]}
We consider subcases based on relevancy of the process offering along $x_d^\beta$: 
\begin{description}
\item {\color{Aquamarine}\bf Subcase 1. $\mb{proc}(x_\beta[d],  y^c_\alpha.k ;P @d_1)$ is irrelevant.} By inversion on the typing rules, $d_1 \sqsubseteq c \sqsubseteq d$. By definition either $d_1 \not \sqsubseteq \xi$  or none of the channels connected to $P$ including $y^c_\alpha$ and $x^d_\beta$ are relevant. In both cases, neither $\mb{msg}(y^c_\alpha.k;y^c_{\alpha}\leftarrow y^c_{\alpha+1})$ nor $\mb{proc}(x_\beta[d], [y^c_{\alpha+1}/y^c_{\alpha}]P @d_1)$ are relevant. Channel $x^d_\beta$  is irrelevant in the pre-step and post-step configurations.

Channel $y^c_\alpha$  is irrelevant in the pre-step, and both $y^c_\alpha$ and $y^c_{\alpha+1}$ are irrelevant in pre-step and post-step configurations. Every other irrelevant resource of $\mb{proc}(x_\beta[d],  y^c_\alpha.k ;P @d_1)$ will remain irrelevant in the post-step too.

In this subcase, our goal is to show 
\[\cproj{\mc{D}'_1\mb{msg}(y^c_\alpha.k;y^c_{\alpha}\leftarrow y^c_{\alpha+1})\mb{proc}(x_\beta[d], [y^c_{\alpha+1}/y^c_{\alpha}]P @d_1) \mc{D}''_1}{\xi} \meq \cproj{\mc{D}'_1 \mc{D}''_1}{\xi}\meq \cproj{\mc{D}_1}{\xi}\meq \cproj{\mc{D}_2}{\xi}\]

With a same argument as in {\bf \color{ForestGreen}Case 1.} {\bf \color{Aquamarine}Subcase 1.}, we can prove that the relevancy of processes in $\mc{D}'_1$ and $\mc{D}''_1$ remain intact.

\item {\color{Aquamarine}\bf Subcase 2. $\mb{proc}(x_\beta[d],  y^c_\alpha.k ;P @d_1)$ is relevant.} 
By assumption that $\cproj{\mc{D}_1}{\xi}\meq\cproj{\mc{D}_2}{\xi}$, and definition of $\meq$: \[\mc{D}_2=\mc{D}'_2\mb{proc}(u_\gamma[d],  z^c_\delta.k ;P_{\mid \xi} @d_1) \mc{D}''_2,\] such that $z_\delta=y_\alpha$ if $c\sqsubseteq \xi$ and $u_\gamma=x_\beta$ if $d \sqsubseteq \xi$ and $P_{\mid \xi}$ is equal to $P$ modulo renaming of some channels with higher than or incomparable to the observer level. We have \[\lr{\mc{C}_2}{\mc{D}_2}{\mc{F}_2} \mapsto_{\Delta \Vdash K}\lr{\mc{C}_2}{\mc{D}'_2\mb{msg}(z^c_\delta.k; z^c_{\delta}\leftarrow z^c_{\delta+1})\mb{proc}(u_{\gamma}[d], [z^c_{\delta+1}/z^c_\delta]P_{\mid \xi} @d_1) \mc{D}''_2}{\mc{F}_2}\]

If $d \sqsubseteq \xi$, then $u_\gamma=x_\beta$ are relevant in the pre-steps of both runs and  remain relevant in the post-steps. Even if the quasi running secrecy increases based on the code of $P$ and $P_{\mid \xi}$, it will be lower than or equal to the observer level by the tree invariant. Thus the relevancy of processes in $\mc{D}''_i$ remain intact. Moreover, every resource of the processes in $\mc{D}'_i$ is relevant in the pre-steps and post-steps.

If $d \not \sqsubseteq \xi$, then $u_\gamma$ and $x_\beta$ remain irrelevant in the pre-steps and post-steps of both runs: the relevancy of processes in $\mc{D}''_i$ remain intact.

It remains to show that the projections of $\mc{D}'_1$ and $\mc{D}'_2$ in post-steps are still equal. Here we only consider the trees offered along $y^c_\alpha$ and $z^c_\delta$ in both runs. The argument for the rest of $\mc{D}'_i$ is similar to {\bf \color{ForestGreen}Case 1.} {\bf \color{Aquamarine}Subcase 2.}

The adaptive running secrecy of the negative message $\mb{msg}(z^c_\delta.k; z^c_{\delta}\leftarrow z^c_{\delta+1})$ is $c$ in the post-steps. If $c\sqsubseteq \xi$ then the same message exists in both runs, and the tree offered along $y^c_\alpha=z^c_\delta$ is relevant in the pre-steps and post-steps.

If $c \not \sqsubseteq \xi$ then the message is irrelevant in both runs. $y^c_\alpha$ and $z^c_\delta$ are both irrelevant in the pre-step and remain irrelevant in the post-step too. By typing rules $\mb{msg}(z^c_\delta.k; z^c_{\delta}\leftarrow z^c_{\delta+1})$ is not a parent of any positive messages.

\end{description}
{\color{ForestGreen} \bf Case 3. $\mc{D}_1=\mc{D}'_1\mc{C}''_1\mb{proc}(y^c_\alpha, \mb{send} x^c_\beta\, y^c_\alpha @d_1) \mc{D}''_1$ and
\[\mc{D}'_1\mc{C}''_1\mb{proc}(y^c_\alpha, \mb{send} x^c_\beta\, y^c_\alpha ;P @d_1) \mc{D}''_1\mapsto_{\Delta \Vdash K} \mc{D}'_1 \mc{C}''_1\mb{proc}(y^c_{\alpha+1}, [y^c_{\alpha+1}/y^c_{\alpha}]P @d_1)\mb{msg}(\mb{send} x^c_\beta\, y^c_\alpha; y^c_{\alpha+1}\leftarrow y^c_{\alpha} )  \mc{D}''_1\]}
such that $\Delta=\Delta'\Delta''$ and $\Delta' \Vdash \mc{D}'_1:: \Lambda$ and  $\Delta'' \Vdash \mc{C}''_1:: (x_{\beta}{:}A[c])$ and $\Lambda, x_{\beta}{:}A[c] \Vdash  \mb{proc}(y_\alpha[c], \mb{send} x^c_\beta\, y^c_\alpha ;P @d_1):: (y_\alpha{:}A \otimes B[c])$. (In the case that $\mc{C}''_1$ is empty, we have $\Delta=\Delta', x_\beta{:}A[c]$.)

We consider subcases based on relevancy of process offering along $y^c_\alpha$: 
\begin{description}
\item {\color{Aquamarine}\bf Subcase 1. $\mb{proc}(y_c^\alpha, \mb{send} x^c_\beta\, y^c_\alpha ;P @d_1)$ is not relevant.} By inversion on the typing rules $d_1 \sqsubseteq c$. By definition either $d_1 \not \sqsubseteq \xi $ or none of the channels connected to $P$ including $y^c_\alpha$ and $x^c_\beta$ are relevant. In both cases, neither $\mb{proc}(y_{\alpha+1}[c], [y^c_{\alpha+1}/y^c_{\alpha}]P @d_1)$ nor $\mb{msg}(\mb{send} x^c_\beta\, y^c_\alpha; y^c_{\alpha+1}\leftarrow y^c_{\alpha})$ are relevant. Channel $y^c_\alpha$  is irrelevant in the pre-step, and both $y^c_\alpha$ and $y^c_{\alpha+1}$ are irrelevant in pre-step and post-step configurations. In this subcase, our goal is to show

\[\cproj{\mc{D}'_1\mc{C}''_1\mb{proc}(y_{\alpha+1}[c], [y^c_{\alpha+1}/y^c_{\alpha}]P @d_1) \mb{msg}(\mb{send} x^c_\beta\, y^c_\alpha; y^c_{\alpha+1}\leftarrow y^c_{\alpha}) \mc{D}''_1}{\xi} \meq \cproj{\mc{D}'_1 \mc{C}''_1 \mc{D}''_1}{\xi}\meq \cproj{\mc{D}_1}{\xi}\meq \cproj{\mc{D}_2}{\xi}.\]

We first prove that the relevancy status of  $\mc{C}''_1$ remain intact too. Note that all channels in $\mc{C}''_1$, except $x^c_\beta$ have the same connections in the pre-step and post-step. So it is enough to consider the changes made to the tree rooted at $x^c_\beta$. If $c \not \sqsubseteq \xi$, then $x^c_\beta$ is not relevant in the pre-step and the post-step. Thus relevancy status of $\mc{C}''_1$ remains intact. Moreover, the message is irrelevant. The tree rooted at $x^c_\beta$ offers to a node with a  quasi running secrecy higher than or incomparable to the obsever before and after the step. Thus, the relevancy of tree rooted at $x^c_\beta$ is not affected. With a same argument as in {\bf \color{ForestGreen}Case 1.} {\bf \color{Aquamarine}Subcase 1.} and the one given for $\mc{C}''_1$, we can prove that the relevancy status of processes in $\mc{D}'_1$ and $\mc{D}''_1$ remain intact.

\item {\color{Aquamarine}\bf Subcase 2. $\mb{proc}(y^c_\alpha, \mb{send} x^c_\beta\, y^c_\alpha ;P @d_1)$ is relevant.}
By assumption that $\cproj{\mc{D}_1}{\xi}\meq\cproj{\mc{D}_2}{\xi}$, and definition of $\meq$:
\[\mc{D}_2=\mc{D}'_2\mc{C}''_2\mb{proc}(z_\delta[c], \mb{send} u^c_\gamma\, z^c_\delta ;P_{\mid \xi} @d_1)\mc{D}''_2\]

 such that $z_\delta=y_\alpha$ and $u_\gamma=x_\beta$, if $c \sqsubseteq \xi$, and $P_{\mid \xi}$ is equal to $P$ modulo renaming of some  channels with higher than or incomparable to the observer level.
 
 We have \[\lr{\mc{C}_2}{\mc{D}_2}{\mc{F}_2} \mapsto_{\Delta \Vdash K}\lr{\mc{C}_2}{\mc{D}'_2 \mc{C}''_2\mb{proc}(z^c_{\delta+1}, [z^c_{\delta+1}/z^c_\delta]P_{\mid \xi} @d_1) \mb{msg}(\mb{send} u^c_\gamma\, z^c_\delta; z^c_{\delta+1}\leftarrow z^c_{\delta})\mc{D}''_2}{\mc{F}_2}\]
 If $c \not\sqsubseteq \xi$, then $\mb{msg}(\mb{send} x^c_\beta\, y^c_\alpha;y^c_{\alpha+1}\leftarrow y^c_{\alpha})$ and $\mb{msg}(\mb{send} u^c_\gamma\, z^c_\delta; z^c_{\delta+1}\leftarrow z^c_{\delta})$ are not relevant in both runs, and will be dismissed by the projections. Moreover, neither $y^c_\alpha$, nor $z^c_\delta$ are relevant in the pre-step and post-step configurations. Thus the relevancy of processes in $\mc{D}''_1$ and $\mc{D}''_2$ will remain intact. Moreover, in this case $x^c_\beta$ and $u^c_\gamma$ are irrelevant in both pre-steps and post-steps. Which means that relevancy status of $\mc{C}''_1$ and $\mc{C}''_2$ remains intact.

If $ c \sqsubseteq \xi$, then $z^c_\delta=y^c_\alpha$ are relevant in the pre-step in both runs. We also know that $x_\beta= u_\gamma$ are relevant in the pre-step and $\mc{C}''_1=\mc{C}''_2$ are relevant in the pre-step. In the post-step, $z^c_{\delta+1}=y^c_{\alpha+1}$ are relevant in both runs. Relevancy of messages $\mb{msg}(\mb{send} x^c_\beta\, y^c_\alpha;  y^c_{\alpha+1}\leftarrow y^c_{\alpha})$  in the post-steps are determined by the quasi running secrecy ($d$ and $d'$) of their parents ($X$ and $X'$) in $\mc{D}''_1$ and $\mc{D}''_2$. If $d\sqsubseteq \xi$, then the parent ($X$) is relevant in the first run and by assumption is equal to a relevant $X'$ in the second run. Thus messages $\mb{msg}(\mb{send} x^c_\beta\, y^c_\alpha;  y^c_{\alpha+1}\leftarrow y^c_{\alpha})$ are relevant in both runs, $z^c_\delta=y^c_\alpha$ are relevant in the post-step, and trees $\mc{C}''_1=\mc{C}''_2$ and their offering channels $x^d_\beta=u^d_\gamma$ are relevant in the post-steps too.  The same holds when $d'\sqsubseteq \xi$.

Otherwise, in both runs the quasi running secrecy of the parent is higher than or incomparable to the observer level (the parents are both irrelevant). Thus messages $\mb{msg}(\mb{send} x^c_\beta\, y^c_\alpha;  y^c_{\alpha+1}\leftarrow y^c_{\alpha})$ are not relevant in the post step of both runs, and will be dismissed by the projections. The channels $z^c_\delta=y^c_\alpha$ will be irrelevant in the post-step too. However, this does not affect the processes in $\mc{D}''_1$ and $\mc{D}''_2$ as the parents of messages ($X$ and $X'$) are already irrelevant in the pre-step. The channels $x^c_\beta=u^c_\gamma$ both become irrelevant in the post-steps. However, we still have $\cproj{\mc{C}''_1}{\xi} =\cproj{\mc{C}''_2}{\xi}$ as they have the same type and their parents have the same quasi running secrecy.

We show that projections of $\mc{D}'_1$ and $\mc{D}'_2$ are equal in the post-step too. The  resources with secrecy level higher than or incomparable to the observer level offered along $\mc{D}'_1$ and $\mc{D}'_2$ in the pre-step will remain higher than or incomparable to and thus irrelevant in the post-step too. For a relevant resources ($w$) offered along $\mc{D}_1$ and $\mc{D}_2$, we need to consider the change in quasi running secrecy as in {\bf \color{ForestGreen}Case 1.} {\bf \color{Aquamarine}Subcase 2.} Moreover, we need to consider the scenario that a relevant resource ($w^{c'}$) in the pre-step loses its relevancy in the post-step because the channel offered along $x_\beta^c$ is transferred to the message. This case only happens if $c', c \sqsubseteq \xi$ and thus the trees $\mc{T}_1$ and $\mc{T}_2$ offered along $w^{c'}$ is present in both runs and $\mc{T}_1=\mc{T}_2$. We know that $w^{c'}$ is irrelevant in the post-step of both runs, and the quasi-running secrecy of the processes using the resource $w^{c'}$ in both runs are the same.

The relevant and irrelevant processes in $\mc{D}''_i$ remain intact.


\end{description}

{\color{ForestGreen} \bf Case 4. $\mc{D}_1=\mc{D}'_1\mc{C}''_1\mb{proc}(w_\eta[c'], \mb{send} x^c_\beta\, y^c_\alpha @d_1) \mc{D}''_1$ and
\[\mc{D}'_1\mc{C}''_1\mb{proc}(w_\eta[c'], \mb{send} x^c_\beta\, y^c_\alpha ;P @d_1) \mc{D}''_1\mapsto_{\Delta \Vdash K} \mc{D}'_1\mc{C}''_1\mb{msg}(\mb{send} x^c_\beta\, y^c_\alpha )\mb{proc}(w_\eta[c'], [y^c_{\alpha+1}/y^c_{\alpha}]P @d_1)  \mc{D}''_1\]}

such that $\Delta=\Delta_1\Delta_2$ and $\Delta_1 \Vdash \mc{D}'_1:: \Lambda, y_\alpha{:}A\multimap B[c]$ and  $\Delta_2 \Vdash \mc{C}''_1:: (x_{\beta}{:}A[c])$ and \[\Lambda, y_\alpha{:}A \multimap B[c], x_{\beta}{:}A[c] \Vdash  \mb{proc}(w_\eta[c'], \mb{send} x^c_\beta\, y^c_\alpha ;P @d_1):: (w_\eta{:}C[c']).\] In the case where $\mc{C}''_1$ is empty we have  $\Delta_2=x_{\beta}{:}A[c]$. We proceed by considering subcases based on relevancy of the process offering along $w_\eta[c']$:

\item {\color{Aquamarine}\bf Subcase 1. $ \mb{proc}(w_\eta[c'], \mb{send} x^c_\beta\, y^c_\alpha ;P @d_1)$ is not relevant. By inversion on the typing rules $ d_1 \sqsubseteq c \sqsubseteq c'$} By definition either $d_1 \not \sqsubseteq \xi$  or none of the channels connected to $P$ including $y^c_\alpha$, and $x^c_\beta$ are relevant. In both cases, neither $\mb{msg}(\mb{send} x^c_\beta\, y^c_\alpha )$ nor $\mb{proc}(w^{c'}_\eta, [y^c_{\alpha+1}/y^c_{\alpha}]P @d_1)$  are relevant. Channel $w^{c'}_\eta$  is irrelevant in the pre-step and post-step configurations.

Channel $y^c_\alpha$  is irrelevant in the pre-step, and both $y^c_\alpha$ and $y^c_{\alpha+1}$ are irrelevant in pre-step and post-step configurations. Every other irrelevant resource of the process in the pre-step will remain irrelevant in the post-step too. See {\bf \color{ForestGreen}Case 2.} {\bf \color{Aquamarine}Subcase 1.} for the discussion on the relevancy of $\mc{C}''_1$.

\[\cproj{\mc{D}'_1\mc{C}''_1 \mb{msg}(\mb{send} x^c_\beta\, y^c_\alpha) \mb{proc}(w^{c'}_\eta, [y^c_{\alpha+1}/y^c_{\alpha}]P @d_1) \mc{D}''_1}{\xi} \meq \cproj{\mc{D}'_1 \mc{D}''_1}{\xi}\meq \cproj{\mc{D}_1}{\xi}\meq \cproj{\mc{D}_2}{\xi}\]

\item {\color{Aquamarine}\bf Subcase 2. $\mb{proc}(w^{c'}_\eta, \mb{send} x^c_\beta\, y^c_\alpha ;P @d_1)$ is relevant.}
By assumption that $\cproj{\mc{D}_1}{\xi}\meq\cproj{\mc{D}_2}{\xi}$, and definition of $\meq$:

\[\mc{D}_2=\mc{D}'_2\mc{C}''_2\mb{proc}(v^{c'}_\omega, \mb{send} u^c_\gamma\, z^c_\delta ;P_{\mid \xi} @d_1)\mc{D}''_2\]

 such that $z_\delta=y_\alpha$ and $u_\gamma=x_\beta$ if $c \sqsubseteq \xi$, and $P_{\mid \xi}$ is equal to $P$ modulo renaming of some channels with secrecy level higher than or incomparable secrecy level to the observer.
 
 We have \[\lr{\mc{C}_2}{\mc{D}_2}{\mc{F}_2} \mapsto_{\Delta \Vdash K}\lr{\mc{C}_2}{\mc{D}'_2\mc{C}''_2\mb{msg}(\mb{send} u^c_\gamma\, z^c_\delta) \mb{proc}(v^{c'}_{\omega}, [z^c_{\delta+1}/z^c_\delta]P_{\mid \xi} @d_1)\mc{D}''_2}{\mc{F}_2}\]
 
 With the same argument as in {\bf \color{ForestGreen}Case 2.} {\bf \color{Aquamarine}Subcase 2.} we can show that relevancy of $\mc{D}''_i$ remains intact.
 
 For $\mc{C}''_i$, we argue that if $ c \sqsubset \xi$, then $\mc{C}''_1=\mc{C}''_2$ is relevant in the pre-step and remains relevant in the post-step too. If $c \not \sqsubseteq \xi$, then the relevancy of $\mc{C}''_i$ remain intact from pre-step to post-step. 
 This is enough to show that $\cproj{\mc{C}''_1}{\xi}= \cproj{\mc{C}''_2}{\xi}$ in the post-step. (See {\bf \color{ForestGreen}Case 3.} {\bf \color{Aquamarine}Subcase 2.} for a more detailed discussion on transferring a tree via message)
 
The discussion on relevancy of $\mc{D}'_i$ is similar to the previous cases.

{\color{ForestGreen} \bf Case 5. $\mc{D}_1=\mc{D}'_1\mb{msg}(y^c_\alpha.k;  y_\alpha^c  \leftarrow v^c) \mb{proc}(x^d_\beta, \mb{case} \,y^c_\alpha (\ell \Rightarrow P_\ell)_{\ell \in L} @d_1)  \mc{D}''_1$ and
\[ \mc{D}'_1\mb{msg}(y^c_\alpha.k;y_\alpha^c  \leftarrow v^c)\mb{proc}(x^d_\beta, \mb{case} \,y^c_\alpha (\ell \Rightarrow P_\ell)_{\ell \in L} @d_1)  \mc{D}''_1 \mapsto_{\Delta \Vdash K} \mc{D}'_1\mb{proc}(w^d_{\beta}, [v^c/y^c_{\alpha}] P_k@c) \mc{D}''_1\]}
We consider sub-cases based on relevancy of process offering along $w^d_\beta$. Observe that  $y^c_\alpha$  is relevant if an only if $v^c$ is relevant, since they share a message of secrecy $c$.
\begin{description}
\item {\color{Aquamarine}\bf Subcase 1. $\mb{proc}(x^d_\beta, \mb{case} \,y^c_\alpha (\ell \Rightarrow P_\ell)_{\ell \in L} @d_1)$ is not relevant.} By definition either $d_1 \sqcup c \not \sqsubseteq \xi  $ or none of the channels connected to $P$ including its offering channel $y_\alpha^c$ are relevant.  In both cases the messages $\mb{msg}(y^c_\alpha.k)$ and the continuation process $\mb{proc}(w^d_{\beta}, [v^c/y^c_{\alpha}] P_k@c\sqcup d_1)$ are not relevant either. It is straightforward to see that   
\[\cproj{\mc{D}'_1\mb{proc}(w^d_{\beta}, [v^c/y^c_{\alpha}] P_k@c\sqcup d_1)\mc{D}''_1}{\xi} \meq \cproj{\mc{D}'_1 \mc{D}''_1}{\xi}\meq \cproj{\mc{D}_1}{\xi}\meq \cproj{\mc{D}_2}{\xi}.\]

\item {\color{Aquamarine}\bf Subcase 2. $\mb{proc}(x^d_\beta, \mb{case} \,y^c_\alpha (\ell \Rightarrow P_\ell)_{\ell \in L} @d_1)$ is relevant.}
By definition of relevancy, we get that $c \sqcup d_1\sqsubseteq \xi$ and thus $y^c_\alpha$ is relevant. This means that $\mb{msg}(y^c_\alpha.k; y_\alpha^c  \leftarrow v^c)$ is relevant too.
By assumption that $\cproj{\mc{D}_1}{\xi}\meq\cproj{\mc{D}_2}{\xi}$, and definition of $\meq$: \[\mc{D}_2=\mc{D}'_2\mb{msg}(y^c_\alpha.k; y_\alpha^c  \leftarrow v^c)\mb{proc}(u^d_\gamma,   \mb{case}\,y^c_\alpha( \ell \Rightarrow {P_\ell}_{\mid \xi})_{\ell\in L} @d_1)  \mc{D}''_2,\] such that ${P_\ell}_{\mid \xi}$ is equal to $P_\ell$ modulo renaming of some channels with secrecy level higher than or incomparable to the observer.

We have \[\lr{\mc{C}_2}{\mc{D}_2}{\mc{F}_2} \mapsto_{\Delta \Vdash K}\lr{\mc{C}_2}{\mc{D}'_2\mb{proc}(u^d_{\gamma}, [v^c/y^c_\alpha]{P_k}_{\mid \xi} @c) \mc{D}''_2}{\mc{F}_2}\]
This completes the proof of the subcase as we know that the relevancy of channels in $\mc{D}'_i$ and $\mc{D}''_i$ remain intact.

\end{description}

{\color{ForestGreen} \bf Case 6. $\mc{D}_1=\mc{D}'_1\mb{proc}(y^c_\alpha, \mb{case} \,y^c_\alpha (\ell \Rightarrow P_\ell)_{\ell \in L} @d_1) \mb{msg}(y^c_\alpha.k; x_1^c \leftarrow y_\alpha^c) \mc{D}''_1$ and
\[ \mc{D}'_1\mb{proc}(y^c_\alpha, \mb{case} \,y^c_\alpha (\ell \Rightarrow P_\ell)_{\ell \in L} @d_1) \mb{msg}(y^c_\alpha.k; x_1^c \leftarrow y_\alpha^c) \mc{D}''_1 \mapsto_{\Delta \Vdash K} \mc{D}'_1\mb{proc}(x_1^c, [x_1^c/y^c_{\alpha}] P_k@c) \mc{D}''_1\]}
We consider sub-cases based on relevancy of process offering along $y^c_\alpha$. Observe that $y^c_\alpha$  is relevant if an only if $x_1^c$ is relevant, since they share a message of secrecy $c$.
\begin{description}
\item {\color{Aquamarine}\bf Subcase 1. $\mb{proc}(y^c_\alpha, \mb{case} \,y^c_\alpha (\ell \Rightarrow P_\ell)_{\ell \in L} @d_1)$ is not relevant.} By definition either $d_1 \sqcup c = c \not \sqsubseteq \xi$ or none of the channels connected to $P$ including its offering channel $y_\alpha^c$ are relevant.  In both cases, means that $x_1^c$ is not relevant and $\mb{msg}(y^c_\alpha.k; x_1^c \leftarrow y_\alpha^c)$ is not relevant either.
Moreover the continuation process $\mb{proc}(x_1^c, [x_1^c/y^c_{\alpha}] P_k@c)$ won't be relevant. And   
\[\cproj{\mc{D}'_1\mb{proc}(x_1^c, [x_1^c/y^c_{\alpha}] P_k@c)\mc{D}''_1}{\xi} \meq \cproj{\mc{D}'_1 \mc{D}''_1}{\xi}\meq \cproj{\mc{D}_1}{\xi}\meq \cproj{\mc{D}_2}{\xi}\]
\item {\color{Aquamarine}\bf Subcase 2. $\mb{proc}(y^c_\alpha, \mb{case} \,y^c_\alpha (\ell \Rightarrow P_\ell)_{\ell \in L} @d_1)$ is relevant.}
By definition of relevancy, we get that $c \sqcup d_1=c \sqsubseteq \xi$ and thus $y^c_\alpha$ is relevant. This means that $\mb{msg}(y^c_\alpha.k; x_1^c \leftarrow y_\alpha^c)$ is relevant too.
By assumption that $\cproj{\mc{D}_1}{\xi}\meq\cproj{\mc{D}_2}{\xi}$, and definition of $\meq$: \[\mc{D}_2=\mc{D}'_2\mb{proc}(y^c_\alpha,   \mb{case}\,y^c_\alpha( \ell \Rightarrow {P_\ell}_{\mid \xi})_{\ell\in L} @d_1) \mb{msg}(y^c_\alpha.k; x_1^c \leftarrow y_\alpha^c) \mc{D}''_2,\] such that ${P_\ell}_{\mid \xi}$ is equal to $P_\ell$ modulo renaming of some channels with secrecy level higher than or incomparable to the observer.

We have \[\lr{\mc{C}_2}{\mc{D}_2}{\mc{F}_2} \mapsto_{\Delta \Vdash K}\lr{\mc{C}_2}{\mc{D}'_2\mb{proc}(x_1^c, [x_1^c/y^c_\alpha]{P_k}_{\mid \xi} @c) \mc{D}''_2}{\mc{F}_2}\]
This completes the proof of the subcase as we know that the relevancy of channels in $\mc{D}'_i$ and $\mc{D}''_i$ remain intact.

\end{description}

{\color{ForestGreen}\bf Case 7. $\mc{D}_1=\mc{D}'_1\mb{msg}(\mb{send}x_\eta^c y_\alpha^c; y_\alpha^c \leftarrow x_1^c) \mb{proc}(v^{c'}_\eta,  w \leftarrow \mb{recv} y^c_\alpha; P @d_1)  \mc{D}''_1$ and
\[ \mc{D}'_1\mb{msg}(\mb{send} x_\eta^c\, y_\alpha^c;y_\alpha^c \leftarrow x_1^c) \mb{proc}(v^{c'}_\eta,  w \leftarrow \mb{recv} y^c_\alpha; P @d_1)\mc{D}''_1 \mapsto_{\Delta \Vdash K} \mc{D}'_1\mb{proc}(v^{c'}_{\eta}, [x^c_\eta/w][x_1^c/y^c_{\alpha}] P@c\sqcup d_1) \mc{D}''_1\]}
We consider sub-cases based on relevancy of process offering along $v^{c'}_\eta$. 
\begin{description}
\item {\color{Aquamarine}\bf Subcase 1. $ \mb{proc}(v^{c'}_\eta,  w \leftarrow \mb{recv} y^c_\alpha; P @d_1)$ is not relevant.} By definition either $d_1 \sqcup c \not \sqsubseteq \xi$ or none of the channels connected to $P$ including $y_\alpha^c$ are relevant.

In both cases by the definition of quasi running secrecy we know that neither $\mb{msg}(\mb{send}x_\eta^c y_\alpha^c;x_1^c \leftarrow y_\alpha^c)$ nor the continuation process $\mb{proc}(v^{c'}_{\eta}, [x^c_\eta/w][x_1^c/y^c_{\alpha}] P@c\sqcup d_1)$  are relevant. It is then straightforward to see that

\[\cproj{\mc{D}'_1\mb{proc}(v^{c'}_{\eta}, [x^c_\eta/w][x_1^c/y^c_{\alpha}] P@c\sqcup d_1) \mc{D}''_1}{\xi} \meq \cproj{\mc{D}'_1 \mc{D}''_1}{\xi}\meq \cproj{\mc{D}_1}{\xi}\meq \cproj{\mc{D}_2}{\xi}.\]

\item {\color{Aquamarine}\bf Subcase 2. $\mb{proc}(v^{c'}_\eta,  w \leftarrow \mb{recv} y^c_\alpha; P @d_1)$ is relevant.}
By definition of relevancy, we get that $c \sqcup d_1\sqsubseteq \xi$. This implies that $y^c_\alpha$ is relevant in the pre-step. From relevancy of $y^c_\alpha$ and the quasi running secrecy lower than or equal to the observer of the positive message $\mb{msg}(\mb{send} x_\eta^c\, y_\alpha^c;y_\alpha^c \leftarrow x_1^c)$ we get that the message is relevant too.
By assumption:
\[\mc{D}_2=\mc{D}'_2\mb{msg}(\mb{send} x_\eta^c\, y_\alpha^c;y_\alpha^c \leftarrow x_1^c)\mb{proc}(u^{c'}_\gamma,    w \leftarrow \mb{recv} y^c_\alpha; P_{\mid \xi} @d_1)  \mc{D}''_2,\] such that ${P_\ell}_{\mid \xi}$ is equal to $P_\ell$ modulo renaming of some  channels with secrecy level higher than or incomparable to the observer.

We have \[\lr{\mc{C}_2}{\mc{D}_2}{\mc{F}_2} \mapsto_{\Delta \Vdash K}\lr{\mc{C}_2}{\mc{D}'_2\mb{proc}(u^{c'}_{\gamma}, [x^c_\eta/w][x_1^c/y^c_{\alpha}] {P}_{\mid \xi} @d_1 \sqcup c) \mc{D}''_2}{\mc{F}_2}\]
We need to consider  that the quasi running secrecy of the process may increases in the post step based on the code of $P$ and $P_{\mid_\xi}$. The argument for this case is similar to the previous cases of the proof. See {\bf \color{ForestGreen}Case 1.}{\bf \color{Aquamarine} Subcase 2.}. One interesting situation is when  the relevancy of chain of positive and relevant messages in the pre-step of $\mc{D}'_i$ changes in the post-step. By relevancy in the pre-step we know that these chains exist in both runs, so the same chain of messages will become irrelevant in the post-step of both runs.  

\end{description}

{\color{ForestGreen}\bf Case 8. $\mc{D}_1=\mc{D}'_1\mb{proc}(y^{c}_\alpha,  w \leftarrow \mb{recv} y^c_\alpha; P @d_1)\mb{msg}(\mb{send}x_\eta^c y_\alpha^c; x_1^c \leftarrow y_\alpha^c)   \mc{D}''_1$ and
\[ \mc{D}'_1 \mb{proc}(y^{c}_\alpha,  w \leftarrow \mb{recv} y^c_\alpha; P @d_1)\mb{msg}(\mb{send}x_\eta^c y_\alpha^c;x_1^c \leftarrow y_\alpha^c) \mc{D}''_1 \mapsto_{\Delta \Vdash K} \mc{D}'_1\mb{proc}(x_1^c, [x^c_\eta/w][x_1^c/y^c_{\alpha}] P@c\sqcup d_1) \mc{D}''_1\]}

We consider sub-cases based on relevancy of process offering along $y^{c}_\alpha$.
\begin{description}
\item {\color{Aquamarine}\bf Subcase 1. $\mb{proc}(y^{c}_\alpha,  w \leftarrow \mb{recv} y^c_\alpha; P @d_1)$ is not relevant.} By definition either $d_1 \sqcup c =c \not \sqsubseteq \xi $ or none of the channels connected to $P$ including $y_\alpha^c$ are relevant. In both cases the negative message $\mb{msg}(\mb{send}x_\eta^c y_\alpha^c; x_1^c \leftarrow y_\alpha^c)$  and the continuation process $\mb{proc}(x_1^c, [x^c_\eta/w][x_1^c/y^c_{\alpha}] P@c\sqcup d_1)$  are not relevant either. It is then straightforward to see that   
\[\cproj{\mc{D}'_1\mb{proc}(x_1^c, [x^c_\eta/w][x_1^c/y^c_{\alpha}] P@c\sqcup d_1) \mc{D}''_1}{\xi} \meq \cproj{\mc{D}'_1 \mc{D}''_1}{\xi}\meq \cproj{\mc{D}_1}{\xi}\meq \cproj{\mc{D}_2}{\xi}.\]

\item {\color{Aquamarine}\bf Subcase 2. $\mb{proc}(y^{c}_\alpha,  w \leftarrow \mb{recv} y^c_\alpha; P @d_1)$ is relevant.}
By definition of relevancy, we get that $c \sqcup d_1=c\sqsubseteq \xi$ and $y^c_\alpha$ is relevant. This means that $\mb{msg}(\mb{send} x_\eta^c\, y_\alpha^c; x_1^c \leftarrow y_\alpha^c)$ and the channel and $x^c_\eta$ are relevant. By assumption that $\cproj{\mc{D}_1}{\xi}\meq\cproj{\mc{D}_2}{\xi}$, and definition of $\meq$:
\[\mc{D}_2=\mc{D}'_2\mb{proc}(y^{c}_\alpha,    w \leftarrow \mb{recv} y^c_\alpha; P_{\mid \xi} @d_1)\mb{msg}(\mb{send} x_\eta^c\, y_\alpha^c; x_1^c \leftarrow y_\alpha^c)  \mc{D}''_2,\] such that ${P}_{\mid \xi}$ is equal to $P$ modulo renaming of some channels with secrecy level higher than or incomparable to the observer.

We have \[\lr{\mc{C}_2}{\mc{D}_2}{\mc{F}_2} \mapsto_{\Delta \Vdash K}\lr{\mc{C}_2}{\mc{D}'_2\mb{proc}(x_1^c, [x^c_\eta/w][x_1^c/y^c_{\alpha}] {P}_{\mid \xi} @d_1 \sqcup c) \mc{D}''_2}{\mc{F}_2}\]

The proof is similar to previous cases.

\end{description}

{\color{ForestGreen} \bf Case 9. $\mc{D}_1=\mc{D}'_1\mb{proc}(y^c_\alpha, y^c_\alpha\leftarrow x^c_\beta @d_1) \mc{D}''_1$ and
\[ \mc{D}'_1\mb{proc}(y^c_\alpha, y^c_\alpha\leftarrow x^c_\beta @d_1)  \mc{D}''_1 \mapsto_{\Delta \Vdash K} \mc{D}'_1 [x^c_{\beta}/y^c_{\alpha}]\mc{D}''_1\]}
We consider sub-cases based on relevancy of process offering along $y^c_\alpha$. Observe that $y^c_\alpha$ is relevant if and only if $x^c_\beta$ is relevant.
\begin{description}
\item {\color{Aquamarine}\bf Subcase 1. $\mb{proc}(y^c_\alpha, y^c_\alpha\leftarrow x^c_\beta @d_1)$ is not relevant.} By definition either $c \not \sqsubseteq \xi $ or none of the channels $y_\alpha^c$ and $x_\beta^c$ are relevant.  In both cases, it means that $y_\alpha^c$ and $x_\beta^c$ are not relevant. As a result, we can safely make the substitution $[x^c_\beta/y^c_{\alpha}]$ in $\mc{D}''_1$. If $c \not \sqsubseteq \xi$, it is only the matter of renaming  channels with secrecy level higher than or incomparable to the observer level. And if $c \sqsubseteq \xi$, then $y_\alpha^c$ does not occur in any relevant process and we can rename it to $x_\beta^c$. 

Moreover, deleting $\mb{proc}(y^c_\alpha, y^c_\alpha\leftarrow x^c_\beta @d_1)$ from the configuration does not decrease the quasi-running secrecy of any of its (grand)children since by the tree invariant $d_1 \sqsubseteq c$. In particular, a chain of positive messages offered along $x^c_\beta$, has its minimum quasi running secrecy of $d_1\sqsubseteq c=c$. Removing $\mb{proc}(y^c_\alpha, y^c_\alpha\leftarrow x^c_\beta @d_1)$ may increase the quasi running secrecy of such chain of messages. This case is only of significance if $c\sqsubseteq \xi$. By the assumption of subcase, we know that $x_\beta^c$ is irrelevant, and thus the chain of messages have to be irrelevant in the pre-state.

\[\cproj{\mc{D}'_1\mb{proc}(y^c_\alpha, y^c_\alpha\leftarrow x^c_\beta @d_1) \mc{D}''_1}{\xi} \meq \cproj{\mc{D}'_1 [x^c_\beta/y^c_{\alpha}]\mc{D}''_1}{\xi}\meq \cproj{\mc{D}_1}{\xi}\meq \cproj{\mc{D}_2}{\xi}\]

\item {\color{Aquamarine}\bf Subcase 2. $\mb{proc}(y^c_\alpha, y^c_\alpha\leftarrow x^c_\beta @d_1)$ is relevant.}
By definition of relevancy, we get that $c \sqsubseteq \xi$ and both $y^c_\alpha$ and $x^c_\beta$ are relevant. By assumption that $\cproj{\mc{D}_1}{\xi}\meq\cproj{\mc{D}_2}{\xi}$, and definition of $\meq$: \[\mc{D}_2=\mc{D}'_2\mb{proc}(y^c_\alpha, y^c_\alpha\leftarrow x^c_\beta @d_1) \mc{D}''_2.\] 

We have \[\lr{\mc{C}_2}{\mc{D}_2}{\mc{F}_2} \mapsto_{\Delta \Vdash K}\lr{\mc{C}_2}{\mc{D}'_2 [x^c_\beta/y^c_\beta]\mc{D}''_2}{\mc{F}_2}\]
With the same reasoning as in  {\color{Aquamarine}\bf Subcase 1.}, the quasi running secrecy of a chain of messages offered along $x^c_\beta$ is $c$ and does not decrease after deleting the process. However, it may increase based on the quasi-running secrecy of the parents of  $\mb{proc}(y^c_\alpha, y^c_\alpha\leftarrow x^c_\beta @d_1)$ in $\mc{D}''_i$. If the parent in one run has a quasi running secrecy lower than or equal to the observer level, then it has to be relevant and thus there is a counterpart in the other run with a  quasi running secrecy lower than or equal to the observer level. Thus in both runs the relevancy of the chain of messages does not change in the post-step. If both of the parents have running secrecy higher than or incomparable to the observer level, then the same chain of messages become irrelevant in both runs.
\end{description}

{\color{ForestGreen} \bf Case 10. $\mc{D}_1=\mc{D}'_1\mb{proc}(y^c_\alpha, (x^d \leftarrow P) @d_2; Q @d_1)\mc{D}''_1$ and
\[ \mc{D}'_1\mb{proc}(y^c_\alpha, (x^d \leftarrow P)@d_2; Q @d_1)  \mc{D}''_1 \mapsto_{\Delta \Vdash K} \mc{D}'_1 \mb{proc}(a^d_0, [a^d_0/x^d] P @d_2)\mb{proc}(y^c_\alpha, [a^d_0/x^d] Q @d_1) \mc{D}''_1\]}
We consider sub-cases based on relevancy of process offering along $y^c_\alpha$.
\begin{description}
\item {\color{Aquamarine}\bf Subcase 1. $\mb{proc}(y^c_\alpha, (x^d \leftarrow P)@d_2; Q @d_1)$ is not relevant.} By definition either $d_1 \not \sqsubseteq \xi $ or none of the channels of this process including $y_\alpha^c$ are relevant.  

In both cases, it means that both $\mb{proc}(a^d_0, [a^d_0/x^d] P @d_2)$ and $\mb{proc}(y^c_\alpha, [a^d_0/x^d] Q @d_1)$ are not relevant either. Note that $d_1 \sqsubseteq d_2$ and thus $d_2 \not \sqsubseteq \xi$.

\[\cproj{\mc{D}'_1\mb{proc}(a^d_0, [a^d_0/x^d] P @d_2)\mb{proc}(y^c_\alpha, [a^d_0/x^d] Q @d_1) \mc{D}''_1}{\xi} \meq \cproj{\mc{D}'_1 \mc{D}''_1}{\xi}\meq \cproj{\mc{D}_1}{\xi}\meq \cproj{\mc{D}_2}{\xi}\]

\item {\color{Aquamarine}\bf Subcase 2. $\mb{proc}(y^c_\alpha, (x^d \leftarrow P)@d_2; Q @d_1)$ is relevant.}
By definition of relevancy, we get that $d_1 \sqsubseteq \xi$ and all  channels  of this process with secrecy levels lower than or equal to the observer level are relevant. By assumption that $\cproj{\mc{D}_1}{\xi}\meq\cproj{\mc{D}_2}{\xi}$, and definition of $\meq$: \[\mc{D}_2=\mc{D}'_2\mb{proc}(z^c_\delta, (x^d \leftarrow P_{\mid \xi})@d_2; Q_{\mid \xi} @d_1) \mc{D}''_2.\] 

We have \[\lr{\mc{C}_2}{\mc{D}_2}{\mc{F}_2} \mapsto_{\Delta \Vdash K}\lr{\mc{C}_2}{\mc{D}'_2 \mb{proc}(a^d_0, [a^d_0/x^d] P_{\mid\xi} @d_2)\mb{proc}(z^c_\delta, [a^d_0/x^d] Q_{\mid\xi} @d_1)\mc{D}''_2}{\mc{F}_2}\]
Remark:we can assume that the fresh channel being spawned will be $a$ in both runs.

Note that if $a^d_0$ has secrecy level lower than or equal to the observer level, then taking this step won't change relevancy of any channels. Otherwise some resource of the process may become irrelevant after this step since $a^d_0$ may block their relevancy path or in the case where $d_2\not \sqsubseteq \xi$ the process becomes irrelevant. But this happens to the processes in the both runs. (similar to the cases 3 and 4 for $\otimes$ and $\multimap$)
\end{description}
\end{proof}

\lemref{buildmn} indicates how two related configurations ${(\mc{B}_1, \mc{B}_2) \in \mc{E}_{\Psi}^\xi\llbracket \Delta \Vdash K \rrbracket}$ can be broken down into  $\lre{\mc{C}_i}{\mc{D}_i}{\mc{F}_i}$ such that $\ctype{\Psi}{\cdot}{\mc{C}_i}{\Delta}$, and $\ctype{\Psi}{\Delta}{\mc{D}_i}{K}$, and $\ctype{\Psi}{K}{\mc{F}_i}{\cdot}$.  \figref{fig:meganode_formation} illustrates interesting key cases, indicating that trees rooted at non-observable channels are internalized into $\mc{D}_i$.

\begin{lemma}[Build a Meganode] \label{buildmn}
Consider $\mc{C}_1\mc{D}_1{\mc{F}_1}$, and ${\mc{C}_2}{\mc{D}_2}{\mc{F}_2}$ such that  $\Psi; \cdot \Vdash \mc{C}_1 :: \Delta_1$ and 
$\Psi; \cdot \Vdash \mc{C}_2 :: \Delta_2$, and $\Psi; \Delta_1 \Vdash \mc{D}_1 :: x^c_{\alpha}{:}A$, and
$\Psi; \Delta_2 \Vdash \mc{D}_2 :: y^d_{\beta}{:}B$, and $\Psi; x^c_{\alpha}{:}A\Vdash \mc{F}_1 :: \cdot$, and
$\Psi; y^d_{\beta}{:}B \Vdash \mc{F}_2 :: \cdot$.

If $\cproj{\Delta_1}{\xi}=\cproj{\Delta_2}{\xi}=\Delta$ and $\cproj{x^c_{\alpha}{:}A}{\xi}=\cproj{y^d_{\beta}{:}B}{\xi}=K$, then we can rewrite  $\mc{C}_i\mc{D}_i\mc{F}_i$ as
$\mc{C}'_i\mc{D}'_i\mc{F}'_i$    such that \[{(\lr{\mc{C}'_1}{\mc{D}'_1}{\mc{F}'_1}; \lr{\mc{C}'_2}{\mc{D}'_2}{\mc{F}'_2}) \in \m{Tree}_\Psi(\Delta \Vdash K)}.\] Moreover, if $\cproj{\mc{D}_1}{\xi}=\cproj{\mc{D}_2}{\xi}$, then $\cproj{\mc{D}'_1}{\xi}=\cproj{\mc{D}'_2}{\xi}$.
\end{lemma}

\begin{proof}
We break down the proof into the following cases based on the structure of $K$ and $\Delta$.
\begin{description}
\item {\bf Case 1.} $\cproj{x^c_{\alpha}{:}A}{\xi}=\cproj{y^d_{\beta}{:}B}{\xi}=K \neq \cdot$. This means that $x^c_{\alpha}{:}A=y^d_{\beta}{:}B$ and $c\sqsubseteq \xi$. By the tree invariant, we know that $\Delta_1=\Delta_2=\Delta.$ So without any rewrite configurations $\mc{C}_i\mc{D}_i\mc{F}_i$ satisfy the properties that we are looking for.

\item {\bf Case 2.} $\cproj{x^c_{\alpha}{:}A}{\xi}=\cproj{y^d_{\beta}{:}B}{\xi}=\cdot$, and $\Delta_1= \Lambda_1, \Delta$, and $\Delta_2=\Lambda_2, \Delta$. By typing of configurations we have $\mc{C}_i= \mc{T}_i\mc{T}'_i$ such that $\Psi; \cdot \Vdash \mc{T}_1 :: \Delta$, and 
$\Psi; \cdot \Vdash \mc{T}'_1 :: \Lambda_1$, and $\Psi; \cdot \Vdash \mc{T}_2 :: \Delta$, and $\Psi; \cdot \Vdash \mc{T}'_2 :: \Lambda_2$.

From the definition of projections, we have $c \not \sqsubseteq \xi$ and $d \not \sqsubseteq \xi$ and for every $u^{c'}_\gamma\in \Lambda_1, \Lambda_2$, we know that $c' \not \sqsubseteq \xi$.
We build $\mc{D}'_i= \mc{T}'_i\mc{D}_i\mc{F}_i$, and $\mc{C}'_i=\mc{T}_i$, and $\mc{F}'_i=\cdot$. By the typing rules, we know that $\Psi; \cdot \Vdash \mc{F}'_i ::\cdot$ and $\Psi; \cdot \Vdash \mc{T}_i ::\Delta$ and $\Psi; \Delta\Vdash  \mc{D}'_i:: \cdot$ as we need to establish 

\[{(\lr{\mc{C}'_1}{\mc{D}'_1}{\mc{F}'_1}; \lr{\mc{C}'_2}{\mc{D}'_2}{\mc{F}'_2}) \in \m{Tree}_\Psi(\Delta \Vdash K)}.\]

Moreover, by the definition of projections, since  $c \not \sqsubseteq \xi$ and $d \not \sqsubseteq \xi$, we know that none of the processes in $\mc{F}_i$ will be relevant in $\mc{D}'_i$. Also, adding $\mc{F}_1$ as a parent of the process/message offering along  $x^c_\alpha$, does not switch relevancy of any process in $\mc{D}_1$, since we already know that $c \not \sqsubseteq \xi$. In particular if we have a message offering along $x^c_\alpha$, its quasi running secrecy is higher than or incomparable to the observer level before adding $\mc{F}_1$ as its parent and will stay higher than or incomparable to the observer level after too. The same reasoning goes with $\mc{F}_2$ and $y^c_{\beta}$. 

Similarly, since $c' \not \sqsubseteq \xi$ for every $u^{c'}_\gamma\in \Lambda_1, \Lambda_2$, and the fact that $\mc{T}'_i$ does not use any resources in $\Delta$, we know that all trees in $\mc{T}'_i$ are irrelevant. As a result, we have $\cproj{\mc{D}'_1}{\xi}= \cproj{\mc{D}_1}{\xi} = \cproj{\mc{D}_2}{\xi}= \cproj{\mc{D}'_2}{\xi}$.

\end{description}

\end{proof}

\begin{figure}
\begin{center}
\includegraphics[scale=0.53]{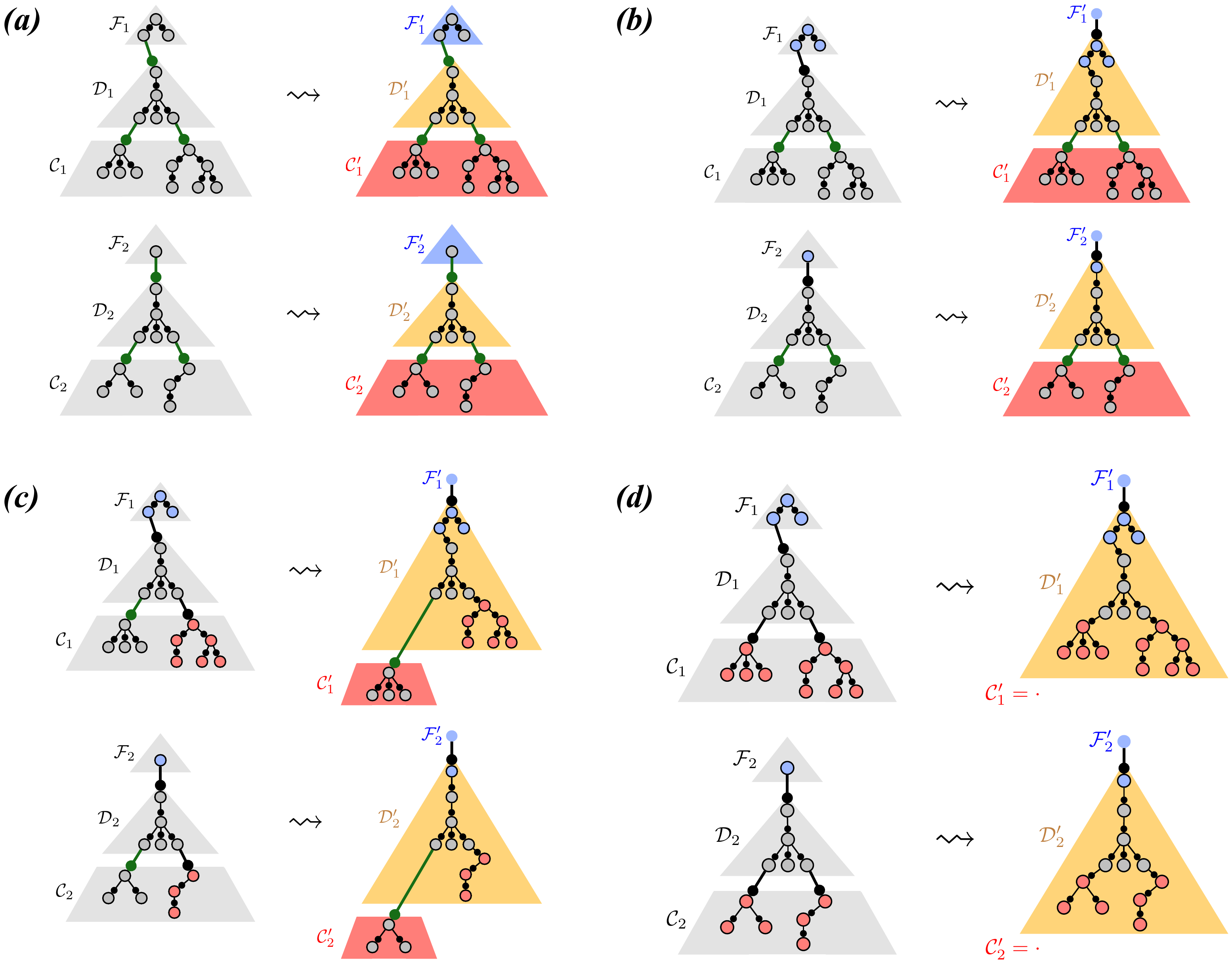}
\caption{Schematic illustration of \lemref{buildmn}.  Observable channels are depicted in green.}
\label{fig:meganode_formation}
\end{center}
\end{figure}

\begin{lemma}[Diamond Property]\label{lem:diamond}
If $\lr{\mc{C}'_1}{\mc{D}'_1}{\mc{F}'_1}\in \m{Tree}(\Delta \Vdash K)$ and  $\lr{\mc{C}_1}{\mc{D}_1}{\mc{F}_1}\mapsto_{\Delta \Vdash K}\lr{\mc{C}'_1}{\mc{D}'_1}{\mc{F}'_1}$ and $\lr{\mc{C}_1}{\mc{D}_1}{\mc{F}_1} \mapsto_{\Delta \Vdash K} \lr{\mc{C}''_1}{\mc{D}''_1}{\mc{F}''_1}$, and $\lr{\mc{C}'_1}{\mc{D}'_1}{\mc{F}'_1} \neq \lr{\mc{C}''_1}{\mc{D}''_1}{\mc{F}''_1}$ then there is a configuration $\lr{\mc{C}}{\mc{D}}{\mc{F}}$ such that 
$\lr{\mc{C}'_1}{\mc{D}'_1}{\mc{F}'_1} \mapsto_{\Delta \Vdash K} \lr{\mc{C}}{\mc{D}}{\mc{F}}$, and $\lr{\mc{C}''_1}{\mc{D}''_1}{\mc{F}''_1} \mapsto_{\Delta \Vdash K} \lr{\mc{C}}{\mc{D}}{\mc{F}}$
\end{lemma}
\begin{proof}
The proof is straightforward by cases.
\end{proof}

Confluence of multistep reduction follows by two standard inductions from the diamond property.

\begin{lemma}[Forward closure]\label{lem:fwdC} If $(\lre{\mc{C}_1}{\mc{D}_1}{\mc{F}_1}, \lre{\mc{C}_2}{\mc{D}_2}{\mc{F}_2}) \in \mc{E}^\xi_\Psi \llbracket \Delta \Vdash K\rrbracket$, and $(\lr{\mc{C}_1}{\mc{D}_1}{\mc{F}_1}; \lr{\mc{C}_2}{\mc{D}_2}{\mc{F}_2})\in \m{Tree}_{\Psi}(\Delta \Vdash K)$, and $\lr{\mc{C}_1}{\mc{D}_1}{\mc{F}_1} \mapsto_{\Delta \Vdash K}\lr{\mc{C}'_1}{\mc{D}'_1}{\mc{F}'_1}$\, then 
$(\lre{\mc{C}'_1}{\mc{D}'_1}{\mc{F}'_1}, \lre{\mc{C}_2}{\mc{D}_2}{\mc{F}_2}) \in \mc{E}^\xi_\Psi \llbracket \Delta \Vdash K\rrbracket$.
\end{lemma}
\begin{proof}
By definition of $\mc{E}^\xi_\Psi$, we know that $(\lr{\mc{C}_1}{\mc{D}_1}{\mc{F}_1}; \lr{\mc{C}_2}{\mc{D}_2}{\mc{F}_2})\in \m{Tree}_{\Psi}(\Delta \Vdash K),$ and
  \[\lr{\mc{C}_1}{\mc{D}_1}{\mc{F}_1}\mapsto^*_{\Delta \Vdash K} \lr{\mc{C}''_1}{\mc{D}''_1}{\mc{F}''_1} \,\m{and}\,  \lr{\mc{C}_2}{\mc{D}_2}{\mc{F}_2}\mapsto^*_{\Delta \Vdash K} \lr{\mc{C}''_2}{\mc{D}''_2}{\mc{F}''_2}.\]
 
 By confluence, we have 
   \[\star \; \lr{\mc{C}'_1}{\mc{D}'_1}{\mc{F}'_1}\mapsto^*_{\Delta \Vdash K} \lr{\mc{C}''_1}{\mc{D}''_1}{\mc{F}''_1}.\]
   
By \thmref{thm:preservation}, we get
   \[\star'\; (\lr{\mc{C}'_1}{\mc{D}'_1}{\mc{F}'_1}; \lr{\mc{C}_2}{\mc{D}_2}{\mc{F}_2})\in \m{Tree}_{\Psi}(\Delta \Vdash K).\]
$\star$ and $\star'$ completes the proof.
\end{proof}

\begin{lemma}[Backward Closure]\label{lem:backC} If $(\lre{\mc{C}_1}{\mc{D}_1}{\mc{F}_1}, \lre{\mc{C}_2}{\mc{D}_2}{\mc{F}_2}) \in \mc{E}^\xi_\Psi \llbracket \Delta \Vdash K\rrbracket$, and $(\lr{\mc{C}'_1}{\mc{D}'_1}{\mc{F}'_1}; \lr{\mc{C}_2}{\mc{D}_2}{\mc{F}_2})\in \m{Tree}_{\Psi}(\Delta \Vdash K)$, and $\lr{\mc{C}'_1}{\mc{D}'_1}{\mc{F}'_1} \mapsto_{\Delta \Vdash K}\lr{\mc{C}_1}{\mc{D}_1}{\mc{F}_1}$\, then 
$(\lre{\mc{C}'_1}{\mc{D}'_1}{\mc{F}'_1}, \lre{\mc{C}_2}{\mc{D}_2}{\mc{F}_2}) \in \mc{E}^\xi_\Psi \llbracket \Delta \Vdash K \rrbracket$.
\end{lemma}
\begin{proof}
By definition of $\mc{E}^\xi_\Psi$, we know that $(\lr{\mc{C}_1}{\mc{D}_1}{\mc{F}_1}; \lr{\mc{C}_2}{\mc{D}_2}{\mc{F}_2})\in \m{Tree}_{\Psi}(\Delta \Vdash K),$ and
  \[\lr{\mc{C}_1}{\mc{D}_1}{\mc{F}_1}\mapsto^*_{\Delta \Vdash K} \lr{\mc{C}''_1}{\mc{D}''_1}{\mc{F}''_1} \,\m{and}\,  \lr{\mc{C}_2}{\mc{D}_2}{\mc{F}_2}\mapsto^*_{\Delta \Vdash K} \lr{\mc{C}''_2}{\mc{D}''_2}{\mc{F}''_2}.\]
 
 By definition of $\mapsto^*_{\Delta\Vdash K}$, we have 
   \[\star \; \lr{\mc{C}'_1}{\mc{D}'_1}{\mc{F}'_1}\mapsto^*_{\Delta \Vdash K}\lr{\mc{C}_1}{\mc{D}_1}{\mc{F}_1}\mapsto^*_{\Delta \Vdash K} \lr{\mc{C}''_1}{\mc{D}''_1}{\mc{F}''_1} \,\m{and}\,  \lr{\mc{C}_2}{\mc{D}_2}{\mc{F}_2}\mapsto^*_{\Delta \Vdash K} \lr{\mc{C}''_2}{\mc{D}''_2}{\mc{F}''_2}.\]
 $\star$ along with the assumption of $(\lr{\mc{C}'_1}{\mc{D}'_1}{\mc{F}'_1}; \lr{\mc{C}_2}{\mc{D}_2}{\mc{F}_2})\in \m{Tree}_{\Psi}(\Delta \Vdash K)$ completes the proof.
\end{proof}

\begin{theorem}[Fundamental Theorem]\label{thm:main}
For all security levels $\xi$, and configurations ${\Psi; \Delta_1 \Vdash \mathcal{D}_1:: u_\alpha^{c_1} {:}T_1}$ and ${\Psi; \Delta_2 \Vdash \mathcal{D}_2:: v_\beta^{c_2} {:}T_2}$  with $\cproj{\mathcal{D}_1}{\xi} \meq \cproj{\mathcal{D}_2}{\xi}$, $\Delta_1 \Downarrow \xi = \Delta_2 \Downarrow \xi$, and $u_\alpha^{c_1} {:}T_1\Downarrow \xi = v_\beta^{c_2} {:}T_2 \Downarrow \xi$  we have \[(\Delta_1 \Vdash \mathcal{D}_1:: u_\alpha^{c_1} {:}T_1)  \equiv^\Psi_{\xi} (\Delta_2 \Vdash \mathcal{D}_2:: v_\beta^{c_2} {:}T_2).\]  
\end{theorem}
\begin{proof}


Put  $\Delta=\cproj{\Delta_1}{\xi} = \cproj{\Delta_2}{\xi}$, and $K=\cproj{u_\alpha^{c_1} {:}T_1}{\xi} = \cproj{v_\beta^{c_2} {:}T_2}{\xi}$. The goal is to prove that  for all $\mathcal{C}_1, \mathcal{C}_2$ and $\mathcal{F}_1, \mathcal{F}_2$ such that
${\ctype{\Psi}{\cdot}{\mc{C}_1}{\Delta_1}}$
and 
${\ctype{\Psi}{\cdot}{\mc{C}_2}{\Delta_2}}$
and 
${\ctype{\Psi}{u^{c_1}_{\alpha}{:}T_1}{\mc{F}_1}{\cdot}}$
and 
${\ctype{\Psi}{v^{c_1}_{\beta}{:}T_2}{\mc{F}_2}{\cdot}}$
\noindent {\color{Aquamarine} $\;\dagger \; (\mathcal{C}_1\mathcal{D}_1 \mathcal{F}_1;\mathcal{C}_2\mathcal{D}_2 \mathcal{F}_2) \in \mathcal{E}^\xi_\Psi\llbracket \Delta \Vdash K \rrbracket.$}

The proof is by induction on a lexicographic order on (a) the size of types $\Delta, K$ and (b) the multiset $\mc{M}$ of derivations $\ptype{\Psi}{\Delta'}{P}{y}{d}{\gamma},{T'}$ used in the derivations of $\ctype{\Psi}{\cdot}{\mc{C}_1\mc{D}_1\mc{F}_1}{\cdot}$ and $\ctype{\Psi}{\cdot}{\mc{C}_2\mc{D}_2\mc{F}_2}{\cdot}$. Derivations are ordered in the standard way and we use the multiset ordering derived from this as the basis of our induction, using $<$ to present it. 
 
By \lemref{buildmn} we can rewrite  $\mc{C}_1\mc{D}_1\mc{F}_1$ as $\mc{C}'_1\mc{D}'_1\mc{F}'_1$ and $\mc{C}_2\mc{D}_2\mc{F}_2$ as $\mc{C}'_2\mc{D}'_2\mc{F}'_2$ such that ${(\lr{\mc{C}'_1}{\mc{D}'_1}{\mc{F}'_1}; \lr{\mc{C}'_2}{\mc{D}'_2}{\mc{F}'_2}) \in \m{Tree}_\Psi(\Delta \Vdash K)}$. Moreover, $\cproj{\mc{D}'_1}{\xi}=\cproj{\mc{D}'_2}{\xi}$.

By progress, either 1) at least one of the configurations can take a step, i.e. ${\lr{\mc{C}'_1}{\mc{D}'_1}{\mc{F}'_1} \mapsto_{\Delta\Vdash K} \lr{\mc{C}''_1}{\mc{D}''_1}{\mc{F}''_1}}$ or $\lr{\mc{C}'_2}{\mc{D}'_2}{\mc{F}'_2} \mapsto_{\Delta\Vdash K} \lr{\mc{C}''_2}{\mc{D}''_2}{\mc{F}''_2}$ such that $({\mc{C}''_i}{\mc{D}''_i}{\mc{F}''_i})<{\mc{C}'_i}{\mc{D}'_i}{\mc{F}'_i}$  in the multiset ordering, or 2)  $\lr{\mathcal{C}'_1}{\mathcal{D}'_1}{ \mathcal{F}'_1}$ and $\lr{\mathcal{C}'_2}{\mathcal{D}'_2} {\mathcal{F}'_2}$ are poised with regard to $\Delta \Vdash K$, or 3) both $\mathcal{C}'_1\mathcal{D}'_1 \mathcal{F}'_1$ and $\mathcal{C}'_2\mathcal{D}'_2 \mathcal{F}'_2$ are empty with $\Delta = K = \cdot$.


The proof of Case 3) is trivial by definition. We consider 1) and 2) separately. Here, without loss of generality, we only consider the cases for ${\lr{\mc{C}'_1}{\mc{D}'_1}{\mc{F}'_1} \mapsto_{\Delta\Vdash K} \lr{\mc{C}''_1}{\mc{D}''_1}{\mc{F}''_1}}$  and ${\lr{\mc{C}'_1}{\mc{D}'_1}{\mc{F}'_1} }$  being poised. 
\begin{enumerate}
 \item {\color{ForestGreen} \bf  $\lr{\mc{C}'_1}{\mc{D}'_1}{\mc{F}'_1} \mapsto_{\Delta\Vdash K} \lr{\mc{C}''_1}{\mc{D}''_1}{\mc{F}''_1}$}:

By \lemref{lem:indinvariant},  $\lr{\mc{C}'_2}{\mc{D}'_2}{\mc{F}'_2}\mapsto^{0,1}_{\Delta\Vdash K}\lr{\mc{C}''_2}{\mc{D}''_2}{\mc{F}''_2}$ with 
$\Delta \Vdash {\mc{D}''_2}::K$, and $\cproj{\mc{D}''_1}{\xi}=\cproj{\mc{D}''_2}{\xi}$. 


We can apply the induction hypothesis on  $(\mc{C}''_1 \mc{D}''_1 \mc{F}''_1 ;\mathcal{C}''_2\mathcal{D}''_2 \mathcal{F}''_2)  < (\mc{C}_1\mc{D}_1\mc{F}_1, \mc{C}_2\mc{D}_2\mc{F}_2)$ to get 
{\color{Aquamarine}  \[\star\;\;  (\mc{C}''_1 \mc{D}''_1 \mc{F}''_1 ;\mathcal{C}''_2\mathcal{D}''_2 \mathcal{F}''_2) \in \mathcal{E}^\xi_\Psi\llbracket \Delta \Vdash K\rrbracket.\]}
By Backward closure lemma (\lemref{lem:backC}):

{\color{Aquamarine}  \[\dagger \;\; (\mc{C}'_1 \mc{D}'_1 \mc{F}'_1 ;\mathcal{C}_2\mathcal{D}_2 \mathcal{F}_2) \in \mathcal{E}^\xi_\Psi\llbracket \Delta \Vdash K\rrbracket.\]}

{\bf \item \color{ForestGreen} $\lr{\mc{C}'_1}{\mc{D}'_1}{\mc{F}'_1}$ and $\lr{\mc{C}'_2}{\mc{D}'_2}{\mc{F}'_2}$ are both poised with regard to $\Delta \Vdash K$:}\\
{{\bf Case 1.} \color{RedViolet} $K=u_\alpha^{c_1}{:}T_1=u_\alpha^{c_1}{:}1$}, and $\Delta= \cdot$, and $\lr{\mc{C}'_1}{\mc{D}'_1}{\mc{F}'_1}= \lr{\cdot}{\msg{\mb{close}\,u^{c_1}_\alpha}}{\mc{F}'_1}$.
Note that by the definition of relevancy, all processes in $\mc{D}'_1$ and $\mc{D}'_2$ are relevant.
As a result, $\mathcal{D}'_{2}=\mathbf{msg}(\mb{close}\,u^{c_1}_\alpha),$ and $\mc{C}'_2=\cdot$.
    By line 1 of the definition, we get 
$\dagger \, (\lr{\mc{C}'_1}{\mc{D}'_1}{\mc{F}'_1}, \lr{\mc{C}'_2}{\mc{D}'_2}{\mc{F}'_2}) \in \mc{E}^\xi_\Psi\llbracket \cdot \Vdash u^{c_1}_\alpha{:}1\rrbracket$.\\

  {\bf Case 2.} {\color{RedViolet} $K=u_\alpha^{c_1}{:}T_1=u_\alpha^{c_1}{:}\oplus\{\ell{:}A_{\ell}\}_{\ell \in I},$}  and $\mathcal{D}'_{1}=\mc{D}''_1\mathbf{msg}(u_\alpha^{c_1}.k).$
    By the assumption of the theorem, $\mc{D}'_1$ and $\mc{D}'_2$ are both relevant, and we have $\mathcal{D}'_{2}=\mc{D}''_2\mathbf{msg}(u_\alpha^{c_1}.k; u_\alpha^{c_1} \leftarrow w_\gamma^{c_1}).$ 
 
Removing $\mb{msg}(u^{c_1}_\alpha.k)$ from $\mc{D}'_1$ and $\mc{D}'_2$ does not change relevancy of the remaining configuration when $c_1 \sqsubseteq \xi$: \[\cproj{\mc{D}''_1}{\xi}=\cproj{\mc{D}''_2}{\xi}.\]
 
We can apply the induction hypothesis on $\dagger_1$ and $\star_2$ and the smaller type $K$ to get 
  \[\begin{array}{l}
{{\color{Aquamarine}}\star}  (\lre{\mc{C}'_1}{\mc{D}''_1} {\msg{u^{c_1}_\alpha.k; u_\alpha^{c_1} \leftarrow w_\gamma^{c_1}}\mc{F}'_1} ,\lre{\mathcal{C}'_2}{\mathcal{D}''_2} {\msg{u^{c_1}_\alpha.k; u_\alpha^{c_1} \leftarrow w_\gamma^{c_1}}\mathcal{F}'_2)} \in \mathcal{E}^\xi_\Psi\llbracket \Delta \Vdash w^{c_1}_{\gamma}{:}A_k\rrbracket.
\end{array}\]

By line (2) in the definition of $\mathcal{V}^\xi_\Psi$:
{\color{Aquamarine}  \[\begin{array}{l}
 \dagger  (\mc{C}'_1 \mc{D}''_1 \msg{u^{c_1}_\alpha.k; u_\alpha^{c_1} \leftarrow w_\gamma^{c_1}}\mc{F}'_1 ;\mathcal{C}'_2\mathcal{D}''_2 \msg{u^{c_1}_\alpha.k; u_\alpha^{c_1} \leftarrow w_\gamma^{c_1}}\mathcal{F}'_2)\in \mathcal{E}^\xi_\Psi\llbracket \Delta \Vdash u_\alpha^{c_1}{:}\oplus\{\ell:A_{\ell}\}_{\ell \in I}\rrbracket.
\end{array}
\]}\\

 {{\bf Case 3.} \color{RedViolet} $K=u_\alpha^{c_1}{:}T_1=u_\alpha^{c_1}{:}\&\{\ell{:}A_{\ell}\}_{\ell \in I}$}, and $\lr{\mc{C}'_1}{\mc{D}'_1}{\mc{F}'_1}= \lr{\mc{C}'_1}{\mc{D}'_1}{\msg{u^{c_1}_\alpha.k; w_\gamma^{c_1}\leftarrow u^{c_1}_\alpha}\mc{F}''_1}$.

We consider two subcases:
\begin{description}
 \item{\bf Subcase 1.}${\mc{F}'_2\neq \msg{u^{c_1}_\alpha.k; w_\gamma^{c_1}\leftarrow u^{c_1}_\alpha}\mc{F}''_2}.$
 By line (3) in the definition of $\mathcal{V}^\xi_\Psi$:
{\color{Aquamarine}
\[\begin{array}{l}
\dagger  (\mc{C}'_1 \mc{D}'_1 \msg{u^{c_1}_\alpha.k; w_\gamma^{c_1}\leftarrow u^{c_1}_\alpha
}\mc{F}''_1 ;\mathcal{C}'_2\mathcal{D}_2 \mathcal{F}''_2) \in \mathcal{E}^\xi_\Psi\llbracket \Delta \Vdash u_\alpha^{c_1}{:}\&\{\ell:A_{\ell}\}_{\ell \in I}\rrbracket.
\end{array}\]}\\
 \item{\bf Subcase 2.}${\mc{F}'_2= \msg{u^{c_1}_\alpha.k; w_\gamma^{c_1}\leftarrow u^{c_1}_\alpha}\mc{F}''_2}.$
 
 We can apply the induction hypothesis on the smaller type $K$, but first we need to show that the invariant of the induction holds. From $\cproj{\mc{D}'_1}{\xi}=\cproj{\mc{D}'_2}{\xi}$ and  $c_1 \sqsubseteq \xi$, we get \[\cproj{\mc{D}'_1\msg{u^{c_1}_\alpha.k; w_\gamma^{c_1}\leftarrow u^{c_1}_\alpha
}}{\xi}=\cproj{\mc{D}'_2\msg{u^{c_1}_\alpha.k; w_\gamma^{c_1}\leftarrow u^{c_1}_\alpha
}}{\xi}.\]

 By induction hypothesis
 \[\begin{array}{l}
        {\color{Aquamarine}\star} (\lre{\mc{C}'_1}{\mc{D}'_1 \msg{u^{c_1}_\alpha.k; w_\gamma^{c_1}\leftarrow u^{c_1}_\alpha
}} {\mc{F}''_1} ,\lre{\mathcal{C}'_2}{\mathcal{D}'_2 \msg{u^{c_1}_\alpha.k; w_\gamma^{c_1}\leftarrow u^{c_1}_\alpha
}} {\mathcal{F}''_2)} \in \mathcal{E}^\xi_\Psi\llbracket \Delta \Vdash w^{c_1}_{\gamma}{:}A_k\rrbracket.
\end{array}\]
By line (3) in the definition of $\mathcal{V}^\xi_\Psi$:
{\color{Aquamarine}
\[\begin{array}{l}
\dagger  (\mc{C}'_1 \mc{D}'_1 \msg{u^{c_1}_\alpha.k; w_\gamma^{c_1}\leftarrow u^{c_1}_\alpha
}\mc{F}''_1 ;\mathcal{C}'_2\mathcal{D}'_2 \msg{u^{c_1}_\alpha.k; w_\gamma^{c_1}\leftarrow u^{c_1}_\alpha
}\mathcal{F}''_2) \in \mathcal{E}^\xi_\Psi\llbracket \Delta \Vdash u_\alpha^{c_1}{:}\&\{\ell:A_{\ell}\}_{\ell \in I}\rrbracket.
\end{array}\]}\\

\end{description}

 {\bf Case 4.} {\color{RedViolet} $K=y_\alpha^{c_1}{:}T_1=u_\alpha^{c_1}{:} A \otimes B$}, and
    \[\mathcal{D}'_1=\mathcal{D}''_1\mc{T}_1\mathbf{msg}(\mb{send}x_\beta^{c_1}\,u_\alpha^{c_1}; u_\alpha^{c_1} \leftarrow w_\gamma^{c_1})\]
     where $\Delta=\Lambda_1, \Lambda_2$ and $\Psi; \Lambda_2\Vdash \mc{T}_1::(x_\beta^{c_1}{:}A)$.
     
     By assumption of the theorem and $ c_1 \sqsubseteq \xi$, we know $\mc{D}'_1$ and $\mc{D}'_2$ are both relevant and \[{\mathcal{D}'_{2}=\mc{D}''_2\mc{T}_2\mathbf{msg}(\mb{send}x_\beta^{c_1}\, u_\alpha^{c_1}; u_\alpha^{c_1} \leftarrow w_\gamma^{c_1}),}\] such that $\Psi; \Lambda_2\Vdash \mc{T}_2::(x_\beta^{c_1}{:}A)$.Moreover, by relevancy of $\mc{D}'_i$ (and relevancy of $x_\beta^{c_1}$) we get $\cproj{\mc{D}''_1}{\xi}=\cproj{\mc{D}''_2}{\xi}$ and 
  $\cproj{\mc{T}_1}{\xi}=\cproj{\mc{T}_2}{\xi}$. 
     
      By configuration typing  we can break down $\mc{C}'_i$ into $\mc{H}_i$ and $\mc{H}'_i$ such that  $\ctype{\Psi}{\cdot}{\mc{H}_i}{\Lambda_1}$, and $\ctype{\Psi}{\cdot}{\mc{H}'_i}{\Lambda_2}$. 
In the case where $x^{c_1}_\beta{:}A \in \Delta$, we have $\mc{T}_i=\cdot$ and $\Lambda_2=x^{c_1}_\beta{:}A$.

 We can apply the induction hypothesis on the smaller type  to get 

\[\begin{array}{l}
{\color{Aquamarine}\star'}  (\lre{\mc{H}'_1}{\mc{T}_1} {\mc{H}_1\mc{D}''_1\msg{\mb{send}x_\beta^{c_1}\,u_\alpha^{c_1}; u_\alpha^{c_1} \leftarrow w_\gamma^{c_1}}\mc{F}'_1}, \lre{\mathcal{H}'_2}{\mathcal{T}_2} {\mc{H}'_2\mc{D}''_1\msg{\mb{send}x_\beta^{c_1}\,u_\alpha^{c_1}; u_\alpha^{c_1} \leftarrow w_\gamma^{c_1}}\mathcal{F}'_2)} \in \mathcal{E}^\xi_\Psi\llbracket \Lambda_2 \Vdash x^{c_1}_{\beta}{:}A\rrbracket.
\end{array}\]

and

\[\begin{array}{l}
{\color{Aquamarine}\star'}  (\lre{\mc{H}_1}{\mc{D}''_1} {\mc{H}'_1\mc{T}_1\msg{\mb{send}x_\beta^{c_1}\,u_\alpha^{c_1}; u_\alpha^{c_1} \leftarrow w_\gamma^{c_1}}\mc{F}'_1}, \lre{\mathcal{H}_2}{\mathcal{D}''_2} {\mc{H}'_2\mc{T}_2\msg{\mb{send}x_\beta^{c_1}\,u_\alpha^{c_1}; u_\alpha^{c_1} \leftarrow w_\gamma^{c_1}}\mathcal{F}'_2)} \in \mathcal{E}^\xi_\Psi\llbracket \Lambda_1 \Vdash w^{c_1}_{\gamma}{:}B\rrbracket.
\end{array}\]

By line (5) in the definition of $\mathcal{V}^\xi_\Psi$:
{\color{Aquamarine} 
\[\begin{array}{l}
    \dagger  (\mc{C}'_1 \mc{D}''_1 \mc{T}_1\msg{\mb{send}x_\beta^{c_1}\,u_\alpha^{c_1};u_\alpha^{c_1} \leftarrow w_\gamma^{c_1}}\mc{F}'_1 ;\mathcal{C}'_2\mathcal{D}''_2 \mc{T}_2\msg{\mb{send}x_\beta^{c_1}\,u_\alpha^{c_1};u_\alpha^{c_1} \leftarrow w_\gamma^{c_1}}\mathcal{F}'_2) \in \mathcal{E}^\xi_\Psi\llbracket \Delta \Vdash u_\alpha^{c_1}{:}A \otimes B\rrbracket.
\end{array}\]}

{\bf Case 5.}  {\color{RedViolet} $K=u_\alpha^{c_1}{:}T_1=u_\alpha^{c_1}{:}A \multimap B$}
 and $\lr{\mc{C}'_1}{\mc{D}'_1}{\mc{F}'_1}= \lr{\mc{C}'_1}{\mc{D}'_1}{\mc{T}_1\msg{\mb{send} x^{c_1}_{\beta}\,u^{c_1}_{\alpha};  w_\gamma^{c_1} \leftarrow u_\alpha^{c_1}}\mc{F}''_1}$.
There are two subcases to consider:
 \begin{description}
  \item {\bf Subcase 1.} ${\mc{F}'_2 \neq \mc{T}_2\msg{\mb{send} x^{c_1}_{\beta}\,u^{c_1}_{\alpha}}\mc{F}''_2,}$. By line (4) in the definition of $\mathcal{V}^\xi_\Psi$:
{\color{Aquamarine}
\[\begin{array}{l}
\dagger  (\mc{C}'_1 \mc{D}'_1 \mc{T}_1\msg{\mb{send}x^{c_1}_{\beta}\, u^{c_1}_{\alpha};w_\gamma^{c_1} \leftarrow u_\alpha^{c_1}}\mc{F}''_1 ;\mathcal{C}'_2\mathcal{D}'_2 \mathcal{F}'_2)\in \mathcal{E}^\xi_\Psi\llbracket \Delta \Vdash u_\alpha^{c_1}{:}A \multimap B\rrbracket.
\end{array}\]}
  
  \item{\bf Subcase 2.} ${\mc{F}'_2 = \mc{T}_2\msg{\mb{send}x^{c_1}_{\beta}\, u^{c_1}_{\alpha}}\mc{F}''_2,}$.
  
By assumption and $c_1 \sqsubseteq \xi$, we know that $\mc{D}'_1$ and $\mc{D}'_2$ are relevant. Adding a negative message along $c_1 \sqsubseteq \xi$ as the root does not change their relevancy: \[\cproj{\mc{D}'_1\msg{\mb{send}x^{c_1}_{\beta}\, u^{c_1}_{\alpha}}}{\xi}=\cproj{\mc{D}'_2\msg{\mb{send}x^{c_1}_{\beta}\, u^{c_1}_{\alpha}}}{\xi}\]

We can apply the induction hypothesis on the smaller type:
\[\begin{array}{l}{\color{Aquamarine}  \star}\;\;
         (\lre{\mc{C}'_1 \mc{T}_1}{\mc{D}'_1 \msg{\mb{send}x^{c_1}_{\beta}\, u^{c_1}_{\alpha};w_\gamma^{c_1} \leftarrow u_\alpha^{c_1}}} {\mc{F}''_2} ,  \lre{\mathcal{C}'_2\mc{T}_2}{\mathcal{D}'_2 \msg{\mb{send}x^{c_1}_{\beta}\, u^{c_1}_{\gamma};w_\gamma^{c_1} \leftarrow u_\alpha^{c_1}}} {\mathcal{F}''_2)} \in \mathcal{E}^\xi_\Psi\llbracket \Delta, x_\beta^{c_1}{:}A \Vdash w^{c_1}_{\gamma}{:}B\rrbracket.
\end{array}\]

By line (4) in the definition of $\mathcal{V}^\xi_\Psi$:
{\color{Aquamarine}
\[\begin{array}{l}
\dagger  (\mc{C}'_1 \mc{D}'_1 \mc{T}_1\msg{\mb{send} x^{c_1}_{\beta}\,u^{c_1}_{\alpha};w_\gamma^{c_1} \leftarrow u_\alpha^{c_1}}\mc{F}''_1 ;\mathcal{C}'_2\mathcal{D}'_2 \mc{T}_2\msg{\mb{send} x^{c_1}_{\beta}\,u^{c_1}_{\alpha};w_\gamma^{c_1} \leftarrow u_\alpha^{c_1}}\mathcal{F}''_2)\in \mathcal{E}^\xi_\Psi\llbracket \Delta \Vdash u_\alpha^{c_1}{:}A \multimap B\rrbracket.
\end{array}\]}\\

 \end{description}

  {\bf Case 6.} {\color{RedViolet} $K=u_{\alpha}^{c_1}{:}A$, and $\mathcal{D}'_{1}=\mc{D}''_1\mathbf{proc}(u^{c_1}_{\alpha},u^{c_1}_{\alpha}\leftarrow w^{c_1}_{\gamma} @d_1).$}
  
    By assumption of the theorem, $\mathcal{D}'_{2}=\mc{D}''_2 \mathbf{proc}(u^{c_1}_{\alpha},u^{c_1}_{\alpha}\leftarrow w^{c_1}_{\gamma} @d_1).$

 By assumption of the theorem, $\cproj{\mc{D}''_1}{\xi}=\cproj{\mc{D}''_2}{\xi}$, since $w^{c_1}_\gamma$ is relevant in $\mc{D}''_i$, and
 we can apply the induction hypothesis on $\dagger_1$ and $\star_2$ and the smaller configuration with respect to $<$ order to get 
 
   \[\begin{array}{l}
{\color{Aquamarine}\star\;}  (\lre{\mc{C}'_1}{\mc{D}''_1} {[w_{\gamma}^{c_1}/u_{\alpha}^{c_1}]\mc{F}'_1} ,\lre{\mathcal{C}'_2}{\mathcal{D}''_2} {[w_{\gamma}^{c_1}/u_{\alpha}^{c_1}]\mathcal{F}'_2)}\in \mathcal{E}^\xi_\Psi\llbracket \Delta \Vdash w^{c_1}_{\gamma}{:}A\rrbracket.
\end{array}\]

By line (11) in the definition of $\mathcal{V}^\xi_\Psi$:
{\color{Aquamarine} 
\begin{equation}
\hspace{-0.5cm}
\begin{array}{l}
 \dagger  (\mc{C}'_1 \mc{D}''_1 \mathbf{proc}(u^{c_1}_{\alpha},u^{c_1}_{\alpha}\leftarrow w^{c_1}_{\gamma} @d_1)\mc{F}'_1 ;\mathcal{C}'_2\mathcal{D}''_2 \mathbf{proc}(u^{c_1}_{\alpha},u^{c_1}_{\alpha}\leftarrow w^{c_1}_{\gamma} @d_1)\mathcal{F}'_2) \in \mathcal{E}^\xi_\Psi\llbracket \Delta \Vdash u_\alpha^{c_1}{:}A\rrbracket.
\end{array}
\nonumber
\end{equation}
 }

{{\bf Case 7.} \color{RedViolet} $\Delta=\Delta', x_\gamma^{c}{:}1$}, and $\lr{\mc{C}'_1}{\mc{D}'_1}{\mc{F}'_1}= \lr{\mc{C}''_1\msg{\mb{close}\,x^{c}_\gamma}}{\mc{D}'_1}{\mc{F}'_1}$.

There are two subcases to consider: 

\begin{description}
 \item{\bf Subcase 1.} $\mathcal{C}'_{2}\neq\mc{C}''_2\mathbf{msg}(\mb{close}\,x^{c}_\gamma)$.
 
 By line 6 in the definition of $\mathcal{V}^\xi_\Psi$,  
{\color{Aquamarine} \[\dagger \,({\mc{C}''_1\msg{\mb{close}\,x^{c}_\gamma}}{\mc{D}'_1}{\mc{F}'_1}, {\mc{C}'_2}{\mc{D}'_2}{\mc{F}'_2}) \in \mathcal{E}^\xi_\Psi\llbracket \Delta', x^c_\gamma{:}1 \Vdash K\rrbracket\]}
 \item{\bf Subcase 2.} $\mathcal{C}'_{2}= \mc{C}''_2\mathbf{msg}(\mb{close}\,x^{c}_\gamma)$.
 
 We first briefly explain why the invariant of induction holds after we bring the closing message inside $\mc{D}'_i$ and remove the channel $x_\gamma^c$ from $\Delta$, i.e.
\[\cproj{\msg{\mb{close}\,x^{c}_\gamma}\mc{D}'_1}{\xi}= \cproj{\msg{\mb{close}\,x^{c}_\gamma}\mc{D}'_2}{\xi}. \]

If the parent of $\msg{\mb{close}\,x^{c}_\gamma}$ in $\mc{D}'_1$ is relevant in  $\mc{D}'_1$ before bringing the message insider then by the assumption of theorem, it is the same as the parent of $\msg{\mb{close}\,x^{c}_\gamma}$ in $\mc{D}'_2$. If after adding  the message $\mc{D}'_1$ the parent still remains relevant, it means that it has at least one other relevant channel other than $x^c_\gamma$ in both $\mc{D}'_1$ which also exists in $\mc{D}'_2$ and will be relevant after adding the message to $\mc{D}'_2$. If after adding the message $\mc{D}'_1$ becomes irrelevant, it means that it does not have a relevant path to any other channel in $\Delta'$ and $K$. By the assumption of theorem the parent of the message in $\mc{D}'_2$ does not have such path either. The same argument holds for any other node that becomes irrelevant because of adding the closing message to $\mc{D}'_i$. As a result the same processes becomes irrelevant in both $\mc{D}'_1$ and $\mc{D}'_2$ after adding the message to them and the proof of this case is complete. The same argument holds for the case in which the parent of $\msg{\mb{close}\,x^{c}_\gamma}$ in $\mc{D}'_2$ before adding the message is relevant. Otherwise the parent of  $\msg{\mb{close}\,x^{c}_\gamma}$ is irrelevant in $\mc{D}'_1$ and $\mc{D}'_2$ before adding the message and remains irrelevant after that too. The proof in this case is straightforward.

Now that the invariant holds, we can apply the induction hypothesis on the smaller types: 
\[{\color{Aquamarine} \star}\; (\lre{\mc{C}''_1}{\msg{\mb{close}\,x^{c}_\gamma}\mc{D}'_1}{\mc{F}'_1}, \lre{\mc{C}''_2}{\msg{\mb{close}\,x^{c}_\gamma}\mc{D}'_2}{\mc{F}'_2}) \in \mathcal{E}^\xi_\Psi\llbracket \Delta' \Vdash K\rrbracket.\]

By line 6 in the definition of $\mathcal{V}^\xi_\Psi$,  
{\color{Aquamarine} \[\dagger \,({\mc{C}''_1\msg{\mb{close}\,x^{c}_\gamma}}{\mc{D}'_1}{\mc{F}'_1}, {\mc{C}''_2\msg{\mb{close}\,x^{c}_\gamma}}{\mc{D}'_2}{\mc{F}'_2}) \in \mathcal{E}^\xi_\Psi\llbracket \Delta', x^c_\gamma{:}1 \Vdash K\rrbracket\]}

\end{description}

{{\bf Case 8.} \color{RedViolet} $\Delta=\Delta', x_\gamma^{c}{:}\oplus\{\ell:A_\ell\}_{\ell \in L}$}, and $\lr{\mc{C}'_1}{\mc{D}'_1}{\mc{F}'_1}= \lr{\mc{C}''_1\msg{x^{c}_\gamma.k; x_\gamma^{c} \leftarrow w_\delta^{c}}}{\mc{D}'_1}{\mc{F}'_1}$.

 By configuration typing, we have $\mc{C}'_1= \mc{T}_i\mc{T}'_i$ such that $\ctype{\Psi}{\cdot}{\mc{T}_i}{\Delta'}$ and  $\ctype{\Psi}{\cdot}{\mc{T}'_i}{x_\gamma^{c}{:}\oplus\{\ell:A_\ell\}_{\ell \in L}}$. We know that $\mc{T}'_1=\mc{T}''_1 \mb{msg}(x^c_\gamma.k; x_\gamma^{c} \leftarrow w_\delta^{c})$.
 We consider two subcases:
 \begin{description}
  \item{\bf Subcase 1.}${\mc{T}'_2\neq \mc{T}''_2\msg{x^{c}_\gamma.k; x_\gamma^{c} \leftarrow w_\delta^{c}}}$.
  
By line (7) in the definition of $\mathcal{V}^\xi_\Psi$:
{\color{Aquamarine}
\[\begin{array}{l}
\dagger  (\mc{T}_1\mc{T}''_1\msg{x^{c}_\gamma.k; x_\gamma^{c} \leftarrow w_\delta^{c}} \mc{D}'_1\mc{F}'_1 ;\mc{T}_2\mc{T}'_2\mathcal{D}'_2 \mathcal{F}''_2)\in \mathcal{E}^\xi_\Psi\llbracket \Delta', x_\gamma^{c}{:}\oplus\{\ell:A_{\ell}\}_{\ell \in L}\Vdash K\rrbracket.
\end{array}\]}
  \item{\bf Subcase 2.} ${\mc{T}'_2= \mc{T}''_2\msg{x^{c}_\gamma.k; x_\gamma^{c} \leftarrow w_\delta^{c}}}$.

 We have $\cproj{\msg{x^{c}_\gamma.k; x_\gamma^{c} \leftarrow w_\delta^{c}}\mc{D}'_1}{\xi}=\cproj{\msg{x^{c}_\gamma.k; x_\gamma^{c} \leftarrow w_\delta^{c}}\mc{D}'_2}{\xi}$, since $w^c_{\delta}$ is relevant and the positive messages $\msg{x^{c}_\gamma.k; x_\gamma^{c} \leftarrow w_\delta^{c}}$ in both configurations are relevant only if their parents are. Thus adding the message does not change relevancy of any other process.

We can apply the induction hypothesis on the smaller type to get 
 \[\begin{array}{l}
{\color{Aquamarine}\star\;}         (\lre{\mc{T}_1\mc{T}''_1}{\msg{x^{c}_\gamma.k; x_\gamma^{c} \leftarrow w_\delta^{c}}\mc{D}'_1} {\mc{F}'_1} ,\lre{\mc{T}_2\mc{T}''_2}{\msg{x^{c}_\gamma.k; x_\gamma^{c} \leftarrow w_\delta^{c}}\mathcal{D}'_2} {\mathcal{F}'_2)} \in \mathcal{E}^\xi_\Psi\llbracket \Delta', w^c_{\delta}:A_k\Vdash K\rrbracket.
\end{array}\]

By line (7) in the definition of $\mathcal{V}^\xi_\Psi$:
{\color{Aquamarine}
\[\begin{array}{l}
\dagger  (\mc{T}_1\mc{T}''_1\msg{x^{c}_\gamma.k; x_\gamma^{c} \leftarrow w_\delta^{c}} \mc{D}'_1\mc{F}'_1 ;\mc{T}_2\mc{T}''_2\msg{x^{c}_\gamma.k; x_\gamma^{c} \leftarrow w_\delta^{c}}\mathcal{D}'_2 \mathcal{F}''_2)\in \mathcal{E}^\xi_\Psi\llbracket \Delta', x_\gamma^{c}{:}\oplus\{\ell:A_{\ell}\}_{\ell \in L}\Vdash K\rrbracket.
\end{array}\]}\\

 \end{description}

{{\bf Case 9.} \color{RedViolet} $\Delta=\Delta', x_\gamma^{c}{:}\&\{\ell:A_\ell\}_{\ell \in L}$}, and $\lr{\mc{C}'_1}{\mc{D}'_1}{\mc{F}'_1}= \lr{\mc{C}'_1}{\msg{x^{c}_\gamma.k; w_\delta^{c}\leftarrow x^{c}_\gamma}\mc{D}''_1}{\mc{F}'_1}$.

  By the assumption of the theorem, $\mathcal{D}'_{2}=\msg{x^{c}_\gamma.k; w_\delta^{c}\leftarrow x^{c}_\gamma}\mc{D}''_2.$ Moreover,  $\cproj{\mc{D}''_1}{\xi}=\cproj{\mc{D}'_2}{\xi}$: $c\sqsubseteq \xi$ and thus $x^c_{\gamma+1}$ remains relevant in $\mc{D}_i$ and no relevancy changes in the configurations after removing the negative message.
 
We can apply the induction hypothesis on the smaller type to get 
 \[\begin{array}{l}
{\color{Aquamarine} \star}  (\lre{\mc{C}'_1\msg{x^{c}_\gamma.k; w_\delta^{c}\leftarrow x^{c}_\gamma}}{\mc{D}''_1} {\mc{F}'_1} ,\lre{\mathcal{C}'_2\msg{x^{c}_\gamma.k; w_\delta^{c}\leftarrow x^{c}_\gamma}}{\mathcal{D}''_2} {\mathcal{F}'_2)} \in \mathcal{E}^\xi_\Psi\llbracket \Delta', x^{c}_{\gamma+1}{:}A_k\Vdash K\rrbracket.
\end{array}\]

By line (3) in the definition of $\mathcal{V}^\xi_\Psi$:
{\color{Aquamarine}  \[\begin{array}{l}
 \dagger  (\mc{C}'_1 \msg{x^{c}_\gamma.k; w_\delta^{c}\leftarrow x^{c}_\gamma} \mc{D}''_1 \mc{F}'_1 ;\mathcal{C}'_2 \msg{x^{c}_\gamma.k; w_\delta^{c}\leftarrow x^{c}_\gamma} \mathcal{D}''_2 \mathcal{F}'_2)\in \mathcal{E}^\xi_\Psi\llbracket \Delta', x^{c}_{\gamma}{:}\&\{\ell{:}A_\ell\}_{\ell\in L}\Vdash K \rrbracket.
\end{array}
\]}\\
 
{\bf Case 10.
\color{RedViolet} $\Delta=\Delta', x_\gamma^{c}{:}A \otimes B$}, and $\lr{\mc{C}'_1}{\mc{D}'_1}{\mc{F}'_1}= \lr{\mc{C}''_1\mc{T}_1\msg{\mb{send}y_\delta^c\,x_\gamma^{c};x_\gamma^c \leftarrow w_\eta^c}  }{\mc{D}'_1}{\mc{F}'_1}$.

We consider two subcases:
\begin{description}
 \item {\bf Subcase 1.} $\mc{C}'_2 \neq \mc{C}''_2\mc{T}_2\msg{\mb{send}y_\delta^c\,x_\gamma^{c};x_\gamma^c \leftarrow w_\eta^c}$.

By line (7) in the definition of $\mc{V}^\xi_\Psi$:
{\color{Aquamarine}
\[\begin{array}{l}
\dagger\; (\mc{C}''_1\mc{T}_1 \msg{\mb{send}y_\delta^c\,x_\gamma^{c};x_\gamma^c \leftarrow w_\eta^c} \mc{D}'_1\mc{F}'_1 ;\mathcal{C}'_2\mathcal{D}'_2 \mathcal{F}'_2) \in \mathcal{E}^\xi_\Psi\llbracket \Delta', x_\gamma^{c}{:}A \otimes B\Vdash K\rrbracket.
\end{array}\]}
 \item {\bf Subcase 2.}  $\mc{C}'_2 = \mc{C}''_2\mc{T}_2\msg{\mb{send}y_\delta^c\,x_\gamma^{c};x_\gamma^c \leftarrow w_\eta^c}$.

Moreover, we have \[\msg{\mb{send}y_\delta^c\,x_\gamma^{c};x_\gamma^c \leftarrow w_\eta^c}\cproj{\mc{D}'_1}{\xi}\meq\cproj{\msg{\mb{send}y_\delta^c\,x_\gamma^{c};x_\gamma^c \leftarrow w_\eta^c}\mc{D}'_2}{\xi}:\] The quasi-running secrecy of  $\msg{\mb{send}y_\delta^c\,x_\gamma^{c};x_\gamma^c \leftarrow w_\eta^c}$ is lower than or equal to the observer level if the quasi-running secrecy of its parent is lower than or equal to the observer level. So the relevancy of the parent of the message and thus the rest of configurations do not change by adding the message to the configuration.

We can apply the induction hypothesis on  the smaller type to get 
 \[\begin{array}{l}
        {\color{Aquamarine}\star\;} (\lre{\mc{C}''_1\mc{T}_1}{\msg{\mb{send}y_\delta^c\,x_\gamma^{c};x_\gamma^c \leftarrow w_\eta^c}\mc{D}'_1} {\mc{F}'_1} ,  \lre{\mathcal{C}''_2\mc{T}_2}{\msg{\mb{send}y_\delta^c\,x_\gamma^{c};x_\gamma^c \leftarrow w_\eta^c}\mathcal{D}'_2} {\mathcal{F}'_2)} \in  \mathcal{E}^\xi_\Psi\llbracket \Delta',  y^c_\delta{:}A, w^c_{\eta}{:}B\Vdash K\rrbracket.
\end{array}\]

By line (7) in the definition of $\mathcal{V}^\xi_\Psi$:
{\color{Aquamarine}
\[\begin{array}{l}
\dagger\; (\mc{A}_1\mc{T}'_1 \mc{T}_1\msg{\mb{send}y_\delta^c\,x_\gamma^{c};x_\gamma^c \leftarrow w_\eta^c} \mc{D}'_1\mc{F}'_1 ;\mathcal{A}_2\mc{T}'_2\mc{T}_2\msg{\mb{send}y_\delta^c\,x_\gamma^{c};x_\gamma^c \leftarrow w_\eta^c}\mathcal{D}'_2 \mathcal{F}''_2) \in \mathcal{E}^\xi_\Psi\llbracket \Delta', x_\gamma^{c}{:}A \otimes B\Vdash K\rrbracket.
\end{array}\]}\\

\end{description}
 
{\bf Case 11. 
 \color{RedViolet} $\Delta=\Delta', \Delta'', x_\gamma^{c}{:}A \multimap B$}. And $\lr{\mc{C}'_1}{\mc{D}'_1}{\mc{F}'_1}= \lr{\mc{C}'_1}{\mc{T}_1\msg{\mb{send} y^c_\delta\,x^{c}_\gamma; w_\eta^c \leftarrow x_\gamma^c }\mc{D}''_1}{\mc{F}'_1}$, such that $\ctype{\Psi}{\Delta''}{\mc{T}_1}{y^c_\delta{:A}}$ and $\ctype{\Psi}{\Delta',w^c_{\eta}{:B}}{\mc{D}''_1}{K}$.

  The message $\msg{\mb{send}y^c_\delta\,x^{c}_\gamma ; w_\eta^c \leftarrow x_\gamma^c }$ is relevant in $\mc{D}'_1$. By assumption of the theorem, $\mathcal{D}'_{2}=\mc{T}_2\msg{\mb{send}y^c_\delta\,x^{c}_\gamma }\mc{D}''_2,$  such that $\ctype{\Psi}{\Delta''}{\mc{T}_2}{y^c_\delta{:A}}$ and $\ctype{\Psi}{\Delta',w^c_{\eta}{:}B}{\mc{D}''_2}{K}$.

 By typing configuration, we have $\mc{C}'_i=\mc{A}_i \mc{A}'_i\mc{A}''_i$ such that $\cdot \Vdash \mc{A}_i :: \Delta'$ and $\cdot \Vdash \mc{A}'_i :: \Delta'', x^c_\gamma {:} A \multimap B$.

 By assumption we know that  $\cproj{\mc{D}'_1}{\xi}=\cproj{\mc{D}'_2}{\xi}$. By definition of relevancy, we know that $\msg{\mb{send}y^c_{\delta}\, x^c_\gamma ; w_\eta^c \leftarrow x_\gamma^c}$ and the tree $\mc{T}_1$ connected to it are relevant in $\mc{D}_1$ and thus $\mc{T}_1$is equal to $\mc{T}_2$ in $\mc{D}_2$. Removing these from both configurations does not change relevancy of the rest of the configuration since $w^c_{\eta}$ will remain relevant, and the relevancy of the message's parent does not change: $\cproj{\mc{D}''_1}{\xi}=\cproj{\mc{D}''_2}{\xi}$.
 
We can apply the induction hypothesis on  the smaller type to get 
 \[\begin{array}{l}
{\color{Aquamarine} \star\,}  (\lre{\mc{A}'_1}{\mc{T}_1}{\mc{A}_1\msg{\mb{send}y^c_\delta\,x^{c}_\gamma ; w_\eta^c \leftarrow x_\gamma^c}\mc{D}''_1 \mc{F}'_1} ,\lre{\mathcal{A}'_2}{\mc{T}_2}{\mc{A}_2\msg{\mb{send} y^c_\delta\, x^{c}_\gamma; w_\eta^c \leftarrow x_\gamma^c}\mathcal{D}''_2\mathcal{F}'_2)}\in \mathcal{E}^\xi_\Psi\llbracket \Delta'' \Vdash y^c_\delta{:}A\rrbracket.
\end{array}\]

and

  \[\begin{array}{l}
{\color{Aquamarine} \star'\,}  (\lre{\mc{A}_1\mc{A}'_1 \mc{T}_1\msg{\mb{send}y^c_\delta\,x^{c}_\gamma ; w_\eta^c \leftarrow x_\gamma^c}}{\mc{D}''_1} {\mc{F}'_1} ,\lre{\mathcal{A}_2\mc{A}'_2\mc{T}_2\msg{\mb{send}y^c_\delta\,x^{c}_\gamma ; w_\eta^c \leftarrow x_\gamma^c}}{\mathcal{D}''_2} {\mathcal{F}'_2)}\in \mathcal{E}^\xi_\Psi\llbracket \Delta', w^{c}_{\eta}{:}B\Vdash K\rrbracket.
\end{array}\]

By line (10) in the definition of $\mathcal{V}^\xi_\Psi$:
{\color{Aquamarine}  \[\begin{array}{l}
 \dagger  (\mc{C}'_1\mc{T}_1\msg{\mb{send} y^c_\delta\,x^{c}_\gamma; w_\eta^c \leftarrow x_\gamma^c} \mc{D}''_1 \mc{F}'_1 ;\mathcal{C}'_2 \mc{T}_2\msg{\mb{send} y^c_\delta\,x^{c}_\gamma; w_\eta^c \leftarrow x_\gamma^c} \mathcal{D}''_2 \mathcal{F}'_2) \in \mathcal{E}^\xi_\Psi\llbracket \Delta'_1, x^{c}_{\gamma}{:}A \multimap B\Vdash K \rrbracket.
\end{array}
\]}\\

{\bf Case 12. $\Delta=\Delta', x_{\delta}^{c}{:}A$.}
$\mathcal{C}'_{1}=\mc{C}''_1\mathbf{proc}(x^{c}_{\delta},x^{c}_{\delta}\leftarrow w^{c}_{\gamma} @d_1).$
  
We consider two subcases:
\begin{description}
 \item {\bf Subcase 1.} $\mathcal{C}'_{2}\neq\mc{C}''_2\mathbf{proc}(x^{c}_{\delta},x^{c}_{\delta}\leftarrow w^{c}_{\gamma} @d_1)$.
 By line (12) in the definition of $\mathcal{V}^\xi_\Psi$:
{\color{Aquamarine} 
\begin{equation}
\hspace{-0.5cm}
\begin{array}{l}
 \dagger  (\mc{C}'_1 \mc{D}''_1 \mathbf{proc}(x^{c}_{\delta},x^{c}_{\delta}\leftarrow w^{c}_{\gamma} @d_1)\mc{F}'_1 ;\mathcal{C}_2\mathcal{D}''_2 \mathcal{F}'_2) \in \mathcal{E}^\xi_\Psi\llbracket \Delta, x^{c}_{\delta}{:A} \Vdash K\rrbracket.
\end{array}
\nonumber
\end{equation}
 }

  \item {\bf Subcase 2.} $\mathcal{C}'_{2}=\mc{C}''_2\mathbf{proc}(x^{c}_{\delta},x^{c}_{\delta}\leftarrow w^{c}_{\gamma} @d_1)$.
  
  By assumption of the theorem, we have $\cproj{[w^{c}_{\gamma}/x^{c}_\delta]\mc{D}'_1}{\xi}=\cproj{[w^{c}_{\gamma}/x^{c}_\delta]\mc{D}'_2}{\xi}$, since we just rename a relevant channel in both configurations. We can apply the induction hypothesis on  the smaller typing judgment of the multiset to get 

  \[\begin{array}{l}
{\color{Aquamarine}\star\;}  (\lre{\mc{C}''_1}{[w_{\gamma}^{c}/u_{\delta}^{c}]\mc{D}'_1} {\mc{F}'_1} ,\lre{\mathcal{C}''_2}{[w_{\gamma}^{c}/x_{\delta}^{c}]\mathcal{D}''_2} {\mathcal{F}'_2)} \in \mathcal{E}^\xi_\Psi\llbracket \Delta', x^{c}_{\delta}{:}A  \Vdash K\rrbracket.
\end{array}\]

By line (12) in the definition of $\mathcal{V}^\xi_\Psi$:
{\color{Aquamarine} 
\begin{equation}
\hspace{-0.5cm}
\begin{array}{l}
 \dagger  (\mc{C}'_1 \mc{D}''_1 \mathbf{proc}(x^{c}_{\delta},x^{c}_{\delta}\leftarrow w^{c}_{\gamma} @d_1)\mc{F}'_1 ;\mathcal{C}'_2\mathcal{D}''_2 \mathbf{proc}(x^{c}_{\delta},x^{c}_{\delta}\leftarrow w^{c}_{\gamma} @d_1)\mathcal{F}'_2) \in \mathcal{E}^\xi_\Psi\llbracket \Delta, x^{c}_{\delta}{:A} \Vdash K\rrbracket.
\end{array}
\nonumber
\end{equation}
 }

\end{description}
 \end{enumerate}
\end{proof}

 \begin{figure*}
    \centering
   {\small
\[\begin{array}{llr}
\m{Property} & \m{Condition}\\
\hline\\
(1)\,\lr{\mc{C}}{\mc{D}}{\mc{F}} \hookrightarrow_{\Delta \Vdash K} \m{queue}& \lr{\mc{C}}{\mc{D}}{\mc{F}}\, \m{is\,not\, poised}\, \m{and}\, \lr{\mc{C}}{\mc{D}}{\mc{F}} \in \m{Tree}(\Delta \Vdash K) \, \m{and}\\ & \lr{\mc{C}}{\mc{D}}{\mc{F}} \mapsto^{\m{poised}}_{\Delta \Vdash K}\lr{\mc{C}''}{\mc{D}''}{\mc{F}''}\,\m{and} \, \lr{\mc{C}''}{\mc{D}''}{\mc{F}''} \hookrightarrow_{\Delta \Vdash K} \m{queue}\\
\textit{For 2-14 we assume that } \lr{\mc{C}}{\mc{D}}{\mc{F}}\, \m{is\, poised}:\\\\
(2)\, \lr{\cdot}{\mb{msg}(\mb{close}\,y^c_\alpha)}{\mc{F}} \hookrightarrow_{\cdot \Vdash y^c_{\alpha}:1}{\mb{msg}(\mb{close}\,y^c_\alpha)}& \lr{\cdot}{\mb{msg}(\mb{close}\,y^c_\alpha)}{\mc{F}} \in \m{Tree}(\cdot\Vdash y_\alpha^c{:}1) \\
(3)\, \lr{\mc{C}'\mb{msg}(\mb{close}\,x^{c'}_\beta)}{\mc{D}_1}{\mc{F}} &  \lr{\mc{C}'\mb{msg}(\mb{close}\,x^{c'}_\beta)}{\mc{D}}{\mc{F}} \in \m{Tree}(\Delta', x^{c'}_\beta{:}1\Vdash \Delta)\, \m{and}  \, \\
\qquad \hookrightarrow_{\Delta', x^{c'}_\beta{:}1 \Vdash K}\m{queue}'\overline{\mb{msg}(\mb{close}\,x^{c'}_\beta)}& \lr{\mc{C}'}{\mb{msg}(\mb{close}\,x^{c'}_\beta)\mc{D}}{\mc{F}}\hookrightarrow_{\Delta' \Vdash K} \m{queue}'  \\
(4)\, \lr{\mc{C}}{\mc{D}'\mb{msg}(y^c_\alpha.k;  y^c_\alpha \leftarrow u^c_\delta)}{\mc{F}}&  \lr{\mc{C}}{\mc{D}'\mb{msg}(y^c_\alpha.k;  y^c_\alpha \leftarrow u^c_\delta)}{\mc{F}} \in \m{Tree}(\Delta\Vdash y_\alpha^c{:}\oplus\{\ell:A_\ell\}_{\ell \in L})\\  \qquad \hookrightarrow_{\Delta\Vdash y_\alpha^c{:}\oplus\{\ell:A_\ell\}_{\ell \in L}}{\m{queue}'\mb{msg}(y^c_\alpha.k;  y^c_\alpha \leftarrow u^c_\delta)}&\m{and}\, \lr{\mc{C}}{\mc{D}'}{\mb{msg}(y^c_\alpha.k;  y^c_\alpha \leftarrow u^c_\delta)\mc{F}} \hookrightarrow_{\Delta \Vdash u^c_{\delta}:A_k}{\m{queue}'} \\
 (5)\, \lr{\mc{C}_1\mb{msg}(x^{c'}_\beta.k;  x^{c'}_\beta \leftarrow u^{c'}_\delta)}{\mc{D}'_1}{\mc{F}_1}& \lr{\mc{C}\mb{msg}(x^{c'}_\beta.k;  x^{c'}_\beta \leftarrow u^{c'}_\delta)}{\mc{D}'}{\mc{F}}\in \m{Tree}(\Delta, x_\beta^{c'}{:}\oplus\{\ell:A_\ell\}_{\ell \in L} \Vdash K)\\ \qquad \hookrightarrow_{\Delta, x_\beta^{c'}{:}\oplus\{\ell:A_\ell\}_{\ell \in L} \Vdash K}{\m{queue}'\overline{\mb{msg}(x^{c'}_\beta.k;  x^{c'}_\beta \leftarrow u^{c'}_\delta)}}& \m{and}\, \lr{\mc{C}}{\mb{msg}(x^{c'}_\beta.k;  x^{c'}_\beta \leftarrow u^{c'}_\delta)\mc{D}'}{\mc{F}} \hookrightarrow_{\Delta, u^{c'}_{\delta}:A_k \Vdash K}\m{queue'} \\
(6)\,\lr{\mc{C}}{\mc{D}'}{\mb{msg}(y^c_\alpha.k; u^c_\delta \leftarrow y^c_\alpha)\mc{F}}&  \lr{\mc{C}}{\mc{D}'}{\mb{msg}(y^c_\alpha.k; u^c_\delta \leftarrow y^c_\alpha )\mc{F}} \in \m{Tree}(\Delta\Vdash y_\alpha^c{:}\&\{\ell:A_\ell\}_{\ell \in L})  \\  \qquad \hookrightarrow_{\Delta\Vdash y_\alpha^c{:}\&\{\ell:A_\ell\}_{\ell \in L}}{\m{queue}'\overline{\mb{msg}(y^c_\alpha.k; u^c_\delta \leftarrow y^c_\alpha)}}& \m{and}\, \lr{\mc{C}}{\mc{D}'\mb{msg}(y^c_\alpha.k; u^c_\delta \leftarrow y^c_\alpha)}{\mc{F}} \hookrightarrow_{\Delta \Vdash u^c_{\delta}:A_k}{\m{queue}'}\\

(7) \,\lr{\mc{C}}{\mb{msg}(x^{c'}_\beta.k;u^{c'}_\delta \leftarrow x^{c'}_\beta)\mc{D}'}{\mc{F}}& \lr{\mc{C}}{\mb{msg}(x^{c'}_\beta.k;  u^{c'}_\delta \leftarrow x^{c'}_\beta)\mc{D}'}{\mc{F}}\in \m{Tree}(\Delta, x_\beta^{c'}{:}\&\{\ell:A_\ell\}_{\ell \in L} \Vdash K)\\ \qquad \hookrightarrow_{\Delta, x_\beta^{c'}{:}\&\{\ell:A_\ell\}_{\ell \in L} \Vdash K}{\m{queue}'\mb{msg}(y^c_\alpha.k;u^{c'}_\delta \leftarrow x^{c'}_\beta)} &  \m{and} \, \lr{\mc{C}\mb{msg}(x^{c'}_\beta.k;u^{c'}_\delta \leftarrow x^{c'}_\beta)}{\mc{D}'}{\mc{F}} \hookrightarrow_{\Delta, u^{c'}_{\delta}:A_k \Vdash K}\m{queue'} \\

(8)\, \lr{\mc{C}'\mc{C}''}{\mc{D}'\mc{T}\mb{msg}(\mb{send}\,z^c_\gamma\,y^{c}_\alpha ; y^c_\alpha\leftarrow u^c_\delta)}{\mc{F}} &  \lr{\mc{C}'\mc{C}''}{\mc{D}'\mc{T}\mb{msg}(\mb{send}\, z^c_\gamma\,y^{c}_\alpha; y^c_\alpha\leftarrow u^c_\delta)}{\mc{F}_1}\in \m{Tree}(\Delta_1,\Delta_2 \Vdash y^c_{\alpha}: A \otimes B) \\\qquad \hookrightarrow_{\Delta \Vdash y^c_{\alpha}{:}A \otimes B}\m{queue}'\m{queue}''\mb{msg}(\mb{send}\, z^c_\gamma\,y^{c}_\alpha; y^c_\alpha\leftarrow u^c_\delta) &\m{and}\,  \lr{\mc{C}'}{\mc{D}'}{\mc{C}\mc{T}\mb{msg}(\mb{send}\,z^c_\gamma\,y^{c}_\alpha ; y^c_\alpha\leftarrow u^c_\delta)\mc{F}} \hookrightarrow_{\Delta_1 \Vdash u^c_{\delta}{:}B} {\m{queue}'}\\&  \m{and}\,  \lr{\mc{C}''}{\mc{T}}{\mc{C}'\mc{D}'\mb{msg}(\mb{send}\, z^c_\gamma\,y^{c}_\alpha; y^c_\alpha\leftarrow u^c_\delta)\mc{F}} \hookrightarrow_{\Delta_2 \Vdash z^c_\gamma{:}A } {\m{queue}''}& \\
(9)\, \lr{\mc{C}'\mb{msg}(\mb{send}\,z^{c'}_\gamma\,x^{c'}_\beta ;x^{c'}_\beta \leftarrow u^{c'}_\delta)}{\mc{D}}{\mc{F}}& \lr{\mc{C}'\mb{msg}(\mb{send}\,z^{c'}_\gamma\,x^{c'}_\beta  ;x^{c'}_\beta \leftarrow u^{c'}_\delta )}{\mc{D}}{\mc{F}}  \in \m{Tree}(\Delta' x^{c'}_\beta{:}A \otimes B \Vdash K) \\\qquad   \hookrightarrow_{\Delta, x_\beta^{c'}{:}A \otimes B\Vdash K} {\m{queue}'\overline{\mb{msg}(\mb{send}\,z^{c'}_\gamma\,x^{c'}_\beta ;x^{c'}_\beta \leftarrow u^{c'}_\delta)}}& \m{and}\, \lr{\mc{C}'}{\mb{msg}(\mb{send}\,z^{c'}_\gamma\,x^{c'}_\beta ;x^{c'}_\beta \leftarrow u^{c'}_\delta)\mc{D}}{\mc{F}} \hookrightarrow_{\Delta_1,z^{c'}_\gamma{:}A,u^{c'}_{\delta}{:}B \Vdash K}{\m{queue}'} \,
\\ 
(10) \,\lr{\mc{C}}{\mc{D}}{\mc{T}_1\mb{msg}(\mb{send}\,z^c_\gamma\,y^{c}_\alpha ;u^c_\delta \leftarrow y^c_\alpha)\mc{F}}& \lr{\mc{C}}{\mc{D}}{\mc{T}\mb{msg}(\mb{send}\, z^c_\gamma\,y^{c}_\alpha;u^c_\delta \leftarrow y^c_\alpha)\mc{F}}  \in \m{Tree}(\Delta \Vdash y^c_{\alpha}: A \multimap B)  \\  \qquad \hookrightarrow_{\Delta \Vdash y^c_{\alpha}{:}A \multimap B}\m{queue}'\overline{\mb{msg}(\mb{send}\, z^c_\gamma\,y^{c}_\alpha;u^c_\delta \leftarrow y^c_\alpha)}& \m{and}\, \lr{\mc{C}\mc{T}}{\mc{D}\mb{msg}(\mb{send}\, z^{c}_\gamma\,y^{c}_\alpha;u^c_\delta \leftarrow y^c_\alpha)}{\mc{F}} \hookrightarrow_{\Delta' , z^c_\gamma{:}A\Vdash u^c_{\delta}{:}B} {\m{queue}'}\\

(11)\,\lr{\mc{C}'\mc{C}''}{\mc{T}\mb{msg}(\mb{send}\,z^{c'}_\gamma\,x^{c'}_\beta ;u^{c'}_\delta \leftarrow x^{c'}_\beta)\mc{D}}{\mc{F}} &  \lr{\mc{C}'\mc{C}''}{\mc{T}\mb{msg}(\mb{send}\, z^{c'}_\gamma\,x^{c'}_\beta;u^{c'}_\delta \leftarrow x^{c'}_\beta)\mc{D}_1}{\mc{F}}\in \m{Tree}(\Delta_1,\Delta_2, x^{c'}_\beta{:}A \multimap B \Vdash K)\\ \qquad \hookrightarrow_{\Delta, x_\beta^{c'}{:}A \multimap B \Vdash K} {\m{queue}'\m{queue}''\mb{msg}(\mb{send}z^{c'}_\gamma\,x^{c'}_\beta ;u^{c'}_\delta \leftarrow x^{c'}_\beta)}&  \m{and}\,  \lr{\mc{C}'\mc{C}''\mc{T}\mb{msg}(\mb{send}z^{c'}_\gamma\,x^{c'}_\beta ;u^{c'}_\delta \leftarrow x^{c'}_\beta)}{\mc{D}}{\mc{F}} \hookrightarrow_{\Delta_1,u^{c'}_{\delta}{:}B \Vdash K}{\m{queue}'} \,  
\\& \m{and}\,\lr{\mc{C}''}{\mc{T}}{\mc{C}'\mb{msg}(\mb{send}\,z^{c'}_\gamma\,x^{c'}_\beta ;u^{c'}_\delta \leftarrow x^{c'}_\beta)\mc{D}'\mc{F}} \hookrightarrow_{\Delta_2 \Vdash  z^{c'}_\beta{:}A}{\m{queue}''}  \\
  (12)\, \lr{\mc{C}}{\mc{D}'\mb{proc}(y^c_{\alpha}, y^c_\alpha \leftarrow x^{c}_\beta @d_1)}{\mc{F}}&  \lr{\mc{C}}{\mc{D}'\mb{proc}(y^c_{\alpha}, y^c_\alpha \leftarrow x^{c}_\beta @d_1)}{\mc{F}}\in \m{Tree}(\Delta \Vdash y^c_{\alpha}{:}A )  \\ \qquad \hookrightarrow_{\Delta \Vdash y^c_{\alpha}{:}A } {\m{queue}'\{x/y\}}& \m{and}\, \lr{\mc{C}}{\mc{D}}{[x^c_\beta/y^c_{\alpha}]\mc{F}} \hookrightarrow_{\Delta \Vdash x^c_{\beta}{:}A } {\m{queue}'} \\
    (13)\, \lr{\mc{C}'\mb{proc}(x^{c'}_{\beta}, x^{c'}_\beta \leftarrow z^{c'}_\gamma @d_1)}{\mc{D}_1}{\mc{F}_1}&  \lr{\mc{C}'\mb{proc}(x^{c'}_{\beta}, x^{c'}_\beta \leftarrow z^{c'}_\gamma @d_1)}{\mc{D}}{\mc{F}}\in \m{Tree}(\Delta,  x^{c'}_\beta{:}A \Vdash K ) \\   \qquad \hookrightarrow_{\Delta, x^{c'}_{\beta}{:}A \Vdash K } {\m{queue}'\overline{\{x/y\}}}& \m{and}\, \lr{\mc{C}'}{[z^{c'}_\gamma/x^{c'}_{\beta}]\mc{D}}{[z^{c'}_\gamma/x^{c'}_{\beta}]\mc{F}} \hookrightarrow_{\Delta, z^{c'}_{\gamma}{:}A \Vdash  K} {\m{queue}'} \\
(14)\,   \lr{\cdot}{\cdot}{\cdot}\hookrightarrow_{\cdot \Vdash \cdot} \cdot&  \lr{\cdot}{\cdot}{\cdot}\in \m{Tree}(\cdot \Vdash \cdot)
\end{array}\]}

    \caption{Definition of $\hookrightarrow_{\Delta \Vdash K}$.  An overline indicates that a message is sent from $\mc{C}$ or $\mc{F}$ to $\mc{D}$, otherwise the message is sent from $\mc{D}$ to $\mc{C}$ or $\mc{F}$.}
   \label{fig:hookdef}
\end{figure*}

\subsection{Examples Typing}

\begin{flushleft}
\begin{small}
\begin{minipage}[t]{\textwidth}
\begin{tabbing}
$\m{auth}$ \phantom{mer} \= $=$ \= $\extchoice \{$\= $\mathit{tok_1}{:} \oplus \{ \mathit{succ}{:} \,\m{account} \otimes 1, \mathit{fail}{:}\, 1\}, $ \\
\> \> \> $\quad \vdots$ \\
\> \> \> $\mathit{tok_n}{:} \oplus \{ \mathit{succ}{:} \,\m{account} \otimes 1,  \mathit{fail}{:}\, 1\}\}$ \\[2pt]
$\m{customer}$ \> $=$ \> $\m{auth} \multimap 1$ \\[2pt]
$\m{account}$ \> $=$ \> $\oplus\{\mathit{high}{:} 1,\,\mathit{med}{:}1,\, \mathit{low}{:}1\}$ \\[2pt]
$\m{rate}$ \> $=$ \> $\&\{\mathit{lowRate}{:} 1,\,\mathit{highRate}{:}1\}$
\end{tabbing}
\end{minipage}
\end{small}
\end{flushleft}

\vspace{3mm}
\subsection*{Must-Type-Check Process Definitions}

\begin{flushleft}
\begin{small}
\begin{minipage}[t]{\textwidth}
\begin{tabbing}
$\cdot \vdash \m{Alice} :: y{:}\,\m{customer}\maxsec{[\mb{alice}]}$ \\
$y \leftarrow \m{Alice}\leftarrow \cdot = ($
\comment{// $\cdot \vdash y{:}\m{customer}[\mb{alice}]@\mb{alice}$} \\
\quad \= $w \leftarrow \mb{recv}\, y;$ 
\comment{// $w{:}\m{auth}[\mb{alice}] \vdash y{:}1 [\mb{alice}]@\mb{alice}$} \\
\> $w.\mathit{tok}_j;$
\comment{// $w{:}\oplus \{ \mathit{succ}{:} \,\m{account} \otimes 1, \mathit{fail}{:}\, 1\}[\mb{alice}] \vdash y{:}1 [\mb{alice}]@\mb{alice}$} \\
\> $\mb{case} \, w \, ($ \\
\> \quad \= $\mathit{succ} \Rightarrow$
\comment{// $w{:} \m{account} \otimes 1[\mb{alice}] \vdash y{:}1 [\mb{alice}]@\mb{alice}$} \\
\> \> \quad \= $v \leftarrow \mb{recv}\, w;$
\comment{// $v{:} \m{account}[\mb{alice}], w{:} 1[\mb{alice}] \vdash y{:}1 [\mb{alice}]@\mb{alice}$} \\
\> \> \> $\mb{case} \, v \, ($ \\
\> \> \> \quad \= $\mathit{high} \Rightarrow$
\comment{// $v{:} 1 [\mb{alice}], w{:} 1[\mb{alice}] \vdash y{:}1 [\mb{alice}]@\mb{alice}$} \\
\> \> \> \> \quad \= $\mb{wait}\,v;$
\comment{// $w{:} 1[\mb{alice}] \vdash y{:}1 [\mb{alice}]@\mb{alice}$} \\
\> \> \> \> \> $\mb{wait}\,w;$
\comment{// $\cdot \vdash y{:}1 [\mb{alice}]@\mb{alice}$} \\
\> \> \> \> \> $\mb{close}\,y$ \\
\> \> \> \> $\mid \mathit{med} \Rightarrow$ \\
\> \> \> \> \> $\mb{wait}\,v;$ \\
\> \> \> \> \> $\mb{wait}\,w;$ \\
\> \> \> \> \> $\mb{close}\,y$ \\
\> \> \> \> $\mid \mathit{low} \Rightarrow$ \\
\> \> \> \> \> $\mb{wait}\,v;$ \\
\> \> \> \> \> $\mb{wait}\,w;$ \\
\> \> \> \> \> $\mb{close}\,y)$ \\
\> \> $\mid \mathit{fail}\Rightarrow$
\comment{// $w{:} 1[\mb{alice}] \vdash y{:}1 [\mb{alice}]@\mb{alice}$} \\
\> \> \> $\mb{wait}\,w;$
\comment{// $\cdot \vdash y{:}1 [\mb{alice}]@\mb{alice}$} \\
\> \> \> $\mb{close}\,y ))\runsec{@\mb{alice}}$
\end{tabbing}
\end{minipage}
\end{small}
\end{flushleft}

\vspace{3mm}

\begin{flushleft}
\begin{small}
\begin{minipage}[t]{\textwidth}
\begin{tabbing}
$u{:}\,\&\{\mathit{s}{:}\m{account}, \mathit{f}{:}1\}\maxsec{[\mb{alice}]} \vdash \m{aAuth} :: x{:}\m{auth}\maxsec{[\mb{alice}]}$ \\
$x \leftarrow \m{aAuth}\leftarrow u = ($
\comment{// $u{:}\,\&\{\mathit{s}{:}\m{account}, \mathit{f}{:}1\} [\mb{alice}] \vdash x{:}\m{auth}[\mb{alice}]@\mb{alice}$} \\
\quad \= $\mb{case} \, x \, ($ \\
\> \quad \= $\mathit{tok}_j \Rightarrow$
\comment{// $u{:}\,\&\{\mathit{s}{:}\m{account}, \mathit{f}{:}1\} [\mb{alice}] \vdash x{:}\oplus \{\mathit{succ}{:}\,\m{account} \otimes 1, \mathit{fail}{:}\, 1\}[\mb{alice}]@\mb{alice}$} \\
\> \> \quad \= $x.\mathit{succ};$
\comment{// $u{:}\,\&\{\mathit{s}{:}\m{account}, \mathit{f}{:}1\} [\mb{alice}] \vdash x{:}\m{account} \otimes 1[\mb{alice}]@\mb{alice}$} \\
\> \> \> $u.\mathit{s};$
\comment{// $u{:}\,\m{account}[\mb{alice}] \vdash x{:}\m{account} \otimes 1[\mb{alice}]@\mb{alice}$} \\
\> \> \> $\mb{send}\, u\, x;$
\comment{// $\cdot \vdash x{:}1[\mb{alice}]@\mb{alice}$} \\
\> \> \> $\mb{close}\,x$ \\
\> \> $\mid \mathit{tok}_{i \neq j} \Rightarrow$
\comment{// $u{:}\,\&\{\mathit{s}{:}\m{account}, \mathit{f}{:}1\} [\mb{alice}] \vdash x{:}\oplus \{\mathit{succ}{:}\,\m{account} \otimes 1, \mathit{fail}{:}\, 1\}[\mb{alice}]@\mb{alice}$} \\
\> \> \quad $x.\mathit{fail};$
\comment{// $u{:}\,\&\{\mathit{s}{:}\m{account}, \mathit{f}{:}1\} [\mb{alice}] \vdash x{:}1[\mb{alice}]@\mb{alice}$} \\
\> \> \> $u.\mathit{f};$
\comment{// $u{:}\,1 [\mb{alice}] \vdash x{:}1[\mb{alice}]@\mb{alice}$} \\
\> \> \> $\mb{wait}\, u;$
\comment{// $\cdot \vdash x{:}1[\mb{alice}]@\mb{alice}$} \\
\> \> \> $\mb{close}\, x ))\runsec{@\mb{alice}}$
\end{tabbing}
\end{minipage}
\end{small}
\end{flushleft}

\vspace{3mm}

\begin{flushleft}
\begin{small}
\begin{minipage}[t]{\textwidth}
\begin{tabbing}
$\cdot \vdash \m{aAcc} :: u{:}\&\{\mathit{s}{:}\m{account}, \mathit{f}{:}1\}\maxsec{[\mb{alice}]}$ \\
$u \leftarrow \m{aAcc}\leftarrow \cdot = ($
\comment{// $\cdot \vdash u{:}\, \&\{\mathit{s}{:}\m{account}, \mathit{f}{:}1\}[\mb{alice}]@\mb{alice}$} \\
\quad \= $\mb{case} \, u \, ($\\
\> \quad \= $\mathit{s} \Rightarrow$
\comment{// $\cdot \vdash u{:}\, \m{account}[\mb{alice}]@\mb{alice}$} \\
\> \> \quad \=$u.\mathit{high};$
\comment{// $\cdot \vdash u{:}\, 1 [\mb{alice}]@\mb{alice}$} \\
\> \> \> $\mb{close}\, u$ \\
\> \> $\mid \mathit{f} \Rightarrow$
\comment{// $\cdot \vdash u{:}\, 1 [\mb{alice}]@\mb{alice}$} \\
\> \> \quad \= $\mb{close}\,u ))\runsec{@\mb{alice}}$
\end{tabbing}
\end{minipage}
\end{small}
\end{flushleft}

\begin{flushleft}
\begin{small}
\begin{minipage}[t]{\textwidth}
\begin{tabbing}
$x{:}\,\m{auth}\maxsec{[\mb{alice}]} , y{:}\,\m{customer}\maxsec{[\mb{alice}]}, x'{:}\m{auth}\maxsec{[\mb{bob}]},y'{:}\,\m{customer}\maxsec{[\mb{bob}]}, u{:}\,\m{rate}\maxsec{[\mb{guest}]} \vdash \m{Bank} :: z{:}1\maxsec{[\mb{bank}]}$ \\
$z \leftarrow \m{Bank}\leftarrow x,x',y,y',u = ($
\comment{// $x{:}\,\m{auth}[\mb{alice}] , y{:}\,\m{customer}[\mb{alice}], x'{:}\m{auth}[\mb{bob}],y'{:}\,\m{customer}[\mb{bob}], u{:}\,\m{rate}[\mb{guest}] \vdash z{:}1[\mb{bank}]@\mb{guest}$} \\
\quad \= $\mb{send}\, x\,y;$
\comment{// $y{:}\,1[\mb{alice}], x'{:}\m{auth}[\mb{bob}],y'{:}\,\m{customer}[\mb{bob}], u{:}\,\m{rate}[\mb{guest}] \vdash z{:}1[\mb{bank}]@\mb{guest}$
\quad \textit{note that} $\mb{guest}\sqsubseteq\mb{alice}$} \\
\> $\mb{send}\, x'\,y';$
\comment{// $y{:}\,1[\mb{alice}],y'{:}\,1[\mb{bob}], u{:}\,\m{rate}[\mb{guest}] \vdash z{:}1[\mb{bank}]@\mb{guest}$
\quad \textit{note that} $\mb{guest}\sqsubseteq\mb{bob}$} \\
\> $u.\mathit{lowRate};$
\comment{// $y{:}\,1[\mb{alice}],y'{:}\,1[\mb{bob}], u{:}\,1[\mb{guest}] \vdash z{:}1[\mb{bank}]@\mb{guest}$
\quad \textit{note that} $\mb{guest}\sqsubseteq\mb{guest}$} \\
\>$\mb{wait}\,y;$
\comment{// $y'{:}\,1[\mb{bob}], u{:}\,1[\mb{guest}] \vdash z{:}1[\mb{bank}]@\mb{bank}$
\quad \textit{note that} $\mb{guest}\sqcup\mb{alice} = \mb{bank}$} \\
\> $\mb{wait}\,y';$
\comment{// $u{:}\,1[\mb{guest}] \vdash z{:}1[\mb{bank}]@\mb{bank}$
\quad \textit{note that} $\mb{bob}\sqcup\mb{bank} = \mb{bank}$} \\
\> $\mb{wait}\, u;$
\comment{// $\cdot \vdash z{:}1[\mb{bank}]@\mb{bank}$
\quad \textit{note that} $\mb{guest}\sqcup\mb{bank} = \mb{bank}$} \\
\> $\mb{close}\,z)\runsec{@\mb{guest}}$
\end{tabbing}
\end{minipage}
\end{small}
\end{flushleft}

\vspace{3mm}
\subsection*{Must-NOT-Type-Check Process Definitions}

\begin{flushleft}
\begin{small}
\begin{minipage}[t]{\textwidth}
\begin{tabbing}
$x{:}\,\m{auth}\maxsec{[\mb{alice}]} , y{:}\,\m{customer}\maxsec{[\mb{guest}]} \vdash \m{LeakyBank} :: z{:}1\maxsec{[\mb{bank}]}$ \\
$z \leftarrow \m{LeakyBank}\leftarrow x,y = ($
\comment{// $x{:}\,\m{auth}[\mb{alice}], y{:}\,\m{customer}[\mb{guest}] \vdash z{:}1[\mb{bank}]@\mb{guest}$} \\
\quad \= $\mb{send}\, x\,y;$
\comment{// FAILS HERE because $\mb{alice} \neq \mb{guest}$} \\
\> $\mb{wait}\,y; \mb{close}\,z)\runsec{@{\mb{guest}}}$
\end{tabbing}
\end{minipage}
\end{small}
\end{flushleft}

\vspace{3mm}

\begin{flushleft}
\begin{small}
\begin{minipage}[t]{\textwidth}
\begin{tabbing}
$x_1{:} \&\{s{:}1, f{:} 1\}\maxsec{[\mb{guest}]},u{:}\,\&\{\mathit{s}{:}\m{account}, \mathit{f}{:}1\}\maxsec{[\mb{alice}]} \vdash \m{SneakyaAuth} :: x{:}\m{auth}\maxsec{[\mb{alice}]}$ \\
$x \leftarrow \m{SneakyaAuth}\leftarrow u, x_1 = ($
\comment{// $x_1{:} \&\{s{:}1, f{:} 1\}[\mb{guest}],u{:}\,\&\{\mathit{s}{:}\m{account}, \mathit{f}{:}1\}[\mb{alice}] \vdash x{:}\m{auth}[\mb{alice}]@\mb{alice}$} \\
\quad \= $\mb{case} \, x \, ($ \\
\> \quad \= $\mathit{tok}_j \Rightarrow$
\comment{// $x_1{:} \&\{s{:}1, f{:} 1\}[\mb{guest}],u{:}\,\&\{\mathit{s}{:}\m{account}, \mathit{f}{:}1\}[\mb{alice}] \vdash x{:}\oplus \{ \mathit{succ}{:} \,\m{account} \otimes 1, \mathit{fail}{:}\, 1\}[\mb{alice}]@\mb{alice}$} \\
\> \> \quad \=$x.\mathit{succ};$
\comment{// $x_1{:} \&\{s{:}1, f{:} 1\}[\mb{guest}],u{:}\,\&\{\mathit{s}{:}\m{account}, \mathit{f}{:}1\}[\mb{alice}] \vdash x{:}\m{account} \otimes 1[\mb{alice}]@\mb{alice}$} \\
\> \> \> $u.\mathit{s};$
\comment{// $x_1{:} \&\{s{:}1, f{:} 1\}[\mb{guest}],u{:}\m{account}[\mb{alice}] \vdash x{:}\m{account} \otimes 1[\mb{alice}]@\mb{alice}$} \\
\> \> \> $x_1.s;$
\comment{// FAILS HERE because $\mb{alice} \not\sqsubseteq \mb{guest}$} \\
\> \> \>$\mb{send}\, u\, x ; \mb{wait}\, x_1; \mb{close}\,x$ \\
\> \> $\mid \mathit{tok}_{i \neq j} \Rightarrow$
\comment{// $x_1{:} \&\{s{:}1, f{:} 1\}[\mb{guest}],u{:}\,\&\{\mathit{s}{:}\m{account}, \mathit{f}{:}1\}[\mb{alice}] \vdash x{:}\oplus \{ \mathit{succ}{:} \,\m{account} \otimes 1, \mathit{fail}{:}\, 1\}[\mb{alice}]@\mb{alice}$} \\
\> \> \quad \=$x.\mathit{fail};$
\comment{// $x_1{:} \&\{s{:}1, f{:} 1\}[\mb{guest}],u{:}\,\&\{\mathit{s}{:}\m{account}, \mathit{f}{:}1\}[\mb{alice}] \vdash x{:}1[\mb{alice}]@\mb{alice}$} \\
\> \> \> $u.\mathit{f};$
\comment{// $x_1{:} \&\{s{:}1, f{:} 1\}[\mb{guest}],u{:}\,1[\mb{alice}] \vdash x{:}1[\mb{alice}]@\mb{alice}$} \\
\> \> \> $x_1.f;$
\comment{// FAILS HERE because $\mb{alice} \not\sqsubseteq \mb{guest}$} \\
\> \> \> $\mb{wait}\, u; \mb{wait}\, x_1; \mb{close}\, x ))\runsec{@\mb{alice}}$
\end{tabbing}
\end{minipage}
\end{small}
\end{flushleft}

\vspace{3mm}

\begin{flushleft}
\begin{small}
\begin{minipage}[t]{\textwidth}
\begin{tabbing}
$x_1{:} \&\{s{:}1, f{:} 1\}\maxsec{[\mb{guest}]}, u{:}\,\&\{\mathit{s}{:}\m{account}, \mathit{f}{:}1\}\maxsec{[\mb{alice}]} \vdash \m{SneakyaAuth} :: x{:}\m{auth}\maxsec{[\mb{alice}]}$ \\
$x \leftarrow \m{SneakyaAuth}\leftarrow u, x_1 = ($
\comment{// $x_1{:} \&\{s{:}1, f{:} 1\}[\mb{guest}], u{:}\,\&\{\mathit{s}{:}\m{account}, \mathit{f}{:}1\}[\mb{alice}] \vdash x{:}\m{auth}[\mb{alice}]@\mb{alice}$} \\
\quad \= $\mb{case} \, x \, ($ \\
\> \quad \= $\mathit{tok}_j \Rightarrow$
\comment{// $x_1{:} \&\{s{:}1, f{:} 1\}[\mb{guest}], u{:}\,\&\{\mathit{s}{:}\m{account}, \mathit{f}{:}1\}[\mb{alice}] \vdash x{:}\oplus \{ \mathit{succ}{:} \,\m{account} \otimes 1, \mathit{fail}{:}\, 1\}[\mb{alice}]@\mb{alice}$} \\
\> \> \quad \=$x.\mathit{succ};$
\comment{// $x_1{:} \&\{s{:}1, f{:} 1\}[\mb{guest}], u{:}\,\&\{\mathit{s}{:}\m{account}, \mathit{f}{:}1\}[\mb{alice}] \vdash x{:}\m{account} \otimes 1[\mb{alice}]@\mb{alice}$} \\
\> \> \> $u.\mathit{s};$
\comment{// $x_1{:} \&\{s{:}1, f{:} 1\}[\mb{guest}], u{:}\,\m{account}[\mb{alice}] \vdash x{:}\m{account} \otimes 1[\mb{alice}]@\mb{alice}$} \\
\> \> \> $z_1 \leftarrow S \leftarrow x_1;$
\comment{// FAILS HERE because $\mb{alice} \not\sqsubseteq \mb{guest} \sqsubseteq \mb{alice}$} \\
\> \> \>$\mb{send}\, u\, x ; \mb{wait}\, z_1; \mb{close}\,x$ \\
\> \> $\mid \mathit{tok}_{i \neq j} \Rightarrow$
\comment{// $x_1{:} \&\{s{:}1, f{:} 1\}[\mb{guest}], u{:}\,\&\{\mathit{s}{:}\m{account}, \mathit{f}{:}1\}[\mb{alice}] \vdash x{:}\oplus \{ \mathit{succ}{:} \,\m{account} \otimes 1, \mathit{fail}{:}\, 1\}[\mb{alice}]@\mb{alice}$} \\
\> \> \quad \=$x.\mathit{fail};$
\comment{// $x_1{:} \&\{s{:}1, f{:} 1\}[\mb{guest}], u{:}\,\&\{\mathit{s}{:}\m{account}, \mathit{f}{:}1\}[\mb{alice}] \vdash x{:}1[\mb{alice}]@\mb{alice}$} \\
\> \> \> $u.\mathit{f};$
\comment{// $x_1{:} \&\{s{:}1, f{:} 1\}[\mb{guest}], u{:}\,1[\mb{alice}] \vdash x{:}1[\mb{alice}]@\mb{alice}$} \\
\> \> \> $z_1 \leftarrow F \leftarrow x_1;$
\comment{// FAILS HERE because $\mb{alice} \not\sqsubseteq \mb{guest} \sqsubseteq \mb{alice}$} \\
\> \> \> $\mb{wait}\, u; \mb{wait}\, z_1; \mb{close}\, x ))\runsec{@\mb{alice}}$ \\[3pt]
\end{tabbing}
\end{minipage}
\end{small}
\end{flushleft}

\begin{flushleft}
\begin{small}
\begin{minipage}[t]{\textwidth}
\begin{tabbing}
$x_1{:}\&\{s{:}1, f{:} 1\}\maxsec{[\mb{guest}]} \vdash \m{S} :: z_1{:}1\maxsec{[\mb{alice}]}$ \\
$z_1 \leftarrow \m{S}\leftarrow x_1 = ($
\comment{// $x_1{:}\&\{s{:}1, f{:} 1\}[\mb{guest}] \vdash z_1{:}1[\mb{alice}]@\mb{guest}$} \\
\quad \= $x_1.s;$
\comment{// $x_1{:}1[\mb{guest}] \vdash z_1{:}1[\mb{alice}]@\mb{guest}$
\quad \textit{note that} $\mb{guest}\sqsubseteq\mb{guest}$} \\
\> $\mb{wait}\, x_1;$
\comment{// $\cdot \vdash z_1{:}1[\mb{alice}]@\mb{guest}$
\quad \textit{note that} $\mb{guest}\sqcup\mb{guest} = \mb{guest}$} \\
\> $\mb{close}\, z_1)\runsec{@\mb{guest}}$ \\[3pt]
\end{tabbing}
\end{minipage}
\end{small}
\end{flushleft}

\begin{flushleft}
\begin{small}
\begin{minipage}[t]{\textwidth}
\begin{tabbing}
$x_1{:}\&\{s{:}1, f{:} 1\}\maxsec{[\mb{guest}]} \vdash \m{F} :: z_1{:}1\maxsec{[\mb{alice}]}$ \\
$z_1 \leftarrow \m{F}\leftarrow x_1 = ($
\comment{// $x_1{:}\&\{s{:}1, f{:} 1\}[\mb{guest}] \vdash z_1{:}1[\mb{alice}]@\mb{guest}$} \\
\quad \= $x_1.f;$
\comment{// $x_1{:}1[\mb{guest}] \vdash z_1{:}1[\mb{alice}]@\mb{guest}$
\quad \textit{note that} $\mb{guest}\sqsubseteq\mb{guest}$} \\
\> $\mb{wait}\, x_1;$
\comment{// $\cdot \vdash z_1{:}1[\mb{alice}]@\mb{guest}$
\quad \textit{note that} $\mb{guest}\sqcup\mb{guest} = \mb{guest}$} \\
\> $\mb{close}\, z_1)\runsec{@\mb{guest}}$
\end{tabbing}
\end{minipage}
\end{small}
\end{flushleft}




%% file: main_arxiv.bbl
\begin{thebibliography}{10}
\providecommand{\url}[1]{#1}
\csname url@samestyle\endcsname
\providecommand{\newblock}{\relax}
\providecommand{\bibinfo}[2]{#2}
\providecommand{\BIBentrySTDinterwordspacing}{\spaceskip=0pt\relax}
\providecommand{\BIBentryALTinterwordstretchfactor}{4}
\providecommand{\BIBentryALTinterwordspacing}{\spaceskip=\fontdimen2\font plus
\BIBentryALTinterwordstretchfactor\fontdimen3\font minus
  \fontdimen4\font\relax}
\providecommand{\BIBforeignlanguage}[2]{{%
\expandafter\ifx\csname l@#1\endcsname\relax
\typeout{** WARNING: IEEEtran.bst: No hyphenation pattern has been}%
\typeout{** loaded for the language `#1'. Using the pattern for}%
\typeout{** the default language instead.}%
\else
\language=\csname l@#1\endcsname
\fi
#2}}
\providecommand{\BIBdecl}{\relax}
\BIBdecl

\bibitem{HondaCONCUR1993}
\BIBentryALTinterwordspacing
K.~Honda, ``Types for dyadic interaction,'' in \emph{4th International
  Conference on Concurrency Theory ({CONCUR})}, ser. Lecture Notes in Computer
  Science, vol. 715.\hskip 1em plus 0.5em minus 0.4em\relax Springer, 1993, pp.
  509--523. [Online]. Available:
  \url{https://doi.org/10.1007/3-540-57208-2\_35}
\BIBentrySTDinterwordspacing

\bibitem{HondaESOP1998}
\BIBentryALTinterwordspacing
K.~Honda, V.~T. Vasconcelos, and M.~Kubo, ``Language primitives and type
  discipline for structured communication-based programming,'' in \emph{7th
  European Symposium on Programming ({ESOP})}, ser. Lecture Notes in Computer
  Science, vol. 1381.\hskip 1em plus 0.5em minus 0.4em\relax Springer, 1998,
  pp. 122--138. [Online]. Available: \url{https://doi.org/10.1007/BFb0053567}
\BIBentrySTDinterwordspacing

\bibitem{ToninhoESOP2013}
\BIBentryALTinterwordspacing
B.~Toninho, L.~Caires, and F.~Pfenning, ``Higher-order processes, functions,
  and sessions: A monadic integration,'' in \emph{22nd European Symposium on
  Programming ({ESOP})}, ser. Lecture Notes in Computer Science, vol.
  7792.\hskip 1em plus 0.5em minus 0.4em\relax Springer, 2013, pp. 350--369.
  [Online]. Available: \url{https://doi.org/10.1007/978-3-642-37036-6\_20}
\BIBentrySTDinterwordspacing

\bibitem{ToninhoPhD2015}
B.~Toninho, ``A logical foundation for session-based concurrent computation,''
  Ph.D. dissertation, Carnegie Mellon University and New University of Lisbon,
  2015.

\bibitem{BalzerICFP2017}
\BIBentryALTinterwordspacing
S.~Balzer and F.~Pfenning, ``Manifest sharing with session types,''
  \emph{Proceedings of the ACM on Programming Languages}, vol.~1, no. {ICFP},
  pp. 37:1--37:29, 2017. [Online]. Available:
  \url{https://doi.org/10.1145/3110281}
\BIBentrySTDinterwordspacing

\bibitem{DezaniECOOP2006}
\BIBentryALTinterwordspacing
M.~Dezani{-}Ciancaglini, D.~Mostrous, N.~Yoshida, and S.~Drossopoulou,
  ``Session types for object-oriented languages,'' in \emph{20th European
  Conference on Object-Oriented Programming}, ser. Lecture Notes in Computer
  Science, vol. 4067.\hskip 1em plus 0.5em minus 0.4em\relax Springer, 2006,
  pp. 328--352. [Online]. Available: \url{https://doi.org/10.1007/11785477\_20}
\BIBentrySTDinterwordspacing

\bibitem{PucellaHASKELL2008}
\BIBentryALTinterwordspacing
R.~Pucella and J.~A. Tov, ``Haskell session types with (almost) no class,'' in
  \emph{1st {ACM} {SIGPLAN} Symposium on Haskell (Haskell)}.\hskip 1em plus
  0.5em minus 0.4em\relax {ACM}, 2008, pp. 25--36. [Online]. Available:
  \url{https://doi.org/10.1145/1411286.1411290}
\BIBentrySTDinterwordspacing

\bibitem{ImaiPLACES2010}
\BIBentryALTinterwordspacing
K.~Imai, S.~Yuen, and K.~Agusa, ``Session type inference in haskell,'' in
  \emph{3rd Workshop on Programming Language Approaches to Concurrency and
  Communication-cEntric Software ({PLACES})}, ser. {EPTCS}, vol.~69, 2010, pp.
  74--91. [Online]. Available: \url{https://doi.org/10.4204/EPTCS.69.6}
\BIBentrySTDinterwordspacing

\bibitem{JespersenWGP2015}
\BIBentryALTinterwordspacing
T.~B.~L. Jespersen, P.~Munksgaard, and K.~F. Larsen, ``Session types for
  {Rust},'' in \emph{11th {ACM} {SIGPLAN} Workshop on Generic Programming
  ({WGP})}.\hskip 1em plus 0.5em minus 0.4em\relax {ACM}, 2015, pp. 13--22.
  [Online]. Available: \url{https://doi.org/10.1145/2808098.2808100}
\BIBentrySTDinterwordspacing

\bibitem{LindleyHASKELL2016}
\BIBentryALTinterwordspacing
S.~Lindley and J.~G. Morris, ``Embedding session types in {Haskell},'' in
  \emph{9th International Symposium on Haskell ({Haskell})}.\hskip 1em plus
  0.5em minus 0.4em\relax {ACM}, 2016, pp. 133--145. [Online]. Available:
  \url{https://doi.org/10.1145/2976002.2976018}
\BIBentrySTDinterwordspacing

\bibitem{ScalasECOOP2016}
\BIBentryALTinterwordspacing
A.~Scalas and N.~Yoshida, ``Lightweight session programming in {Scala},'' in
  \emph{30th European Conference on Object-Oriented Programming ({ECOOP})},
  ser. LIPIcs, no.~56.\hskip 1em plus 0.5em minus 0.4em\relax Schloss Dagstuhl
  - Leibniz-Zentrum f{\"{u}}r Informatik, 2016, pp. 21:1--21:28. [Online].
  Available: \url{https://doi.org/10.4230/LIPIcs.ECOOP.2016.21}
\BIBentrySTDinterwordspacing

\bibitem{PadovaniARTICLE2017}
\BIBentryALTinterwordspacing
L.~Padovani, ``A simple library implementation of binary sessions,''
  \emph{Journal of Functional Programming}, vol.~27, p.~e4, 2017. [Online].
  Available: \url{https://doi.org/10.1017/S0956796816000289}
\BIBentrySTDinterwordspacing

\bibitem{ImaiARTICLE2019}
\BIBentryALTinterwordspacing
K.~Imai, N.~Yoshida, and S.~Yuen, ``Session-ocaml: A session-based library with
  polarities and lenses,'' \emph{Science of Computer Programming}, vol. 172,
  pp. 135--159, 2019. [Online]. Available:
  \url{https://doi.org/10.1016/j.scico.2018.08.005}
\BIBentrySTDinterwordspacing

\bibitem{KokkeICE2019}
\BIBentryALTinterwordspacing
W.~Kokke, ``Rusty variation: Deadlock-free sessions with failure in {Rust},''
  in \emph{12th Interaction and Concurrency Experience ({ICE})}, ser. {EPTCS},
  vol. 304, 2019, pp. 48--60. [Online]. Available:
  \url{https://doi.org/10.4204/EPTCS.304.4}
\BIBentrySTDinterwordspacing

\bibitem{ChenARXIV2020}
\BIBentryALTinterwordspacing
R.~Chen and S.~Balzer, ``{Ferrite}: A judgmental embedding of session types in
  {Rust},'' \emph{CoRR}, vol. abs/2009.13619, 2020. [Online]. Available:
  \url{https://arxiv.org/abs/2009.13619}
\BIBentrySTDinterwordspacing

\bibitem{CairesCONCUR2010}
\BIBentryALTinterwordspacing
L.~Caires and F.~Pfenning, ``Session types as intuitionistic linear
  propositions,'' in \emph{21th International Conference onf Concurrency Theory
  ({CONCUR})}, ser. Lecture Notes in Computer Science, vol. 6269.\hskip 1em
  plus 0.5em minus 0.4em\relax Springer, 2010, pp. 222--236. [Online].
  Available: \url{https://doi.org/10.1007/978-3-642-15375-4\_16}
\BIBentrySTDinterwordspacing

\bibitem{WadlerICFP2012}
\BIBentryALTinterwordspacing
P.~Wadler, ``Propositions as sessions,'' in \emph{{ACM} {SIGPLAN} International
  Conference on Functional Programming ({ICFP})}.\hskip 1em plus 0.5em minus
  0.4em\relax {ACM}, 2012, pp. 273--286. [Online]. Available:
  \url{https://doi.org/10.1145/2364527.2364568}
\BIBentrySTDinterwordspacing

\bibitem{KokkePOPL2019}
\BIBentryALTinterwordspacing
W.~Kokke, F.~Montesi, and M.~Peressotti, ``Better late than never: A
  fully-abstract semantics for classical processes,'' \emph{Proceedings of the
  {ACM} on Programming Languages}, vol.~3, no. {POPL}, pp. 24:1--24:29, 2019.
  [Online]. Available: \url{https://doi.org/10.1145/3290337}
\BIBentrySTDinterwordspacing

\bibitem{volpano96}
D.~Volpano, C.~Irvine, and G.~Smith, ``A sound type system for secure flow
  analysis,'' \emph{J. Comput. Secur.}, vol.~4, no. 2–3, p. 167–187, Jan.
  1996.

\bibitem{ifc-survey}
\BIBentryALTinterwordspacing
A.~Sabelfeld and A.~C. Myers, ``Language-based information-flow security,''
  \emph{IEEE J.Sel. A. Commun.}, vol.~21, no.~1, p. 5–19, Sep. 2006.
  [Online]. Available: \url{https://doi.org/10.1109/JSAC.2002.806121}
\BIBentrySTDinterwordspacing

\bibitem{HondaESOP2000}
\BIBentryALTinterwordspacing
K.~Honda, V.~T. Vasconcelos, and N.~Yoshida, ``Secure information flow as typed
  process behaviour,'' in \emph{9th European Symposium on Programming
  ({ESOP})}, ser. Lecture Notes in Computer Science, vol. 1782.\hskip 1em plus
  0.5em minus 0.4em\relax Springer, 2000, pp. 180--199. [Online]. Available:
  \url{https://doi.org/10.1007/3-540-46425-5\_12}
\BIBentrySTDinterwordspacing

\bibitem{HondaYoshidaPOPL2002}
\BIBentryALTinterwordspacing
K.~Honda and N.~Yoshida, ``A uniform type structure for secure information
  flow,'' in \emph{29th {SIGPLAN-SIGACT} Symposium on Principles of Programming
  Languages ({POPL})}.\hskip 1em plus 0.5em minus 0.4em\relax {ACM}, 2002, pp.
  81--92. [Online]. Available: \url{https://doi.org/10.1145/503272.503281}
\BIBentrySTDinterwordspacing

\bibitem{CrafaARTICLE2002}
\BIBentryALTinterwordspacing
S.~Crafa, M.~Bugliesi, and G.~Castagna, ``Information flow security for boxed
  ambients,'' \emph{Electronic Notes in Theoretical Computer Science}, vol.~66,
  no.~3, pp. 76--97, 2002. [Online]. Available:
  \url{https://doi.org/10.1016/S1571-0661(04)80417-1}
\BIBentrySTDinterwordspacing

\bibitem{CrafaTGC2005}
\BIBentryALTinterwordspacing
S.~Crafa and S.~Rossi, ``A theory of noninterference for the
  {$\pi$}-calculus,'' in \emph{International Symposium on Trustworthy Global
  Computing ({TGC})}, ser. Lecture Notes in Computer Science, vol. 3705.\hskip
  1em plus 0.5em minus 0.4em\relax Springer, 2005, pp. 2--18. [Online].
  Available: \url{https://doi.org/10.1007/11580850\_2}
\BIBentrySTDinterwordspacing

\bibitem{CrafaFMSE2006}
\BIBentryALTinterwordspacing
------, ``P-congruences as non-interference for the pi-calculus,'' in
  \emph{{ACM} Workshop on Formal Methods in Security Engineering
  ({FMSE})}.\hskip 1em plus 0.5em minus 0.4em\relax {ACM}, 2006, pp. 13--22.
  [Online]. Available: \url{https://doi.org/10.1145/1180337.1180339}
\BIBentrySTDinterwordspacing

\bibitem{Crafa2007}
\BIBentryALTinterwordspacing
------, ``Controlling information release in the {$\pi$}-calculus,''
  \emph{Information and Computation}, vol. 205, no.~8, pp. 1235 -- 1273, 2007.
  [Online]. Available:
  \url{http://www.sciencedirect.com/science/article/pii/S089054010700003X}
\BIBentrySTDinterwordspacing

\bibitem{HENNESSYRIELY2002}
\BIBentryALTinterwordspacing
M.~Hennessy and J.~Riely, ``Information flow vs. resource access in the
  asynchronous pi-calculus,'' \emph{ACM Trans. Program. Lang. Syst.}, vol.~24,
  no.~5, p. 566–591, Sep. 2002. [Online]. Available:
  \url{https://doi.org/10.1145/570886.570890}
\BIBentrySTDinterwordspacing

\bibitem{HENNESSY20053}
\BIBentryALTinterwordspacing
M.~Hennessy, ``The security pi-calculus and non-interference,'' \emph{The
  Journal of Logic and Algebraic Programming}, vol.~63, no.~1, pp. 3 -- 34,
  2005, special issue on The pi-calculus. [Online]. Available:
  \url{http://www.sciencedirect.com/science/article/pii/S1567832604000049}
\BIBentrySTDinterwordspacing

\bibitem{KOBAYASHI2005}
N.~Kobayashi, ``Type-based information flow analysis for the pi-calculus,''
  \emph{Acta Inf.}, vol.~42, no.~4, p. 291–347, Dec. 2005.

\bibitem{ZDANCEWIC2003}
S.~{Zdancewic} and A.~C. {Myers}, ``Observational determinism for concurrent
  program security,'' in \emph{16th IEEE Computer Security Foundations Workshop
  (CSFW)}, 2003, pp. 29--43.

\bibitem{POTTIER2002}
F.~{Pottier}, ``A simple view of type-secure information flow in the
  {$\pi$}-calculus,'' in \emph{Proceedings 15th IEEE Computer Security
  Foundations Workshop (CSFW-15)}, 2002, pp. 320--330.

\bibitem{cowl}
D.~Stefan, E.~Z. Yang, B.~Karp, P.~Marchenko, A.~Russo, and D.~Mazi\`{e}res,
  ``Protecting users by confining {J}ava{S}cript with {COWL},'' in \emph{Proc.\
  OSDI}, 2014.

\bibitem{ifc-browser-ndss15}
L.~Bauer, S.~Cai, L.~Jia, T.~Passaro, M.~Stroucken, and Y.~Tian, ``Run-time
  monitoring and formal analysis of information flows in chromium,'' in
  \emph{Proceedings of the 22nd Annual Network \& Distributed System Security
  Symposium (NDSS)}, 2015.

\bibitem{android-esorics13}
L.~Jia, J.~Aljuraidan, E.~Fragkaki, L.~Bauer, M.~Stroucken, K.~Fukushima,
  S.~Kiyomoto, and Y.~Miyake, ``Run-time enforcement of information-flow
  properties on android (extended abstract),'' in \emph{Computer Security –
  ESORICS 2013: 18th European Symposium on Research in Computer Security
  (ESORICS)}, 2013.

\bibitem{krohn:flume}
M.~Krohn, A.~Yip, M.~Brodsky, N.~Cliffer, M.~F. Kaashoek, E.~Kohler, and
  R.~Morris, ``Information flow control for standard {OS} abstractions,'' in
  \emph{Proc.\ SOSP}, 2007.

\bibitem{CapecchiCONCUR2010}
\BIBentryALTinterwordspacing
S.~Capecchi, I.~Castellani, M.~Dezani{-}Ciancaglini, and T.~Rezk, ``Session
  types for access and information flow control,'' in \emph{21th International
  Conference on Concurrency Theory ({CONCUR})}, 2010, pp. 237--252. [Online].
  Available: \url{https://doi.org/10.1007/978-3-642-15375-4\_17}
\BIBentrySTDinterwordspacing

\bibitem{CapecchiARTICLE2014}
\BIBentryALTinterwordspacing
S.~Capecchi, I.~Castellani, and M.~Dezani{-}Ciancaglini, ``Typing access
  control and secure information flow in sessions,'' \emph{Information and
  Computation}, vol. 238, pp. 68--105, 2014. [Online]. Available:
  \url{https://doi.org/10.1016/j.ic.2014.07.005}
\BIBentrySTDinterwordspacing

\bibitem{CairesARTICLE2016}
\BIBentryALTinterwordspacing
L.~Caires, F.~Pfenning, and B.~Toninho, ``Linear logic propositions as session
  types,'' \emph{Mathematical Structures in Computer Science}, vol.~26, no.~3,
  pp. 367--423, 2016. [Online]. Available:
  \url{https://doi.org/10.1017/S0960129514000218}
\BIBentrySTDinterwordspacing

\bibitem{TaitARTICLE1967}
\BIBentryALTinterwordspacing
W.~W. Tait, ``Intensional interpretations of functionals of finite type {I},''
  \emph{The Journal of Symbolic Logic}, vol.~32, no.~2, pp. 198--212, 1967.
  [Online]. Available: \url{http://www.jstor.org/stable/2271658}
\BIBentrySTDinterwordspacing

\bibitem{StatmanARTICLE1985}
\BIBentryALTinterwordspacing
R.~Statman, ``Logical relations and the typed $\lambda$-calculus,''
  \emph{Information and Control}, vol.~65, no. 2/3, pp. 85--97, 1985. [Online].
  Available: \url{https://doi.org/10.1016/S0019-9958(85)80001-2}
\BIBentrySTDinterwordspacing

\bibitem{PerezESOP2012}
\BIBentryALTinterwordspacing
J.~A. P{\'{e}}rez, L.~Caires, F.~Pfenning, and B.~Toninho, ``Linear logical
  relations for session-based concurrency,'' in \emph{21st European Symposium
  on Programming ({ESOP})}, ser. Lecture Notes in Computer Science, vol.
  7211.\hskip 1em plus 0.5em minus 0.4em\relax Springer, 2012, pp. 539--558.
  [Online]. Available: \url{https://doi.org/10.1007/978-3-642-28869-2\_27}
\BIBentrySTDinterwordspacing

\bibitem{PerezARTICLE2014}
\BIBentryALTinterwordspacing
------, ``Linear logical relations and observational equivalences for
  session-based concurrency,'' \emph{Information and Computation}, vol. 239,
  pp. 254--302, 2014. [Online]. Available:
  \url{https://doi.org/10.1016/j.ic.2014.08.001}
\BIBentrySTDinterwordspacing

\bibitem{DeYoungFSCD2020}
\BIBentryALTinterwordspacing
H.~DeYoung, F.~Pfenning, and K.~Pruiksma, ``Semi-axiomatic sequent calculus,''
  in \emph{5th International Conference on Formal Structures for Computation
  and Deduction ({FSCD})}, ser. LIPIcs, vol. 167.\hskip 1em plus 0.5em minus
  0.4em\relax Schloss Dagstuhl - Leibniz-Zentrum f{\"{u}}r Informatik, 2020,
  pp. 29:1--29:22. [Online]. Available:
  \url{https://doi.org/10.4230/LIPIcs.FSCD.2020.29}
\BIBentrySTDinterwordspacing

\bibitem{CairesESOP2013}
\BIBentryALTinterwordspacing
L.~Caires, J.~A. P{\'{e}}rez, F.~Pfenning, and B.~Toninho, ``Behavioral
  polymorphism and parametricity in session-based communication,'' in
  \emph{22nd European Symposium on Programming ({ESOP})}, 2013, pp. 330--349.
  [Online]. Available: \url{https://doi.org/10.1007/978-3-642-37036-6\_19}
\BIBentrySTDinterwordspacing

\bibitem{DasCSF2021}
A.~Das, S.~Balzer, J.~Hoffmann, F.~Pfenning, and I.~Santurkar, ``Resource-aware
  session types for digital contracts,'' in \emph{34th {IEEE} Computer Security
  Foundations Symposium ({CSF})}.\hskip 1em plus 0.5em minus 0.4em\relax
  {IEEE}, 2021.

\bibitem{PfenningFOSSACS2015}
\BIBentryALTinterwordspacing
F.~Pfenning and D.~Griffith, ``Polarized substructural session types,'' in
  \emph{18th International Conference on Foundations of Software Science and
  Computation Structures ({FoSSaCS})}, ser. Lecture Notes in Computer Science,
  vol. 9034.\hskip 1em plus 0.5em minus 0.4em\relax Springer, 2015, pp. 3--22.
  [Online]. Available: \url{https://doi.org/10.1007/978-3-662-46678-0\_1}
\BIBentrySTDinterwordspacing

\bibitem{Pierre1986}
J.-P. Jouannaud and H.~Kirchner, ``Completion of a set of rules modulo a set of
  equations,'' \emph{SIAM Journal on Computing}, vol.~15, no.~4, pp.
  1155--1194, 1986.

\bibitem{BowmanIAhmedCFP2015}
\BIBentryALTinterwordspacing
W.~J. Bowman and A.~Ahmed, ``Noninterference for free,'' in \emph{20th {ACM}
  {SIGPLAN} International Conference on Functional Programming ({ICFP})}.\hskip
  1em plus 0.5em minus 0.4em\relax {ACM}, 2015, pp. 101--113. [Online].
  Available: \url{https://doi.org/10.1145/2784731.2784733}
\BIBentrySTDinterwordspacing

\bibitem{CastellaniARTICLE2016}
\BIBentryALTinterwordspacing
I.~Castellani, M.~Dezani{-}Ciancaglini, and J.~A. P{\'{e}}rez,
  ``Self-adaptation and secure information flow in multiparty communications,''
  \emph{Formal Aspects of Computing}, vol.~28, no.~4, pp. 669--696, 2016.
  [Online]. Available: \url{https://doi.org/10.1007/s00165-016-0381-3}
\BIBentrySTDinterwordspacing

\bibitem{Ciancaglini2016}
S.~Capecchi, I.~Castellani, and M.~Dezani-Ciancaglini, ``Information flow
  safety in multiparty sessions,'' \emph{Mathematical Structures in Computer
  Science}, vol.~26, no.~8, p. 1352–1394, 2016.

\bibitem{BalzerESOP2019}
\BIBentryALTinterwordspacing
S.~Balzer, B.~Toninho, and F.~Pfenning, ``Manifest deadlock-freedom for shared
  session types,'' in \emph{28th European Symposium on Programming ({ESOP})},
  ser. Lecture Notes in Computer Science, vol. 11423.\hskip 1em plus 0.5em
  minus 0.4em\relax Springer, 2019, pp. 611--639. [Online]. Available:
  \url{https://doi.org/10.1007/978-3-030-17184-1\_22}
\BIBentrySTDinterwordspacing

\bibitem{CairesCONCUR2019}
\BIBentryALTinterwordspacing
L.~Caires, J.~A. P{\'{e}}rez, F.~Pfenning, and B.~Toninho, ``Domain-aware
  session types,'' in \emph{30th International Conference on Concurrency Theory
  ({CONCUR})}, ser. LIPIcs, vol. 140.\hskip 1em plus 0.5em minus 0.4em\relax
  Schloss Dagstuhl - Leibniz-Zentrum f{\"{u}}r Informatik, 2019, pp.
  39:1--39:17. [Online]. Available:
  \url{https://doi.org/10.4230/LIPIcs.CONCUR.2019.39}
\BIBentrySTDinterwordspacing

\bibitem{PittsStarkHOOTS1998}
A.~M. Pitts and I.~Stark, ``Operational reasoning for functions with local
  state,'' \emph{Higher Order Operational Techniques in Semantics (HOOTS)}, pp.
  227--273, 1998.

\bibitem{AhmedPOPL2009}
\BIBentryALTinterwordspacing
A.~Ahmed, D.~Dreyer, and A.~Rossberg, ``State-dependent representation
  independence,'' in \emph{36th {ACM} {SIGPLAN-SIGACT} Symposium on Principles
  of Programming Languages (POPL)}.\hskip 1em plus 0.5em minus 0.4em\relax
  {ACM}, 2009, pp. 340--353. [Online]. Available:
  \url{https://doi.org/10.1145/1480881.1480925}
\BIBentrySTDinterwordspacing

\bibitem{DreyerPOPL2010}
\BIBentryALTinterwordspacing
D.~Dreyer, G.~Neis, and L.~Birkedal, ``The impact of higher-order state and
  control effects on local relational reasoning,'' in \emph{15th {ACM}
  {SIGPLAN} International Conference on Functional Programming ({ICFP})}.\hskip
  1em plus 0.5em minus 0.4em\relax {ACM}, 2010, pp. 143--156. [Online].
  Available: \url{https://doi.org/10.1145/1863543.1863566}
\BIBentrySTDinterwordspacing

\end{thebibliography}
